\documentclass[reqno, 11pt]{amsart}

\usepackage{amsfonts,amssymb,amsbsy,amsmath,amsthm,enumerate}
\usepackage{txfonts}
\usepackage{eucal}
\usepackage{dsfont}
\usepackage{mathrsfs}

\usepackage[small,nohug,
]{diagrams}
\diagramstyle[labelstyle=\scriptstyle]
\newarrow{Corresponds}<--->

\usepackage[dvips]{graphicx}
\usepackage{color}

\newtheorem{theorem}{Theorem}[section]
\newtheorem{proposition}[theorem]{Proposition}
\newtheorem{lemma}[theorem]{Lemma}
\newtheorem{corollary}[theorem]{Corollary}
\newtheorem{assumption}[theorem]{Assumption}
\newtheorem{definition}[theorem]{Definition}
\newtheorem{remark}[theorem]{Remark}

\newtheorem{example}[theorem]{Example}

\numberwithin{equation}{section}

\numberwithin{figure}{section}
\numberwithin{table}{section}

\newcommand\beq{\begin{equation}}
\newcommand{\bea}{\begin{eqnarray}}
\newcommand{\eea}{\end{eqnarray}}
\newcommand{\beas}{\begin{eqnarray*}}
\newcommand{\eeas}{\end{eqnarray*}}
\newcommand{\beql}{\begin{equation} \label}
\newcommand{\eeq}{\end{equation}}

\newcommand{\R}{\mathbb R}
\newcommand{\N}{\mathbb N}
\newcommand{\C}{\mathbb C}                           

\newcommand{\Z}{\mathbb Z}
\newcommand{\T}{\mathbb T}

\newcommand{\s}[1]{\CMcal{#1}}
\newcommand{\f}[1]{\mathcal{#1}}                  
\newcommand{\bb}[1]{\mathscr{#1}}
\newcommand{\rr}[1]{\mathfrak{#1}}
\newcommand{\n}[1]{\mathds {#1}}

\newcommand{\expo}[1]{{\rm e}^{#1}}                 
\newcommand{\dd}{\,{\rm d}}
\newcommand{\ii}{\,{\rm i}\,}
\newcommand{\ncint}{\mathrel{{\ooalign{$\int$\cr\kern+.07em\raise.15ex\hbox{$\pmb{\scriptstyle-}$}\cr}}}}           \newcommand{\ncpartial}{\mathrel{{\ooalign{$\partial$\cr\kern+.29em\raise.79ex\hbox{$\pmb{\scriptstyle-}$}\cr}}}}

\newcommand{\virg}[1]{\lq\lq#1\rq\rq}                
\newcommand{\ie}{{\sl i.\,e.\,}}
\newcommand{\eg}{{\sl e.\,g.\,}}
\newcommand{\cf}{{\sl cf.\,}}

\newcommand{\blu}{\textcolor[rgb]{0.05,0.24,0.57}}

\begin{document}

\title[Classification of \emph{\virg{Real}} Bloch-bundles]{Classification of \emph{\virg{Real}} Bloch-bundles:\\
Topological Quantum Systems of type {\bf AI}}


\author[G. De~Nittis]{Giuseppe De Nittis}
\address[De~Nittis]{Department Mathematik, Universit\"{a}t Erlangen-N\"{u}rnberg,
Germany}
\email{denittis@math.fau.de}

\author[K. Gomi]{Kiyonori Gomi}
\address[Gomi]{Department of Mathematical Sciences, Shinshu University,  Nagano, Japan}
\email{kgomi@math.shinshu-u.ac.jp}




\begin{abstract}
\vspace{-4mm}
We provide a classification of type {\bf AI} topological quantum systems in dimension $d=1,2,3,4$ which is based on the
 \emph{equivariant} homotopy properties of \virg{Real} vector bundles. This allows us to produce a fine classification able to take care also of the non stable regime which is usually not accessible via   $K$-theoretic techniques.
We prove the absence of non-trivial phases for one-band {\bf AI} free or periodic quantum particle systems  in each spatial dimension
by inspecting  the second equivariant cohomology group which classifies  \virg{Real} line bundles. We also show that  the classification of \virg{Real} line bundles suffices for the complete classification of {\bf AI} topological quantum systems  in dimension $d\leqslant 3$.
In dimension $d=4$ the determination of different topological phases (for free or periodic  systems) is fixed by the second \emph{\virg{Real} Chern class} which provides 
an even labeling identifiable with the degree of a suitable map. Finally, we provide explicit  realizations of
non trivial $4$-dimensional free models for each given topological degree.
\end{abstract}


\maketitle

%
%
\vspace{-4mm}
\tableofcontents

\section{Introduction}\label{sect:intro}

The  construction of a classification scheme for \emph{topological insulators} \cite{schnyder-ryu-furusaki-ludwig-08,kitaev-09,ryu-schnyder-furusaki-ludwig-10} has been one of the most important results obtained in the investigation of topological states of matter. In this scheme topological insulators are classified in terms of the dimensionality of the system and the existence of certain discrete symmetries (for a modern review see, \eg, \cite{hasan-kane-10}). The class of systems which are not subject to any symmetry is denoted with the uppercase letter ${\bf A}$
according to the \emph{Altland-Zirnbauer-Cartan 
classification} (AZC) of topological insulators \cite{altland-zirnbauer-97,schnyder-ryu-furusaki-ludwig-08}.  
This class contains
Quantum Hall systems which have been extensively studied in the last decades  \cite{thouless-kohmoto-nightingale-nijs-82,bellissard-elst-schulz-baldes-94}.
A second class, which is of interest for the present work, consists of systems subject to \emph{even time reversal symmetry} (+TR). This class is denoted by 
${\bf AI}$
in the AZC scheme and contains 
 spinless, or integer spin (\ie bosonic), quantum systems which are invariant under time inversion (\cf Section \ref{sect:bloc_bund}). For reasons that will be clarified in the sequel, we will use the specification \virg{Real} for systems in class ${\bf AI}$\footnote{ 
The use of the adjective  \virg{Real} is justified (by topological reasons) only for operator of class ${\bf AI}$. 
 Sometimes the same nomenclature   is used also for systems in class ${\bf AII}$ (\cf Remark \ref{rk:classAII}), namely for systems subjected to an \emph{odd} time reversal symmetry (-TR). However, a careful analysis of the geometric structure that underlies these systems suggests that the term
 \virg{Quaternionic} is certainly more appropriate \cite{vaisman-90, denittis-gomi}.}.

 \medskip

 In the absence of disorder, 
the electronic properties of solids are described by translationally invariant systems of independent particles
with a typical band spectrum pattern deductible from the Bloch-Floquet theory  \cite{kuchment-93}. In particular, topological insulators are \emph{gapped systems}
in the sense that they show  
 a bulk energy gap separating the highest occupied  valence 
bands from the lowest empty conduction bands. The description of the topological phases  for this kind of systems is captured by the \emph{Fermi (spectral) projection} on the occupied states which, according to a well-established procedure (\cf Section \ref{sect:bloc_bund}), can be associated with a vector bundle known as \emph{Bloch-bundle} \cite{panati-07,denittis-lein-11}. This identification,
described in detail in Section \ref{sect:bloc_bund}, 
  allows to treat the problem of the  classification of topological  phases of electronic systems looking at the underlying vector bundle theory (in the proper category). Mathematically, electronic gapped systems can be treated inside the more general framework of \emph{topological quantum systems}:
\begin{definition}[Topological quantum systems: class {\bf A}]\label{def:tsq}
Let $X$ be a topological space which verifies Assumption \ref{ass:1},
$\s{H}$ a separable Hilbert space and $\bb{K}(\s{H})$ the algebra of compact operators on $\s{H}$.
A \emph{topological quantum system} (of class {\bf A}) is a self-adjoint map
\begin{equation}\label{eq:tqsA1}
X\;\ni\;x\; \longmapsto\; H(x)=H(x)^*\;\in\;\bb{K}(\s{H})
\end{equation}
continuous with respect to the norm-topology on $\bb{K}(\s{H})$. Let $\sigma(H(x))=\{\lambda_j(x)\ |\ j\in\s{I}\subseteq\Z\}\subset\R$,  be the sequence of eigenvalues of $H(x)$ ordered  according to  $\ldots\lambda_{-2}(x)\leqslant\lambda_{-1}(x)<0\leqslant\lambda_1(x)\leqslant\lambda_2(x)\leqslant \ldots$. The map $x\mapsto \lambda_j(x)$ (which is continuous by standard perturbative arguments \cite{kato-95}) is called \emph{$j$-th energy band}. 
An \emph{isolated family} of energy bands is any (finite) collection $\{\lambda_{j_1}(\cdot),\ldots,\lambda_{j_m}(\cdot)\}$
of energy bands such that 
\begin{equation}\label{eq:tqsA2}
\min_{x\in X}\ {\rm dist}\left(
\bigcup_{s=1}^m\{\lambda_{j_s}(x)\}\;,\; \bigcup_{j\in\s{I}\setminus\{j_1,\ldots,j_m\}}\{\lambda_{j}(x)\}
\right)\;=\;C_g\;>0.
\end{equation}
Inequality \eqref{eq:tqsA2} is usually called \emph{gap condition} (\cf with equation \eqref{eq:gap}).
\end{definition}

A standard construction, similar to that explained in Section \ref{sect:bloc_bund}, associates to each  topological quantum system of type 
\eqref{eq:tqsA1} with an {isolated family} of $m$ energy bands a complex vector bundle $\bb{E}\to X$ of rank $m$ that we can still call Bloch-bundle\footnote{The setting described in Definition \ref{def:tsq} can be generalized to unbounded operator-valued maps $x\mapsto H(x)$ by requiring 
the continuity  of the \emph{resolvent map} $x\mapsto R_z(x):=\big(H(x)-z\n{1}\big)^{-1}\in\bb{K}(\s{H})$. 
 Another possible generalization is to replace the norm-topology with the open-compact topology as in \cite[Appendix D]{freed-moore-13}. However, these kind of generalizations have no particular consequences  for the purposes of this work.}. 
The first step of this construction is the realization of a continuous map of rank $m$ spectral projections $X\ni x\mapsto P(x)$
associated with the isolated family of energy bands by means of the Riesz-Dunford integral (\cf equation \ref{eq:proj}); the second step
 turns out to be a concrete application of the \emph{Serre-Swan Theorem} \cite[Theorem 2.10]{gracia-varilly-figueroa-01} which relates vector bundles and continuous family of projections.
Since topological quantum systems
 in class {\bf A} lead to complex vector bundles (without any other extra structure) they are topologically classified by the sets ${\rm Vec}_\C^m(X)$ of  equivalence classes (up to vector bundle isomorphisms) of rank $m$ complex vector bundles over the base space $X$.  The theory of classification of complex vector bundles is a classical, well-studied, problem in topology. Under rather general assumptions for the base space $X$, the sets ${\rm Vec}_\C^m\big(X\big)$ can be classified using homotopy theory techniques and, for dimension $d\leqslant 4$ (and for all $m$) a complete description can be done in terms of cohomology groups and Chern classes \cite{peterson-59,bott-tu-82,woodward-82,cadek-vanzura-93}. A summary of these results is given in Sections \ref{sect:class_C-vb}.

\medskip

Topological quantum systems  described by Definition \ref{def:tsq} are ubiquitous in mathematical physics and  are not necessarily related to condensed matter systems (see \eg the rich monograph \cite{bohm-mostafazadeh-koizumi-niu-zwanziger-03}). Usually the set $X$ plays the role of a configuration space for \emph{parameters} which describe the action of external fields on a system governed by a Hamiltonian $H(x)$.
The phenomenology of these systems can be  enriched by the presence of certain symmetries. In this work we are mainly interested to effects induced by symmetries of  bosonic time-reversal type, namely:

\begin{definition}[Topological quantum systems: class {\bf AI}]\label{def:tsqAI}
Let $x\mapsto H(x)$ be a topological quantum system according to Definition \ref{def:tsq}. Assume that $X$ is endowed with an \emph{involution} $\tau:X\to X$ (\cf Section \ref{sect:Rs-vb_inv}) and $\s{H}$ with a \emph{complex conjugation} (\ie an anti-linear involution) $C:\s{H}\to\s{H}$.
We say that $x\mapsto H(x)$ is a topological quantum system of class {\bf AI} if there exists a continuous unitary-valued map $x\mapsto J(x)$ on $X$ such that
\begin{equation}\label{eq:tqsA3}
J(x)\; H(x)\; J(x)^*\;=\;C\;H(\tau(x))\;C\;,\qquad\quad C\;J(\tau(x))\;C\;=\; J(x)^*\;.
\end{equation}
We refer to  \eqref{eq:tqsA3} as an  \emph{even time reversal symmetry} (+TR) (\cf with equation \eqref{intro:1}).
\end{definition}

Similarly to systems of type {\bf A}  also systems in class {\bf AI} leads to complex vector bundles. However, as consequence of the presence of a +TR symmetry, these vector bundles are endowed with an extra structure named \virg{Real} by M.~F. Atiyah in \cite{atiyah-66}. According to its very definition (\cf Section \ref{sect:Rs-vb}) a  \virg{Real} vector bundle, often named \emph{$\rr{R}$-bundle} in this work, is a complex vector bundle over an \emph{involutive space} $(X,\tau)$ endowed with an involutive automorphism of the total space which covers the involution $\tau$ and restricts to an \emph{anti-linear} map between conjugate  fibers.  $\rr{R}$-bundles define a new category of locally trivial fibered objects with their own morphisms. The set of equivalence classes  (in the appropriate category) of rank $m$ $\rr{R}$-bundles over the {involutive space} $(X,\tau)$ is denoted with ${\rm Vec}_{\rr{R}}^m(X,\tau)$. Therefore, the description of  ${\rm Vec}_{\rr{R}}^m(X,\tau)$ provides the classification for type {\bf AI} topological quantum systems.
However, the  classification theory of $\rr{R}$-bundles is not as developed as for the complex case.
Even if a cohomological classification \cite{edelson-71}  and a characteristic class theory \cite{kahn-59,krasnov-92} are available for \virg{Real} vector bundles since long time, usually only $K$-theoretic results has been developed for the classification of topological phases of condensed matter systems
\cite{kitaev-09} (see also \cite{budich-trauzettel-13} for a recent review). $K$-theory is an extremely powerful tool for the study of vector bundles, in general. 
However, its application to a complete classification of topological quantum systems	
results 	
insufficient for two reasons: It provides a classification only for the \emph{stable regime} ignoring the situation of isolated family of $m$ bands
when $m$ is a small number; It does not provides any information about the objects that classify different Bloch-bundles.
The second point is of great importance in our opinion. In fact there is a great advantage in the identification of the parameters which classify different phases of a topological quantum system with cohomology classes. For instance, the discovery of the link between quantum Hall states
and Chern numbers provided by the Kubo formula \cite{thouless-kohmoto-nightingale-nijs-82,bellissard-elst-schulz-baldes-94} has been one of the most impressive and fruitful result in the study of topological insulators.
Finally, let us point out that there are particular situations (\eg presence of extra symmetries as considered in Appendix \ref{sec:inv_symm}) in which the $K$-theoretic description is genuinely different from the vector bundle classification (\cf Remark \ref{rk:redKZ2-VBZ2}).

\medskip

The aim of this work is double: First of all we want
to discuss a classification procedure for \virg{Real} vector bundles 
over a sufficiently general involutive space $(X,\tau)$
based on the analysis of the \emph{equivariant} structure induced by the involution $\tau$; Second of all, we want to apply such a classification scheme to the interesting case of topological insulators, considered as special examples of topological quantum systems. We expect in this case
a classification  finer (at least mathematically) than the usual $K$-theoretic description capable to take care also of possible effects due to  non stable regime.
As regards the first objective, our major result is:
\begin{theorem}[Classification of {\bf AI} topological quantum systems]\label{theoTQS_main1}
Let $(X,\tau)$ be an involutive space such that $X$ has a finite $\Z_2$-CW-complex decomposition of dimension $d$ with fixed cells only in dimension 0. Then, for all $m\in\N$:
\begin{enumerate}
\item[(i)] ${\rm Vec}_{\rr{R}}^m(X,\tau)\;=\;0$\ \ \ if $d=1$;\vspace{1.2mm}
\item[(ii)] ${\rm Vec}_{\rr{R}}^m(X,\tau)\;\simeq\;{\rm Vec}_{\rr{R}}^{[d/2]}(X,\tau)$\ \ \ if $d\geqslant2$  (here $[x]$ denotes the integer part of $x\in\R$);\vspace{1.2mm}
\item[(iii)] ${\rm Vec}_{\rr{R}}^m(X,\tau)\;\simeq\;H^2_{\Z_2}\big(X,\Z(1)\big)$\ \ \   if $1\leqslant d\leqslant3$ and the isomorphism is given by the first \emph{\virg{Real} Chern class} $\tilde{c}_1$.
\end{enumerate}
\end{theorem}
\noindent
The various items in Theorem \ref{theoTQS_main1} are proved separately in the paper. The proof of (i) is essentially the same proof of Proposition \ref{prop_(i)} (see Remark \ref{rk:gen_prof_int}) and it is 
a consequence of the  \emph{homotopy classification} of \virg{Real} vector bundles (\cf Theorem \ref{theo:honotopy_class})
and of the \emph{$\Z_2$-homotopy reduction} (Lemma \ref{lemma:Z2-reduct}). 
  Item (ii) is a direct consequence of  Theorem \ref{theo:stab_ran_R} which fixes the {stable range} condition for \virg{Real} vector bundles.
 The proof of (iii) is a consequence of (ii) together with Theorem \ref{theo:clasR-lin} which establishes the cohomological classification for \virg{Real} line bundles \cite{kahn-59,krasnov-92}.
 For the sake of completeness, let us recall that an involutive space $(X,\tau)$ has the structure of a $\Z_2$-CW-complex if it admits a skeleton decomposition given by gluing cells of different dimensions which carry a $\Z_2$-action. More details about this notion are given in Section \ref{ssec:stab_reng_equi_constr}. The cohomology group that  appears in (iii)  is the \emph{equivariant Borel} cohomolgy of the space $(X,\tau)$ computed with respect to the \emph{local system} of coefficients $\Z(1)$. 
A short reminder of this theory is given in Section \ref{ssec:borel_constr}. Finally, the construction of the \emph{\virg{Real} Chern classes} is explained in Section \ref{sect:eq_chern_class}.

\medskip

In order to relate Theorem \ref{theoTQS_main1} with the theory of condensed matter type electron systems we need to specify a proper base space and a proper (time-reversal) involution $\tau$. Let us introduce the following two definitions:
\begin{definition}[Free charge systems]
\label{def:free_ferm}
A $d$-dimensional gapped system of independent free charged particles is 
(represented by) an element of ${\rm Vec}_\C^m\big(\n{S}^d\big)$ where
$$
\n{S}^{d}\;:=\;\big\{k\in\R^{d+1}\,|\; \|k\|= 1\big\}
$$
is the $d$-dimensional sphere. Free systems with a +TR symmetry are (represented by) elements of ${\rm Vec}_{\rr{R}}^m\big(\n{S}^{d},\tau\big)$ where the map $\tau:\n{S}^{d}\to \n{S}^{d}$, defined by
\begin{equation}
\label{eq:inv_sphe}
\tau(k_0,k_1,\ldots,k_d)\;:=\;(k_0,-k_1,\ldots,-k_d)
\end{equation}
for all $k= (k_0,k_1,\ldots,k_d)\in\n{S}^{d}$, is called \emph{TR-involution}.
\end{definition}
\begin{definition}[Periodic charge systems]
\label{def:per_ferm}
A $d$-dimensional gapped system of independent charged particles in a periodic environment is 
(represented by) an element of ${\rm Vec}_\C^m\big(\n{T}^d\big)$ where
$$
\n{T}^{d}\;:=\;\underbrace{\n{S}^{1}\;\times\;\ldots\;\times\;\n{S}^{1}}_{d\text{\upshape -times}} 
$$
is the $d$-dimensional torus. Periodic systems with a +TR symmetry are (represented by) elements of ${\rm Vec}_{\rr{R}}^m\big(\n{T}^{d},\tau\big)$ where the \emph{TR-involution} $\tau:\n{T}^{d}\to \n{T}^{d}$
is the diagonal map $\tau:=\tau_1\times\ldots\times\tau_1$  defined by   the {TR-involution} $\tau_1:\n{S}^{1}\to \n{S}^{1}$  on the 1-sphere
according to  \eqref{eq:inv_sphe}.
\end{definition}
\noindent
The short notations $\tilde{\n{S}}^{d}\equiv(\n{S}^{d},\tau)$ and $\tilde{\n{T}}^{d}\equiv(\n{T}^{d},\tau)$ for the involutive spaces associated to free and periodic fermionic systems, respectively, will be often used hereafter.
The physical content of the nomenclature introduced in Definition  \ref{def:free_ferm} and Definition \ref{def:per_ferm} will be justified in Section \ref{sect:bloc_bund}, where we provide the link between translation invariant quantum systems (with +TR symmetry) and complex vector bundles (with a \virg{Real} structure). 
Assuming, for a moment, the correctness
of these definitions, we can argue that
 type {\bf A} topological insulators are classified by ${\rm Vec}_\C^m\big(\n{S}^d\big)$ (free case) and ${\rm Vec}_\C^m\big(\n{T}^d\big)$ (periodic case), while type {\bf AI} topological insulators are classified by ${\rm Vec}_{\rr{R}}^m\big(\n{S}^{d},\tau\big)$ (free case) and ${\rm Vec}_{\rr{R}}^m\big(\n{T}^{d},\tau\big)$ (periodic case). However, this last statement deserves a necessary clarification.
 
 \medskip
 
Vector bundles over a space $X$ are \emph{completely} classified in terms of equivalence classes of homotopic maps (this is true both for complex and \virg{Real} vector bundles). $K$-theory \cite{karoubi-97} provides a  \emph{less} fine classification (of cohomological type).
The difference between the $K$-theoretic classification and the homotopic classification consists in the fact that the first is insensitive to the number of isolated energy bands that one considers to define  the associated Bloch-bundle. In fact $K$-theory is based on a weak form of equivalence which declares that two vector bundles are \emph{stable} equivalent if they are isomorphic when completed with a suitable number of topological trivial bands. From a physical point of view the \emph{topological} state of a fermionic system cannot change if one includes in the description an arbitrary number of  chemically \emph{inert} bands. Since {inert} bands are, in particular, topologically trivial, this means that the relevant invariants for a \emph{general}  description of topological insulators are only those predicted from the $K$-theory. On the other side the \emph{spurious} invariants which are detected only by the homotopic classification can still be of some relevance in \emph{particular} situations. First of all there is a mathematical interest related to the fact that the triviality of the Bloch-bundle is a crucial ingredient for a rigorous proof of the \emph{Peierls substitution} and a rigorous derivation of effective tight-binding models
(see \eg \cite{panati-spohn-teufel-03}) ubiquitously used in physics. 
Usually these tight-binding models are obtained from a perturbation of few energy bands (one or two in the interesting cases) and in this situation the triviality of the related Bloch-bundle really depends also on the spurious  invariants. 
In principle the existence  of 
{spurious} invariants  can be of some interest also from a physical point of view: in fact this is related with \emph{particular} behaviors of electronic systems (or photonic crystals \cite{joannopoulos-johnson-winn-meade-08})
when artificially confined in a narrow range of energies (or frequencies).
 
\medskip

  A second important aspect  of the classification of topological phases is the role of  the Brillouin zone which describes the base space of the associated Bloch-Bundle. Since the Brillouin zone for electrons in a periodic background turns out to be a $d$-dimensional torus $\n{T}^d$ one expects to classify topological insulators with (stable) inequivalent classes  of vector bundles over $\n{T}^d$. On the other side there are experimental evidences that not all the topological invariants associated to  the topology of $\n{T}^d$ are 
 easily detectable with experiments. For instance,  in dimension $d=3$  the aspected classification for periodic systems in class {\bf A} is $\Z^3$  (see the third row of Table \ref{tab:01}.1)  and each phase should be described by the three (quantized) spatial components of the Hall conductance \cite{kohmoto-halperin-wu-92,koshino-aoki-halperin-02}. However, the experimental measurement of these quantities turns out to be extremely difficult (although not impossible in principle) since the high instability under weak disorder effects. In the jargon of topological insulators this type of unstable invariants are called \emph{week} and are distinct from the \emph{strong} invariants which are stable with respect to the presence of disorder. From a physical point of view {strong} invariants are associated to topological phases
characterized by  boundary states that avoid the Anderson localization when disorder is present (according to \cite{schnyder-ryu-furusaki-ludwig-08}  this is the  most effective defining property of a topological insulator); from a mathematical point of view {strong} invariants emerge from the classification of vector bundles when the Brillouin torus $\n{T}^d$ (periodic case) is replaced by  the sphere $\n{S}^d$ (free case). According to this nomenclature the first two rows of Table \ref{tab:01}.1 describe {strong} invariants in class  {\bf A} and {\bf AI}, respectively while the third and fourth rows contain the full set of invariants (weak + strong). Although strong invariants are prominent for the detection of stable topological phases this does not imply that  weak  invariants are void of interest.
Mathematically weak invariants contain relevant information about the geometry of the Bloch-bundle 
and are responsible for interesting facts: \eg they provide obstructions to a fast decay of the Wannier functions \cite{panati-07,denittis-lein-11}. Also physically one cannot exclude that an improvement of the measurement techniques will make possible in the next future a detection of these weak unstable effects. For this reason, for instance,  the theoretical investigation around the description of the three dimensional QHE is still active, as shown by a large number of recent papers devoted to this subject (see \eg \cite{bernevig-hughes-raghu-arovas-07}).

\medskip

Summarizing, the situation is as follows: on one side the \emph{standard} (universally accepted in the physics community) classification of topological insulators
is based on the $K$-theory (complex and \virg{Real}) of the sphere \cite{kitaev-09}; on the other side there are {solid} mathematical motivations 
and particular physical situations which make relevant also the study of  {spurious} and weak invariants\footnote{In this respect we point out that a formula for the weak invariants has been proposed already in \cite[eq. (26)]{kitaev-09}.}
 of the Bloch-bundle.  The application of Theorem  \ref{theoTQS_main1} to the case of free and periodic electron systems leads to:

\begin{theorem}[Homotopic classification of {\bf AI} topological insulators]\label{theo:AI_class}
Let $(\n{S}^d,\tau)$ and $(\n{T}^d,\tau)$ be the involutive spaces introduced in Definition \ref{def:free_ferm} and Definition \ref{def:per_ferm}, respectively. Then:
\begin{enumerate}
\item[(i)] ${\rm Vec}_{\rr{R}}^m\big(\n{S}^{1},\tau\big)=0$\ \  for all $m\in\N$;\vspace{1.2mm}
\item[(ii)] ${\rm Vec}_{\rr{R}}^1\big(\n{S}^{d},\tau\big)=0$\ \  and\ \  ${\rm Vec}_{\rr{R}}^1\big(\n{T}^{d},\tau\big)=0$\ \  for all $d\in\N$;\vspace{1.2mm}
\item[(iii)] ${\rm Vec}_{\rr{R}}^m\big(\n{S}^{d},\tau\big)=0$\ \  and\ \  ${\rm Vec}_{\rr{R}}^m\big(\n{T}^{d},\tau\big)=0$\ \  for all $m\in\N$ and $d=2,3$;\vspace{1.2mm}
\item[(iv)] ${\rm Vec}_{\rr{R}}^m\big(\n{S}^{4},\tau\big)=2\Z$\ \  and\ \  ${\rm Vec}_{\rr{R}}^m\big(\n{T}^{4},\tau\big)=2\Z$\ \  for all $m\geqslant 2$.
\end{enumerate}
\end{theorem}
\noindent
Items (i), (ii) and (iii) are direct consequences of Theorem  \ref{theoTQS_main1} and the  explicit computations
$$
H^2_{\Z_2}\big(\tilde{\n{T}}^d,\Z(1)\big)\;\simeq\;H^2_{\Z_2}\big(\tilde{\n{S}}^d,\Z(1)\big)\;=\;0\qquad\quad \forall\ d\in\N
$$
proved in Proposition \ref{prob:R_linT} and Proposition \ref{prob:R_linS}. The case $d=4$ considered in item (iv) is proved in Proposition \ref{prop:S_4} and Proposition \ref{prop:T_4} for the sphere and the torus, respectively.
 \begin{center}
 \begin{table}[h]\label{tab:01}
 \begin{tabular}{|c||c||c|c|c|c||c|}
\hline
VB  & AZC  & $d=1$ & $d=2$&$d=3$&$d=4$&\\
\hline
 \hline
 \rule[-3mm]{0mm}{9mm}
 ${\rm Vec}_{\C}^m(\n{S}^d)$& {\bf A} & $0$ & $\Z$ & $0$ &   \begin{tabular}{ll}
 \blu{$0$}& \blu{($m=1$)}\\
  $\Z$& ($m\geqslant2$)\\
\end{tabular}   &Free systems\\
\cline{1-6}
 \rule[-3mm]{0mm}{9mm}
${\rm Vec}_{\rr{R}}^m(\n{S}^d,\tau)$ & {\bf AI} & $0$ & $0$ &$0$&\begin{tabular}{ll}
 \blu{$0$}& \blu{($m=1$)}\\
  $2\Z$& ($m\geqslant2$)\\
\end{tabular}&(strong invariants)\\
\hline
 \hline
  \rule[-3mm]{0mm}{9mm}
 ${\rm Vec}_{\C}^m(\n{T}^d)$ & {\bf A} & $0$ & $\Z$ & $\Z^3$&    \begin{tabular}{ll}
 \blu{$\Z^6$}& \blu{($m=1$)}\\
  $\Z^7$& ($m\geqslant2$)\\
\end{tabular}   &Periodic systems\\
\cline{1-6}
 \rule[-3mm]{0mm}{9mm}
${\rm Vec}_{\rr{R}}^m(\n{T}^d,\tau)$ & {\bf AI} & 0 &$0$&$0$&
\begin{tabular}{ll}
 \blu{$0$}& \blu{($m=1$)}\\
  $2\Z$& ($m\geqslant2$)\\
\end{tabular} &(strong + weak)\\
\hline
\end{tabular}\vspace{2mm}
 \caption{
 This table summarizes the  content of Theorem \ref{theo:AI_class}.
 The column VB lists the relevant  equivalence classes of vector bundles and the related
 Altland-Zirnbauer-Cartan labels \cite{altland-zirnbauer-97,schnyder-ryu-furusaki-ludwig-08}  are displayed in column AZC.  The classification in blue corresponds to the \emph{non-stable} regime ($d> 2m$)
 which is not covered by the $K$-theoretic classification.
  The involutive spaces $(\n{S}^d,\tau)$ and $(\n{T}^d,\tau)$
 are described in Definition \ref{def:free_ferm} and Definition \ref{def:per_ferm}, respectively. The first two 
 rows describe the so-called \emph{strong} invariants. The \emph{weak} invariants can be read from the table as the difference between the invariants of the torus $\n{T}^d$ and those of  the sphere $\n{S}^d$.}
 \end{table}
 \end{center}
\noindent
A comparison between the first two rows of Table \ref{tab:01}.1 and the standard physics literature (\eg \cite[Table I]{schnyder-ryu-furusaki-ludwig-08})
shows that the homotopic classification of free electron systems completely agrees with the standard $K$-theoretic classification of topological insulators if $d\leqslant3$. In this dimensional range, also the weak invariants are exactly those prescribed by the $K$-theory of the torus (\cf \cite[eq. (26)]{kitaev-09}). In fact the complete equivalence between $K$-theory and homotopic classification up to dimension 3 is a general fact which depends essentially on Theorem \ref{theoTQS_main1} (ii). In $d=4$ there is the first difference between stable and non-stable regime:
Bloch-bundle of rank 1 (one isolated energy band) behave differently with respect to Bloch-bundles of higher rank (family of several intersecting energy bands). Dimension $d=4$ is also interesting since it is the
 only situation in Table  \ref{tab:01}.1 where non trivial topological states may appear in presence of TR-symmetry.
A closer look to Proposition \ref{prop:S_4} and Proposition \ref{prop:T_4} (together with Proposition \ref{prop:zero_odd_chern}) shows that:
\begin{theorem}[Classifying invariant for {\bf AI} topological insulators in $d=4$]\label{theo:AI_class4D}
Necessary conditions for a rank $m\geqslant 2$ complex vector bundle $\bb{E}$ over $\n{S}^4$ or $\n{T}^4$ to admit a \virg{Real} structure compatible with the TR-involution $\tau$ are $c_1(\bb{E})=0$ and $c_2(\bb{E})\in 2\Z$ where $c_j(\bb{E})$
is the $j$-th Chern class of $\bb{E}$.
\end{theorem}
For a complex vector bundle $\bb{E}$ over $X$ the Chern class
$c_j(\bb{E})$ is an element of $H^{2j}(X,\Z)$. If $X$ has dimension four, the map $H^{4}(X,\Z)\ni c_2(\bb{E})\mapsto \Z$ can be defined via the integration of the 4-form (associated with) $c_2(\bb{E})$ over $X$. The resulting integer is usually called second \emph{Chern number} (\cf Remark \ref{rk:deg_eq}). Then, we can rephrase Theorem \ref{theo:AI_class4D} saying that an even 
second  {Chern number} is a necessary condition for 
the existence of a \virg{Real} structure for four-dimensional systems.
On the other side, from a physical point of view, the existence of non-trivial topological quantities related with second  {Chern numbers} suggests the possibility to realize such  topological states in time-dependent three-dimensional systems subjected to electromagnetic interaction. In this situation the 
\emph{isotropic
magneto-electric response} \cite{qi-hughes-zhang-08,essin-moore-vanderbilt-09,hughes-prodan-bernevig-11} seems to be the right
physical phenomenon able to distinguish between different topological phases. Albeit  we do not investigate explicitly the physics of the magneto-electric response, we provide in Section \ref{sect:non-trivial_ex} a 
concrete recipe for  the construction of  all topologically non-trivial 
 systems of type {\bf AI}. This construction is based on a family of standard prototype models, described by equation \eqref{eq:hamilt_Sigma}, which are ubiquitous in the literature about topological insulators (see \eg \cite[eq. 103]{ryu-schnyder-furusaki-ludwig-10}). In Lemma \ref{lem:AppA} we prove that these models are in fact sufficient to realize all inequivalent topological  phases of type {\bf AI} in dimension $d=4$. Moreover, in Remark \ref{rk:deg2} we describe a standard and  extremely effective computational procedure for the determination of the associated  invariants  based on the notion of
Brouwer degree of maps
 \cite[Chapter I, Section 4]{bott-tu-82}.

\medskip

The classification procedure developed in this work is strongly based on the use of the proper characteristic classes able to describe the category of \virg{Real} vector bundles. These classes appear in the literature under different names; for instance they are called \emph{equivariant} Chern classes in
\cite{kahn-59} (where they were introduced for the first time) or \emph{mixed} Chern classes in \cite{krasnov-92}. To avoid confusion, we prefer to use  the name \virg{Real} Chern classes (or $\rr{R}$-Chern classes for short) to denote these objects.
\virg{Real} Chern classes are elements of an equivariant cohomolgy theory known as \emph{Borel cohomolgy}.
As a complementary result of our analysis we prove that the first two \virg{Real} Chern classes classify completely \virg{Real} vector bundles over $(\n{S}^{d},\tau)$ and $(\n{T}^{d},\tau)$
up to dimension $d=4$ (\cf Proposition \ref{prop:c_2-classific}). Moreover, we provide explicit computations for the related {Borel cohomology groups}; We expect that this may have an independent interest in other fields.
Finally, let us point out that the classification given in  Theorem \ref{theo:AI_class} can be, in principle, extended also to dimensions $d>4$. In fact Proposition \ref{prop:eq_clut_constr} (or equivalently Lemma \ref{lemma:doub_coset})
works for all $d\in\N$ as well as the (equivariant)
\emph{collapsing} map 
$\upsilon:{\n{T}}^d\to{\n{S}}^d$ described in Section \ref{sect:4-d}. On the other side, it is also true that
increasing $d$
the computation of the homotopy classes becomes more and more complicated.

\medskip

\medskip

\noindent
{\bf Acknowledgements.} 
GD's research is supported
 by
the Alexander von Humboldt Foundation. KG's research is supported by 
the Grant-in-Aid for Young Scientists (B 23740051), JSPS.
GD wants to thank H. Schulz-Baldes, G. Landi and G. Pezzini for very helpful discussions and C. Villegas Blas for the hospitality at
Istituto de Matem\'{a}ticas at Universidad Nacional Aut\'{o}noma de M\'{e}xico in
 Cuernavaca where this investigation has begun. KG is indebted to K. Shiozaki for many valuable discussions.
\medskip

\section{\virg{Real} Bloch-bundles}
\label{sect:bloc_bund}

\subsection{The free case} 
Let us start with systems of \emph{free charges}. These are described by self-adjoint operators on the Hilbert space $L^2(\R^d,\dd x)\otimes\C^L$ with a typical structure 
\begin{equation}\label{eq:free_sys1}
\hat{H}\;:=\;\sum_{j=1}^NF_j(-\ii\partial_{x_1},\ldots,-\ii\partial_{x_d})\;\otimes\; M_j
\end{equation}
where the  functions $F_j$'s are usually continuous, real valued and bounded and the $M_j$'s self-adjoint elements in the algebra ${\rm Mat}_L(\C)$ (complex square matrices of size $L$).
The number $L\in\N$ takes in account possible extra degrees of freedom like {spin} or isospin. The typical example of a Hamiltonian of type \eqref{eq:free_sys1} is the \emph{Dirac operator}
given by setting $F_j= -\ii\partial_{x_j}$ (although these functions are unbounded!) and choosing the $M_j$'s among the generators of a representation of a Clifford algebra.
The main feature of operators of type \eqref{eq:free_sys1} is the invariance under \emph{continuous} translations. For all $a\in\R^d$ let
$\hat{U}_a$ be the unitary operator defined by $(\hat{U}_a\psi)(\cdot)=\psi(\cdot-a)$ for all $\psi\in L^2(\R^d,\dd x)\otimes\C^L$. One can easily verify that
 $\big[\hat{H};\hat{U}_a\big]=0$ for all $a\in\R^d$.

\medskip

The most important consequence of the translational invariance is that the spectrum of $\hat{H}$ can be computed using the \emph{Fourier transform}
$\f{F}:L^2(\R^d,\dd x)\otimes\C^L\to L^2(\hat{\R}^d,\dd \kappa)\otimes\C^L$ defined (diagonally on the discrete degrees of freedom) by the usual formula
$$
(\f{F}\psi)(\kappa)\;:=\;\frac{1}{(2\pi)^{\frac{d}{2}}}\int_{\R^d}\dd x\;\expo{-\ii \kappa\cdot x}\;\psi(x)\;.
$$
The transformed operator $\f{F}\;\hat{H}\;\f{F}^{-1}$ decomposes in a  family of matrices   according to
\begin{equation}\label{eq:free_sys2}
{H}(\kappa_1,\ldots,\kappa_d)\;:=\;\sum_{j=1}^LF_j(\kappa_1,\ldots,\kappa_d)\;\otimes\; M_j\;.
\end{equation}
The continuity of the $F_j$'s assures the continuity of the map
$\hat{\R}^d\ni \kappa\mapsto H(\kappa)\in {\rm Mat}_L(\C)$.
Let $\epsilon_1(\kappa)\leqslant\ldots\leqslant \epsilon_L(\kappa)$ be the $L$ real eigenvalues of $H(\kappa)$. Due to usual perturbative arguments \cite{kato-95} also the \emph{energy bands} $\kappa\mapsto \epsilon_j(\kappa)$ are continuous. Moreover,
\begin{equation}\label{eq:spec_free1}
\sigma\big(\hat{H}\big)\;=\; \bigcup_{\kappa\in \hat{\R}^d}\;\sigma\big({H}(\kappa)\big)\;=\;\bigcup_{j=1}^L \s{I}_j
\end{equation}
where $\s{I}_j:=\overline{\{\epsilon\in\R\;|\; \exists\; \kappa\in \hat{\R}^d,\ \ \epsilon=\epsilon_j(\kappa)\}}$. If in addition
 $F_j(\kappa)\to c_j$ when $|\kappa|\to\infty$ independently
of the direction (usually one introduces suitable cutoffs) we can extend continuously the family $H(\kappa)$ over the \emph{one-point compactification} $\hat{\R}^d\cup\{\infty\}$ which is topologically identifiable with  $\n{S}^d$ via the \emph{stereographic coordinates}
$(\kappa_1,\ldots,\kappa_d)\mapsto(k_0,k_1,\ldots,k_d)$
\begin{equation}\label{eq:ster_proj}
k_0(\kappa)\;:=\;\frac{\|\kappa\|^2-1}{\|\kappa\|^2+1}\;,\qquad k_j(\kappa)\;:=\;\frac{2\kappa_j}{\|\kappa\|^2+1}\;,\quad j=1,\ldots,d\;.
\end{equation}
This construction, schematically represented in
 Figure \ref{fig:01} is quite common in the theory of topological insulators \cite[Section 2.1]{ryu-schnyder-furusaki-ludwig-10}.
 Let us remark that the addition of the point $\{\infty\}$ does not change the spectrum \eqref{eq:spec_free1}.

\begin{figure}[htbp]\label{fig:01}
\begin{center}
\fbox{
\includegraphics[height=6cm]{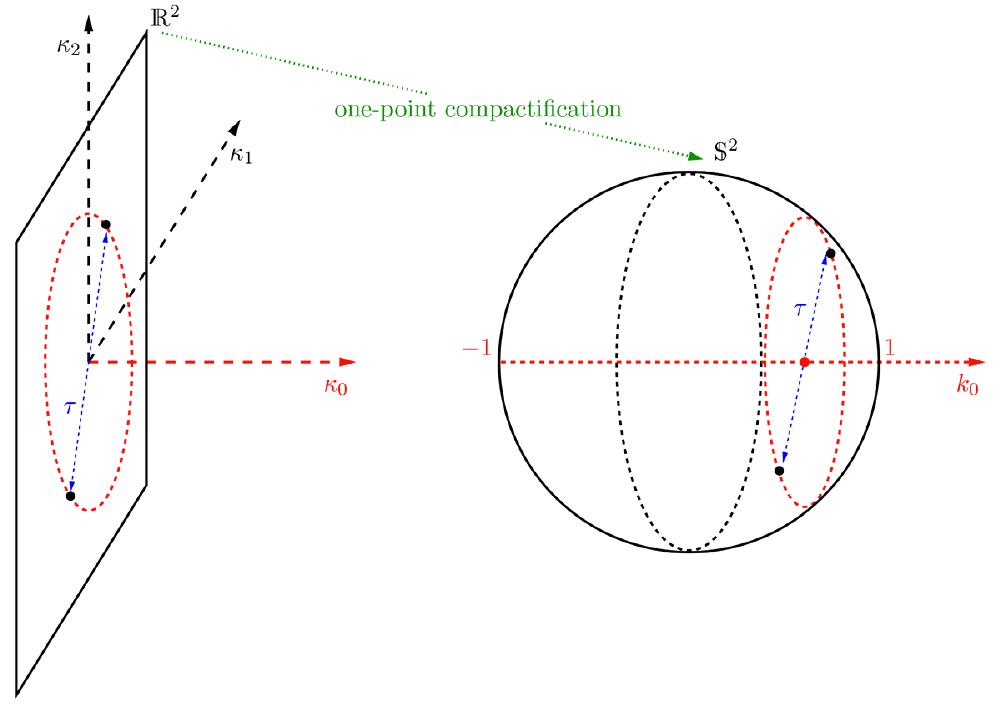}
}
\end{center}
\caption{\footnotesize 
One-point compactification  of the momentum space $\hat{\R}^d$ via stereographic projection and associated TRS-involution.
}
\end{figure}

\medskip

When the spectrum of $\hat{H}$  has a non-trivial \emph{isolated} spectral region $S\subset\sigma\big(\hat{H}\big)$ which verifies  the \emph{gap condition}
\begin{equation}\label{eq:gap}
{\rm dist}\big(S,\sigma\big(\hat{H}\big)\setminus S\big)\;=\; g>0\;,
\end{equation}
the associated spectral projection $\hat{P}_S$ decomposes by means of  $\f{F}$ in a continuous family of projections over the extended space $\n{S}^d$. Each element $P_S(k)$ of the family is a projection (of size $L$) which can be represented as a Riesz-Dunford integral
\begin{equation}\label{eq:proj}
P_S(k)\;:=\;\frac{\ii}{2\pi}\int_{\s{C}}\dd z\;\Big(H(k)-z\;\n{1}_{L}\Big)^{-1}\;.
\end{equation}
Here $\s{C}\subset\C$ is any regular closed path which encloses, without touching, the spectral region $S$. 
Due to the continuity of  $\n{S}^d\ni k\mapsto P_S(k)$, the dimension of  ${\rm Ran}\; {P_S}(k):=\{{\rm v}\in\C^L\ |\ {P_S}(k){\rm v}={\rm v}\}$ is constant over $\n{S}^d$ and the disjoint union
\begin{equation}\label{eq:tot_spac}
\bb{E}_S\;:=\; \bigsqcup_{k\in\n{S}^{d}}\; {\rm Ran}\; {P_S}(k)
\end{equation}
has, in a natural way, the structure of a fibered space $\pi:\bb{E}_S\to \n{S}^{d}$ with \emph{fibers} ${\rm Ran}\; {P_S}(k)$.
The total space $\bb{E}_S$ can be endowed with the topology generated by  \emph{tubular neighborhoods} (\cf \cite[eq. 4.1]{denittis-lein-11}) which makes  $\pi$ continuous and open. Moreover, a set of local trivializations can be constructed using 
the continuity of $P_S(k)$ and the \emph{Nagy formula} (\cf \cite[Lemma 4.5]{denittis-lein-11}). Summarizing, the collection of vector spaces \eqref{eq:tot_spac} defines the total space of a complex vector bundle
$\pi:\bb{E}_S\to \n{S}^{d}$ which is usually called \emph{Bloch-bundle}.
The rank of this vector bundle is fixed by 
 $m={\rm Tr}_{\C^L}{P_S}(k)$.

\medskip

The above construction shows that each isolated spectral region $S$ of a Hamiltonian   \eqref{eq:free_sys1}
defines, up to {isomorphisms},  an element of ${\rm Vec}_\C^m(\n{S}^d)$. 
We recall that a {morphism} (resp. {isomorphism})
 of complex vector bundles over the base space $\n{S}^d$ is a continuous map (resp. homeomorphism) between total spaces which 
 commutes with the bundle projections  and restricts to a linear map on each fiber \cite{atiyah-67,husemoller-94}.

\medskip

Let us consider now the implications of an \emph{anti-}unitary  symmetry of type $\hat{\Theta}:=\hat{C}\hat{J}$ where 
$\hat{C}$ is the  {complex-conjugation}   $(\hat{C}\psi)(\cdot)=\overline{\psi}(\cdot)$  and $\hat{J}$ is a {unitary operator} with the property 
\beql{intro:1}
\left\{
\begin{aligned}
\hat{C}\hat{J}\hat{C}&=\;\hat{J}^\ast\\
\hat{J}\;\hat{H}\;\hat{J}^\ast\;&=\;\hat{C}\hat{H}\hat{C}\;
\end{aligned}
\right.\qquad\quad (\text{{\bf AI} - symmetry})\;.
\eeq
Equation \eqref{intro:1} is equivalent to $\hat{\Theta}\; \hat{H}\; \hat{\Theta}^\ast=\hat{H}$ and the first in  
\eqref{intro:1} implies that
$\hat{\Theta}$ is an \emph{involution} in the sense that 
$\hat{\Theta}^2=\n{1}$ (or equivalently $\hat{\Theta}=\hat{\Theta}^\ast$). 
From a physical point of view  operators $\hat{H}$ with the  symmetry  \eqref{intro:1} are called type
 {\bf AI} topological insulators  \cite{altland-zirnbauer-97,schnyder-ryu-furusaki-ludwig-08}.
This is the class  which contains spinless or integer spin systems with a +TR symmetry. For instance,
if $\hat{J}=\n{1}$ and $L=1$ equation \eqref{intro:1} describes the usual time-reversal symmetry for a spinless quantum particle.

\begin{remark}[Topological insulators in class {\bf AII}]\label{rk:classAII}{\upshape
If in  \eqref{intro:1} one replaces the first condition with $\hat{C}\hat{J}\hat{C}=-\hat{J}^\ast$, or equivalently 
$\hat{\Theta}^2=-\n{1}$, one obtains a different class of topological insulators denoted by {\bf AII} (\cf \cite{altland-zirnbauer-97,schnyder-ryu-furusaki-ludwig-08}). An operator  of this type  has an \emph{odd} {time-reversal symmetry} (-TR). This is the typical case  for quantum system  with half-integer spin.
}\hfill $\blacktriangleleft$
\end{remark}

Let us consider for simplicity the case of $\hat{J}=\n{1}\otimes T$ where $T\in{\rm Mat}_L(\C)$ is a unitary matrix such that $\overline{T}=T^*$. Since $\hat{J}$ commutes with the translations $\hat{U}_a$ the 
 anti-unitary $\hat{\Theta}$ factorizes through the Fourier transform in the sense that
 $$
 (\f{F}\hat{\Theta}\psi)(\kappa)\;=\; T^*\overline{\psi}(-\kappa)\;=:\;\Theta \big(\psi(\kappa)\big)
 $$
 where the last equality  defines the map $\Theta$. A look to equations \eqref{eq:ster_proj} shows that the involution $\kappa\to-\kappa$ on the momentum space $ \hat{\R}^d$ is mapped
into the involution $\tau:\n{S}^d\to\n{S}^d$ 
described in Definition \ref{def:free_ferm}. Moreover, with the help of the Riesz-Dunford formula \eqref{eq:proj} one verifies 
 \beql{intro:2}
\Theta\;P_S(k)\; \Theta\; =\; P_S\big(\tau(k)\big)\qquad\quad \forall\ \ k\in\n{S}^d\;,
\eeq
namely $\Theta$ acts as a morphism  $\Theta:\bb{E}_S\to\bb{E}_S$  which intertwines \emph{anti}-linearly   fibers over $k$ and $\tau(k)$ in  such a way that $\Theta^2={\rm Id}$. This extra structure induced by the {symmetry}  
 $\hat{\Theta}$ makes the Bloch-bundle $\bb{E}_S$ a \virg{Real} vector bundle in the sense of  Atiyah \cite{atiyah-66}. We will describe these objects more in details in Section \ref{sect:Rs-vb}.

\subsection{The periodic case}
 A \emph{periodic} system of {independent charges} is described by a self-adjoint Hamiltonian on $L^2(\R^d)\otimes\C^L$ (continuous case) or $\ell^2(\Z^d)\otimes\C^L$ (discrete case) which is \emph{$\Z^d$-periodic} in the sense that $\big[\hat{H};\hat{U}_n\big]=0$ for all $n\in\Z^d$. As for the free case, one exploits this symmetry to decompose $\hat{H}$ into a continuous family of simpler operators. This can be done through the \emph{Bloch-Floquet transform} \cite{kuchment-93} 
which generalizes the discrete Fourier transform. In detail the {Bloch-Floquet transform} of a vector $\psi$ in $L^2(\R^d)\otimes\C^L$  or $\ell^2(\Z^d)\otimes\C^L$ is given by the formula
\begin{equation}\label{eq:BF-tras}
(\f{F}\psi)(y;\kappa)\;:=\;\sum_{n\in\Z^d}\expo{-\ii \; \kappa\cdot n}\;\psi(y+n)
\end{equation}
which is initially defined on a dense set of  fast-decaying vectors and then extended by continuity to a unitary map between
$L^2(\R^d)\otimes\C^L$  or $\ell^2(\Z^d)\otimes\C^L$ and the direct integral $\int_{\n{B}}^\oplus\dd \kappa\ \s{H}_{\rm f}$.
The space $\n{B}:=\R^d/(2\pi\Z)^d$ is called \emph{Brillouin zone} and the \emph{fiber space} $\s{H}_{\rm f}$ coincide with $L^2([0,1]^d)\otimes\C^L$ in the continuous case and with $\C^L$ in the discrete case. The space $\n{B}$ has the topology of a $d$-dimensional torus $\n{T}^d=\n{S}^1\times\ldots\times \n{S}^1$, as it  is evident from the homeomorphism 
 $\n{B}\ni(\kappa_1,\ldots,\kappa_d)\mapsto(k_1,\ldots,k_d)\in\n{T}^d$ (\cf Figure \ref{fig:02}) given by
\begin{equation}\label{eq:per_sys1}
\begin{aligned}
k_j(\kappa)&\;=\big(\cos\kappa_j,\sin\kappa_j\big)\qquad\quad j=1,\ldots,d\;.
\end{aligned}
\end{equation}
The Bloch-Floquet transform  $\f{F}\;\hat{H}\;\f{F}^{-1}$ of a periodic operator has a fibered structure  $\T^d\ni k\mapsto H(k)\in \bb{B}(\s{H}_{\rm f})$  and, at least  in the cases of physical relevance, the map $k\mapsto H(k)$ is  continuous or even more regular.

\begin{figure}[htbp]\label{fig:02}
\begin{center}
\fbox{
\includegraphics[height=6cm]{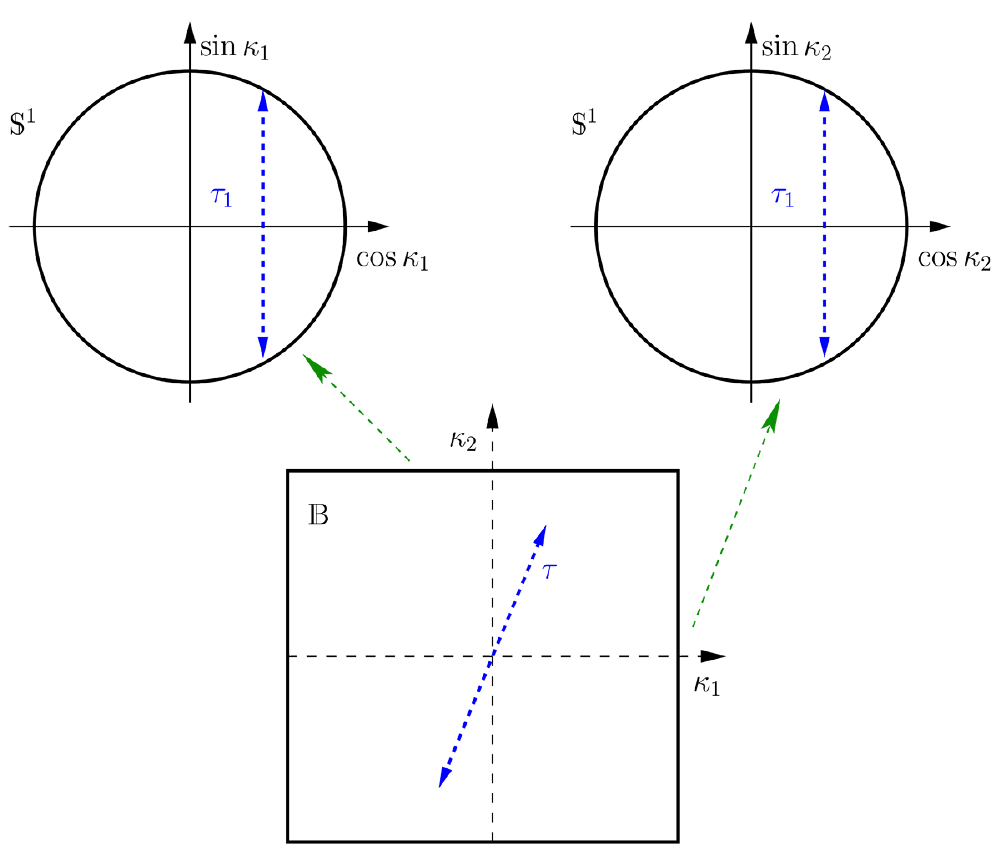}
}
\end{center}
\caption{\footnotesize 
Passage from the Brillouin zone $\n{B}\simeq [-\tfrac{\pi}{2},-\tfrac{\pi}{2}]^d$ to the torus $\n{T}^d=\n{S}^1\times\ldots\times\n{S}^1$ and associated TRS-involution.
}
\end{figure}

\medskip

The spectrum of $\hat{H}$ can be reconstructed from the spectra of the fiber operators $H(k)$ according to the usual relation
\begin{equation}\label{eq:spec_per1}
\sigma\big(\hat{H}\big)\;=\; \bigcup_{k\in {\T}^d}\;\sigma\big({H}(k)\big)\;.
\end{equation}
If there exists
 a non-trivial {isolated} spectral region $S\subset\sigma\big(\hat{H}\big)$ in the sense of the
{gap condition} \eqref{eq:gap}  one can associate to $S$ the continuous family of projections $k\mapsto P_S(k)$ as in equation
\eqref{eq:proj} and the related Bloch-bundle with a total space $\bb{E}_S$ defined as  in \eqref{eq:tot_spac}. Replacing $\n{S}^{d}$ with $\n{T}^{d}$ this argument explains the relation between gapped periodic systems  and  elements of ${\rm Vec}_\C^m(\n{T}^d)$. 

\medskip

Also for periodic systems the presence of a +TR  symmetry of type \eqref{intro:1} endows the Bloch-bundle with a \virg{Real} structure.  We need only  to assume that $\hat{J}$ {commutes} with the translations $\hat{U}_n$ for all $n\in\Z^d$ (\eg this is true if $\hat{J}=\n{1}$) in order to apply consistently the Bloch-Floquet transform. If this is the case from equation \eqref{eq:BF-tras} one obtains
$$
 (\f{F}\hat{\Theta}\psi)(\cdot;\kappa)\;=\; \hat{J}^*\sum_{n\in\Z^d}\expo{-\ii\;\kappa\cdot n}\;\hat{U}_n^\ast\overline{\psi}(\cdot)\;=:\;\Theta \big(\psi(\cdot,-\kappa)\big)\;.
 $$
Equation \eqref{eq:per_sys1} shows that the involution $\kappa\to-\kappa$ on the Brillouin zone $\n{B}$ 
is equivalent to the  involution $\tau:\n{T}^d\to\n{T}^d$ described in Definition \ref{def:per_ferm}. Moreover, the last equation
translates in terms of projections in a relation similar to \eqref{intro:2} which establishes the \virg{Real} structure for the Bloch-bundle $\bb{E}_S$ induced by the symmetry $\hat{\Theta}$.

\section{Classification of complex vector bundles}
\label{sect:class_C-vb}

We assume familiarity of the reader with the basics of the vector bundle theory
and we refer to  classical monographs \cite{atiyah-66,milnor-stasheff-74,husemoller-94} for an extended presentation of this subject.
Hereafter we tacitly assume that:
 \begin{assumption}\label{ass:1}
$X$ is a compact and path connected Hausdorff space with a CW-complex structure.
\end{assumption}
This assumption covers our cases of interest $X=\n{S}^d$ and $X=\T^d$.

\subsection{Homotopy classification for complex vector bundles.}
\label{sec:hom_vecc_compx}
Complex vector bundles can be obtained, up to isomorphisms, as pullbacks of a \emph{universal} classifying vector bundle (the same is true also for real vector bundles, see Remark \ref{rk:iso_categ}).
A model for the universal vector bundle is provided by the \emph{tautological} vector bundle over the \emph{Grassmann manifold}  
$$
G_m(\C^\infty)\;:=\;\bigcup_{n=m}^{\infty}\;G_m(\C^n)\;,
$$
where,
for each pair $m\leqslant n$, $G_m(\C^n)\simeq\n{U}(n)/\big(\n{U}(m)\times \n{U}(n-m)\big)$ is the set of $m$-dimensional (complex) subspaces of $\C^n$. 
Any $G_m(\C^n)$ can be endowed with the structure of a finite CW-complex, making it into
a closed (\ie compact without boundary) manifold  of (real) dimension $2m(n- m)$.
 The inclusions $\C^n \subset \C^{n+1} \subset\ldots$  yield inclusions $G_m(\C^n)\subset G_m(\C^{n+1})\subset\ldots$ and one can equip $G_m(\C^\infty)$
with the direct limit topology. The resulting space has the structure of an infinite CW-complex which is, in particular, paracompact and path-connected.

\medskip

Each manifold $G_m(\C^n)$ is the base space of a {canonical} rank $m$ complex vector bundle $\pi:\bb{T}_m^n\to G_m(\C^n)$ where the total space $\bb{T}_m^n$ consists of all pairs $(\Sigma,{\rm v})$ with $\Sigma\in G_m(\C^n)$  and ${\rm v}$ any vector in $\Sigma$ and the projection is defined by
 $\pi(\Sigma,{\rm v})=\Sigma$. Now, when  $n$ tends to infinity, the same construction leads to 
 the \emph{tautological} $m$-plane bundle $\pi:\bb{T}_m^\infty\to G_m(\C^\infty)$. 
This vector bundle is the universal object which classifies  complex vector bundles in the sense that any rank $m$ complex vector bundle $\bb{E}\to X$ can be realized, up to isomorphisms, as the pullback of  $\bb{T}_m^\infty$ with respect to a \emph{classifying map} $\varphi : X \to G_m(\C^\infty)$, that is $\bb{E}\simeq \varphi^\ast \bb{T}_m^\infty$.
Since pullbacks of homotopic maps yield isomorphic vector bundles (\emph{homotopy property}), the classification of  $\bb{E}$ only depends on the homotopy class  of $\varphi$.  This leads to the fundamental result 
\beql{eq:class_compl_VB1}
{\rm Vec}^m_\C(X)\; \simeq\;[X , G_m(\C^\infty)]
\eeq
where in the right-hand side  there is the set of the equivalence classes of homotopic maps between $X$ and $G_m(\C^\infty)$.

\medskip

The space $G_m(\C^\infty)$ which classifies rank $m$ complex vector bundles is sometimes denoted by $B\n{U}(m)$. The maps  $G_m(\C^n)\hookrightarrow G_{m+1}(\C^{n+1})$ which consist of adding the subspace generated by the last vector ${\rm e}_{n+1}=(0,\ldots,0,1)$ induce  maps between the Grassmann manifolds $G_m(\C^\infty)\hookrightarrow G_{m+1}(\C^\infty)$. Let $B\n{U}:=\bigcup_{m=0}^\infty G_m(\C^\infty)$ with the inductive limit topology. One can prove that
\beql{eq:class_compl_VB2}
{\rm Vec}_\C(X)\;:=\;\bigcup_{m\in\N}{\rm Vec}^m_\C(X)\; \simeq\;[X , B\n{U}]
\eeq
\cite[Chapter 4, Section 4]{husemoller-94}, \ie
the space  $B\n{U}$ classifies the full category of complex vector bundle of any rank.

\medskip

In the case of spheres equation \eqref{eq:class_compl_VB1} provides
$$
{\rm Vec}^m_\C\big(\n{S}^d\big)\; \simeq\;[\n{S}^d , G_m(\C^\infty)]\;=\; \pi_d\big(G_m(\C^\infty)\big)/\pi_1\big(G_m(\C^\infty)\big)
$$
where $\pi_d(\cdot)$ denotes the $d$-th homotopy group. To complete the classification of complex vector bundles over $\n{S}^d$ it
is enough to use the formula
$$
\pi_d\big(G_m(\C^\infty)\big)\;=\;\pi_{d-1}\big(\n{U}(m)\big)
$$
and  the homotopy of the unitary groups (\cf Table \ref{tab:2}). This provides the classification showed in Table 1.\ref{tab:01}
for dimensions $d=1,2,3,4$.

 \begin{table}[htp]
 \label{tab:2}
 \begin{tabular}{|c||c|c|c|c|c|}
\hline
$\pi_k(\n{U}(m))$   & $k=0$ & $k=1$&$k=2$&$k=3$&$k=4$\\
\hline
 \hline
 \rule[-2mm]{0mm}{6mm}
 $m=1$&   0 & $\Z$ & $\blu{0}$ & \blu{0} &\blu{0}\\
\hline
 \rule[-2mm]{0mm}{6mm}
$m=2$&  0 &  $\Z$ &0& $\Z$  &$\blu{\Z_2}$\\
\hline
 \rule[-2mm]{0mm}{6mm}
$m\geqslant3$ &   0 & $\Z$ &0&$\Z$&$0$\\
\hline
\end{tabular}\vspace{2mm}
 \caption{The homotopy of unitary groups. The groups in the \emph{unstable regime} $k\geqslant 2m$ are listed in blue (\cf \cite{kervaire-60}).}
 \end{table}

\subsection{Picard group and line bundles}
\label{ssec:comp-Picard}
The computation of the homotopy classes needed for the
classification of complex vector bundles according to \eqref{eq:class_compl_VB1} is usually non accessible, as with most things in homotopy theory. For this reason, one tries to reduce the problem of the classification to a cohomolgy theory which is more easily accessible. 

\medskip  

We start with the case of line bundles, \ie vector bundles of rank $m=1$.
The space ${\rm Vec}^1_\C(X)$  carries a natural structure of abelian group given by the tensor product. This is know as \emph{Picard} group (in algebraic geometry) and it is usually denoted by ${\rm Pic}(X)$. More precisely,  ${\rm Pic}(X)\equiv \check{H}^1(X,\bb{O}_\C^*)$ is the first  \emph{\v{C}ech} cohomology group of $X$ \cite{griffiths-harris-78,bredon-97}. Here $\bb{O}_\C$ is the \emph{sheaf of germs}
of continuous complex-valued functions on $X$
and $\bb{O}_\C^\ast\subset \bb{O}_\C$ is the subsheaf of germs
of invertible functions.
The identification  between
${\rm Vec}^1_\C(X)$ and $\check{H}^1(X,\bb{O}_\C^*)$ requires an interpretation of  line bundles  in the jargon of sheaves. 
Given an open cover $\{U_\alpha\}$ for  $X$, a complex  line bundle $\pi:\bb{L}\to X$ is completely specified by a family of \emph{transition  functions} $g_{\alpha\beta}:U_\alpha\cap U_\beta\to \n{U}(1)$ such that: i) $g_{\alpha\alpha}=1$; ii) $g_{\alpha\beta}\cdot g_{\beta\gamma}\cdot g_{\gamma\alpha}=1$ (cocycle condition). Moreover two systems of {transition  functions} $\{g_{\alpha\beta}\}$ and $\{g'_{\alpha\beta}\}$ give rise to
isomorphic line bundles if and only if $g'_{\alpha\beta}=\lambda^{-1}_\alpha \cdot g_{\alpha\beta}\cdot \lambda_\beta$ (coboundary condition) for some system $\lambda_\alpha:U_\alpha\to\C^\times$ of invertible functions, where  $\C^\times:=\C\setminus\{0\}$ according to a standard notation. In therms of {\v{C}ech cohomology}
a collection of {transition  functions} $\{g_{\alpha\beta}\}$ is a \v{C}ech 1-cocycle and isomorphic line bundles are identified up to  \v{C}ech 1-coboundary. This proves that ${\rm Vec}^1_\C(X)\simeq \check{H}^1(X,\bb{O}_\C^*)$.
When $X$ is a compact space the exponential map ${\rm exp}:\bb{O}_\C\to\bb{O}_\C^\ast$ induces an
exact sequence $0\to\underline{\Z}\to \bb{O}_\C\to\bb{O}_\C^\ast\to1$ (here $\underline{\Z}$ is the constant sheaf with stalks $\Z$). The related boundary map in cohomology
provides the isomorphism $\check{H}^1(X,\bb{O}_\C^*)\simeq H^2(X,\Z)$ and the resulting classifying isomorphism
\begin{equation}\label{eq:clasLin_bun}
c_1\;:\; {\rm Vec}^1_\C(X)\;\longrightarrow\;H^2(X,\Z)
\end{equation}
is called \emph{(first) Chern class}.

\medskip

The cohomology ring of spheres is well known \cite[Example 3.16]{hatcher-02}:
\begin{equation}\label{eq:cohom_sph}
H^\bullet(\n{S}^d,\s{R})\;\simeq\;{\s{R}}[\alpha]/(\alpha^2)
\end{equation}
where $\s{R}$ is a commutative ring, $\alpha$ is a generator in 
$H^d(\n{S}^d,\s{R})$ and ${\s{R}}[\alpha]$ is the polynomial ring with coefficient in $\s{R}$. Since the only non-trivial groups are $H^0(\n{S}^d,\s{R})\simeq H^d(\n{S}^d,\s{R})\simeq\s{R}$ one concludes 
from \eqref{eq:clasLin_bun} that
\begin{equation}\label{eq:comp_lin1}
{\rm Vec}^1_\C(\n{S}^d)\;=\;
\left\{
\begin{aligned}
&\Z&&\text{if}\ \ d=2\\
&0&&\text{if}\ \ d\neq2
\end{aligned}
\right.
\end{equation}
For tori the cohomology ring is  \cite[Example 3.16]{hatcher-02}
\begin{equation}\label{eq:cohom_torus}
H^\bullet(\n{T}^d,\s{R})\;\simeq\;\Lambda_{\s{R}}[\alpha_1,\ldots,\alpha_d]
\end{equation}
where the right-hand side denotes the exterior algebra over the ring $\s{R}$ with basis the degree one elements $\alpha_j$ such that $\alpha_j^2=0$ and $\alpha_i\alpha_j=-\alpha_j\alpha_i$ if $i\neq j$. Since $H^k(\n{T}^d,\s{R})\simeq \s{R}^{\frac{d!}{k!(d-k)!}}$ one deduces from \eqref{eq:clasLin_bun} that
\begin{equation}\label{eq:comp_lin2}
{\rm Vec}^1_\C(\n{T}^d)\;=\;\Z^{\frac{d(d-1)}{2}}\;.
\end{equation}
Equations \eqref{eq:comp_lin1} and \eqref{eq:comp_lin2} provide the classification for $m=1$ showed in Table 1.\ref{tab:01}.

\subsection{Chern classes and cohomology classification for complex vector bundles}
\label{ssec:chern_classes}
For rank $m>1$ the classification of complex vector bundles requires a generalization of the map \eqref{eq:clasLin_bun}
to higher cohomology groups. 
This is exactly what \emph{Chern classes} do.

\medskip

An important result in the theory of vector bundles is the computation of 
the cohomology ring of the {Grassmann manifold}  \cite[Theorem 14.5]{milnor-stasheff-74}:
\begin{equation}\label{eq:univ_chern_class}
H^\bullet\big(G_m(\C^\infty),\Z\big)\;\simeq\;\Z[\rr{c}_1,\ldots,\rr{c}_m]
\end{equation}
is the ring of polynomials with integer coefficients and $m$ generators $\rr{c}_j\in H^{2j}\big(G_m(\C^\infty),\Z\big)$. 
These generators $\rr{c}_j$ are called \emph{universal} Chern classes and there are no polynomial relationships between them.
The Chern classes of a general vector bundle $\bb{E}\in {\rm Vec}^m_\C(X)$ are constructed as follows: let $\varphi\in[X,G_m(\C^\infty)]$ be the map which represents $\bb{E}$ according to \eqref{eq:class_compl_VB1}, then 
$\varphi^\ast:H^j(G_m(\C^\infty),\Z)\to H^j(X,\Z)$ is a homomorphism of cohomology groups for all $j$.
The \emph{$j$-th} Chern class of $\bb{E}$ is by definition
$$
c_j(\bb{E})\;:=\;\varphi^\ast(\rr{c}_j)\;\in\;H^{2j}(X,\Z)\qquad\quad j=1,2,3,\ldots\;.
$$
Since the homomorphism $\varphi^\ast$ only depends on the homotopy class of $\varphi$,  isomorphic vector bundles possess
 the same family of Chern classes. Because \emph{trivial} vector bundles (vector bundles isomorphic to  $X\times\C^m$) are pullbacks of constant maps, the associated cohomology homomorphisms are necessarily  trivial. Hence, Chern classes of trivial vector bundles are trivial. In this sense, Chern classes can be used as a measure of the \emph{non-triviality} of a vector bundle. An important property of these objects is that $c_j(\bb{E})=0$ for all $j>m$ if $m$ is the rank of $\bb{E}$.

\medskip 

Chern classes are usually a set of incomplete invariants for the classification of vector bundles unless certain additional assumptions on the base space $X$ and on the rank of the fibers are introduced. First of all, let us recall a fundamental result: let $X$ be a CW-complex of dimension $d$
and $\bb{E}\to X$  a complex vector bundle of rank $m$ such that $d< 2m$. Under these assumptions  there exists an isomorphism \cite[Theorem 1.2, Chapter 9]{husemoller-94}
\beql{eq:stab_rank}
\bb{E}\;\simeq\;\bb{E}_0\;\oplus\;(X\times\C^{m-\sigma})
\eeq
where $\bb{E}_0\to X$ is a complex vector bundle of rank $\sigma:=[d/2]$ (here $[x]$ denotes the integer part of $x\in\R$). The decomposition \eqref{eq:stab_rank} states that the topology of vector bundles with fibers of rank sufficiently bigger compared with the dimension of the base space is entirely contained in the non trivial summand  $\bb{E}_0$ which has a \virg{minimal rank} $\sigma$. Accordingly,  a rank $m$ vector bundle $\bb{E}\to X$ has \emph{stable rank} if $m\geqslant\sigma$.  Decomposition \eqref{eq:stab_rank} is extremely useful in the context of complex vector bundles and we will show that a similar decomposition holds true also in the context  of \virg{Real} vector bundles (\cf Theorem \ref{theo:stab_ran_R}).
 
\medskip

In the case of low spatial dimension $d=1,2,3,4$ the {stable rank} condition is certainly verified if $m\geqslant 2$. Moreover a
classical results  \cite{peterson-59, cadek-vanzura-93} states that if $X$ is a CW-complex of dimension $d\leqslant 4$ with no torsion in the even cohomology groups then there is an isomorphism
\begin{equation}\label{eq:clasVect_bun}
(c_1,c_2)\;:\; {\rm Vec}^m_\C(X)\;\longrightarrow\;H^2(X,\Z)\;\oplus\;H^4(X,\Z)\;,\qquad\quad m\geqslant 2
\end{equation}
induced by the \emph{first} and \emph{second}  Chern class.
If one applies \eqref{eq:clasVect_bun} to the special cases $X=\n{T}^d$ or $X=\n{S}^d$ one recovers the classification displayed in Table 1.\ref{tab:01}.

\section{\virg{Real} vector bundles}
\label{sect:Rs-vb}

This Section is devoted to the description of the category of {\virg{Real}} vector bundles
 introduced by  M.~F. Atiyah  in the seminal paper \cite{atiyah-66}.  
Through the paper we often use the shorter expression \emph{$\rr{R}$-bundle} to denote a {\virg{Real}} vector bundle.

\subsection{Involutive space}\label{sect:Rs-vb_inv}
An \emph{involution} $\tau$ on a topological space $X$ is a homeomorphism of period 2, \ie  $\tau^2={\rm Id}_X$.
An \emph{involutive space} is a pair $(X,\tau)$ given by a topological space and an involution.
An involutive space is a particular case of {$G$-space} with $G=\{\tau,{\rm Id}_X\}\simeq\Z_2$. Involutive spaces are also called \emph{\virg{Real}} spaces. This is in analogy with the case of the  algebraic geometry: if $X$ is the set of complex points of an algebraic variety then there is a natural involution given by the complex conjugation.
In this case the fixed points are  just the real points of the variety. We refer to
\begin{equation}
X^{\tau}\;:=\; \{x\in X\ |\ \tau(x)=x\}
\end{equation}
as the \emph{fixed point} set of $(X,\tau)$.
We now describe  some examples of involutive spaces which will be of some utility in the rest of the paper.

\begin{example}[Antipodal action on spheres]\label{ex:sphere_antipodal}{\upshape
 The \emph{antipodal} involution $\vartheta$ on the $d$-sphere $
\n{S}^{d}=\{k\in\R^{d+1}\,|\; \|k\|= 1\}$ is given by the  reflection around the origin:
 $\vartheta(k)=-k$.  
 The pair $(\n{S}^{d},\vartheta)$
is an involutive space and the orbit space $\n{S}^{d}/\vartheta\simeq \R P^{d}$ coincides with the real projective space, namely the set of   lines through the origin of $\R^{d+1}$. The involutive space $(\n{S}^{d},\vartheta)$ has no fixed points since the action of $\vartheta$  is free.
Let $\mathbb{S}^\infty:=\bigcup_{d=0}^\infty \mathbb{S}^{d}$ and $\R P^\infty:=\bigcup_{d=0}^\infty \R P^{d}$.
Since the antipodal maps $\vartheta$  is  compatible with the inclusions $\mathbb{S}^d \hookrightarrow\mathbb{S}^{d+1}$, there exists an induced involution
$\vartheta:\mathbb{S}^\infty\to \mathbb{S}^\infty$ which is free and has orbit space $\mathbb{S}^\infty/\vartheta=\R P^\infty$. 
 Let us recall that  $\mathbb{S}^\infty$ is a contractible space \cite[Example 1B.3, pg. 88]{hatcher-02}. To simplify the notation, we sometimes use  $\hat{\n{S}}^{d}$ for the pair $(\n{S}^{d},\vartheta)$. 
}\hfill $\blacktriangleleft$
\end{example}

\begin{example}[TR-involution on spheres]\label{ex:homot_spher}{\upshape
According to Definition \ref{def:free_ferm}, the \emph{TR-involution} on $\n{S}^d$ is given  by  $\tau(k_0,k_1,\ldots,k_d)=(k_0,-k_1,\ldots,-k_d)$.
The  space $(\n{S}^d,\tau)$ has  2 fixed  points $(\pm1,0,\ldots,0)$, independently of the dimension. 
It is useful to introduce the short notation $\tilde{\n{S}}^d$ for the pair $(\n{S}^d,\tau)$.
We can identify $\tilde{\n{S}}^d\simeq \s{S}(\hat{\n{S}}^{d-1})$ with  the \emph{suspension} of $\hat{\n{S}}^{d-1}$
along the direction $k_0\in[-1,1]$. Indeed,  $\tau$ acts as the antipodal map on $\n{S}^{d-1}$ preserving the coordinate $k_0$. This implies  that $\tilde{\n{S}}^d/\tau\simeq \s{S}(\R P^{d-1})$ and, in particular, $\tilde{\n{S}}^1/\tau\simeq\{\ast\}$ and $\tilde{\n{S}}^2/\tau\simeq\s{S}(\n{S}^1)\simeq\n{S}^2$. In the jargon of the {orbifold theory} these orbit spaces are examples of \emph{effective global quotients} \cite[Definition 1.7]{adem-leida-ruan-07}.

As suggested in \cite{atiyah-66}, TR-involution on $\n{S}^d$ can be related to the usual complex involution. Let us introduce the space $\R^{p,q}:=\R^p\oplus\,{\rm i}\R^q$ and the related unit sphere  
$\n{S}^{p,q}:=\{k\in \R^{p,q}\ |\ \|k\|=1\}$. We observe that $\R^{p,p}\simeq\C^p$ and $\n{S}^{p,q}$ has (real) dimension $p+q-1$. On $\n{R}^{p,q}$ and $\n{S}^{p,q}$ the complex conjugation acts reversing only the sign of the last $q$ coordinates. With this notation the involutive space $\tilde{\n{S}}^d\equiv(\n{S}^d,\tau)$ is represented by $\n{S}^{1,d}$ endowed with the  complex conjugation. This 
notation is particularly convenient for $K$-theoretic computations (\cf Appendix \ref{sect:KR-theory}).
   }\hfill $\blacktriangleleft$
\end{example}

\begin{example}[TR-involution on tori]\label{ex:tori_inv}{\upshape
The \emph{TR-involution} $\tau:\n{T}^d\to\n{T}^d$ has been defined in Definition \ref{def:per_ferm} as the product of 
TR-involutions on $\n{S}^1$. The pair $(\n{T}^d, \tau)$ is an involutive space
with fixed point set of cardinality $2^d$. Introducing the short notation $\tilde{\n{T}}^d$ for the pair $(\n{T}^d,\tau)$ we have the obvious relation
$\tilde{\n{T}}^d=\tilde{\n{S}}^d\times\ldots\times\tilde{\n{S}}^1$.
The orbit spaces $\tilde{\n{T}}^d/\tau$ are particular examples of {effective global quotients} called
\emph{toroidal orbifolds} \cite{adem-duman-gomez-11}: $\tilde{\n{T}}^2/\tau\simeq\n{S}^2$ is know as \emph{pillow} (see \eg \cite{easther-greene-jackson-02}) while  $\tilde{\n{T}}^4/\tau$  is the famous \emph{Kummer surface} $K3$ (see \eg \cite[Section 3.3]{scorpan-05}).
 }\hfill $\blacktriangleleft$
\end{example}

\begin{example}[Involution on Grassmann manifold]
\label{ex:grassman_inv}
{\upshape
The {complex Grassmann manifold} $G_m(\C^\infty)$ described in Section \ref{sec:hom_vecc_compx}
can be endowed with the structure of an involutive space in a natural way.
 If $\Sigma=\langle{\rm v}_1,\ldots,{\rm v}_m\rangle_\C$ is
the $m$-plane in $G_m(\C^n)$  generated by the basis $\{{\rm v}_1,\ldots,{\rm v}_m\}$ then we define $\overline{\Sigma}\in G_m(\C^n)$ as  $\overline{\Sigma}:=\langle\overline{\rm v}_1,\ldots,\overline{\rm v}_m\rangle_{\C}$ where $\overline{\rm v}_j$ is the complex conjugated of  ${\rm v}_j$. Clearly, the definition of 
$\overline{\Sigma}$ does not depend on the choice of a particular basis. The map $\varrho:G_m(\C^n)\to G_m(\C^n)$ given by $\varrho({\Sigma}):=\overline{\Sigma}$ is an involution which makes the pair $(G_m(\C^n),\varrho)$ an involutive space. Since all the inclusions $G_m(\C^n)\hookrightarrow G_m(\C^{n+1})\hookrightarrow\ldots$ are equivariant, the involution extends to 
 the infinite Grassmann manifold in such a way that $(G_m(\C^\infty),\varrho)$ becomes an involutive space. 

The description of the fixed point set of $(G_m(\C^\infty),\varrho)$
requires the real version of the Grassmann manifold.
  Let $G_m(\R^n)\simeq\n{O}(n)/\big(\n{O}(m)\times \n{O}(n-m)\big)$ be  the set of 
$m$-dimensional real hyperplanes passing through the origin of
$\R^n$. The \emph{real} Grassmann manifold is defined as in the complex case by the inductive limit $G_m(\R^\infty):=\bigcup_{n=m}^{\infty}\;G_m(\R^n)$. We recall that
$G_m(\R^n)$ is a closed manifold  of (real) dimension $m(n- m)$ which is path-connected. 
Let $\Sigma=\langle{\rm v}_1,\ldots,{\rm v}_m\rangle_\C$ be a fixed point of $G_m(\C^\infty)$ under $\varrho$.
This means that $\overline{\rm v}_j\in \Sigma$ for all $j=1,\ldots,m$ and the vectors ${\rm u}_j:={\rm v}_j+\overline{\rm v}_j$
and ${\rm w}_j:=\ii({\rm v}_j-\overline{\rm v}_j)$ 
provide a \emph{real} basis which span $\Sigma$ over $\R$. Let 
$\Sigma\in G_m(\C^\infty)$ be a fixed point, then
$\Sigma=\langle{\rm u}_1,\ldots,{\rm u}_m\rangle_\C$ 
where the ${\rm u}_j$'s are real and  the map $G_{m}(\R^\infty)\to G_{m}(\C^\infty)$ given  by
$\langle{\rm u}_1,\ldots,{\rm u}_m\rangle_\R\to \langle{\rm u}_1,\ldots,{\rm u}_m\rangle_\C$ is an embedding of 
$G_{m}(\R^\infty)$ onto the fixed point set of $G_{m}(\C^\infty)$.
}\hfill $\blacktriangleleft$
\end{example}

\subsection{\virg{Real} structure on vector bundles}
\label{sect:Rs-vb_main}
Following \cite{atiyah-66} we call {\virg{Real}} vector bundle, or $\rr{R}$-bundle, over the involutive space $(X,\tau)$
a complex vector bundle $\pi:\bb{E}\to X$ endowed with a (topological) homeomorphism $\Theta:\bb{E}\to \bb{E}$
 such that:
\begin{itemize}

\item[$(\rr{R}_1)$] the projection $\pi$ is \emph{real} in the sense that $\pi\circ \Theta=\tau\circ \pi$;

\item[$(\rr{R}_2)$] $\Theta$ is \emph{anti-linear} on each fiber, \ie $\Theta(\lambda p)=\overline{\lambda}\ \Theta(p)$ for all $\lambda\in\C$ and $p\in\bb{E}$ where $\overline{\lambda}$ denotes the complex conjugate of $\lambda$;

\item[$(\rr{R}_3)$] $\Theta^2={\rm Id}_{\bb{E}}$.

\end{itemize}
It is always possible to endow
$\bb{E}$ with
a Hermitian metric with respect to which $\Theta$ is an \emph{anti}-unitary map between conjugate fibers
(\cf Remark \ref{rk:eq_metric}).
 We remark that an $\rr{R}$-bundle differs from an \emph{equivariant} vector bundle which is a complex vector bundle in the category of $\Z_2$-spaces (see \cite[Section 1.6]{atiyah-67}). Indeed in the first case the map $\bb{E}_x \to \bb{E}_{\tau(x)}$ induced by $\Theta$ is anti-linear,  while in the second case it is linear.
 
 \medskip
 
A vector bundle \emph{morphism} $f$ between two vector bundles  $\pi:\bb{E}\to X$ and $\pi':\bb{E}'\to X$ 
over the same base space
is a  continuous map $f:\bb{E}\to \bb{E}'$ which is \emph{fiber preserving} in the sense that  $\pi=\pi'\circ f$
 and that restricts to a \emph{linear} map on each fiber $\left.f\right|_x:\bb{E}_x\to \bb{E}'_x$. Complex (resp. real) vector bundles over  $X$ together with vector bundle morphisms define a category and we use ${\rm Vec}^m_\C(X)$
 (resp. ${\rm Vec}^m_\R(X)$) to denote the set of equivalence classes of isomorphic vector bundles of rank $m$.
    Also 
 $\rr{R}$-bundles define a category with respect to  \emph{$\rr{R}$-morphisms}. An $\rr{R}$-morphism $f$ between two  $\rr{R}$-bundles
 $(\bb{E},\Theta)$ and $(\bb{E}',\Theta')$ over the same involute space $(X,\tau)$ 
  is a vector bundle morphism  commuting with the involutions, \ie $f\circ\Theta\;=\;\Theta'\circ f$ \cite[Section 1]{atiyah-66}. The set of equivalence classes of isomorphic $\rr{R}$-bundles of  rank $m$ over $(X,\tau)$ 
  is denoted with ${\rm Vec}_{\rr{R}}^m(X,\tau)$. Since the sum of a product line $\rr{R}$-bundle (see below) defines ${\rm Vec}_{\rr{R}}^m(X,\tau) \to {\rm Vec}_{\rr{R}}^{m+1}(X,\tau)$
  one considers the inductive limit ${\rm Vec}_{\rr{R}}(X,\tau):=\bigcup_{m\in\N}{\rm Vec}_{\rr{R}}^m(X,\tau)$
  when the rank of the fiber is not crucial.

\medskip

The set ${\rm Vec}^m_\C(X)$ is non-empty since it contains at least the
\emph{product} vector bundle $X\times\C^m\to X$ with canonical projection $(x,{\rm v})\mapsto x$. Similarly, in the real case one has that 
$X\times\R^m\to X$ provides  an element of ${\rm Vec}^m_\R(X)$. A complex (resp. real) vector bundle is 
called \emph{trivial}  if it is isomorphic to the  complex (resp. real) product vector bundle. For $\rr{R}$-bundles the situation is similar. First of all the set ${\rm Vec}_{\rr{R}}^m(X,\tau)$ is non-empty since it contains at least the \emph{\virg{Real} product  bundle} 
 $X\times\C^m\to X$ endowed with the \emph{product} $\rr{R}$-structure
$\Theta_0 (x,{\rm v})=(\tau(x),\overline{{\rm v}})$ 
given by the complex conjugation ${\rm v}\mapsto \overline{{\rm v}}$.
We say that an $\rr{R}$-bundle is \emph{$\rr{R}$-trivial}  if it is isomorphic  to the product $\rr{R}$-bundle in the category of $\rr{R}$-bundles.

\medskip

At  fixed point of $(X,\tau)$  the homeomorphism $\Theta$ acts fiberwise as an anti-linear map $\Theta_x:\bb{E}_x\to \bb{E}_x$ such that $\Theta_x^2={\rm Id}_{\bb{E}_x}$. This means that $\bb{E}_x$ is, in a natural way, the \emph{complexification} of the real vector space  $\bb{E}^\Theta_{x}$ spanned by the +1-eigenvectors of $\Theta_x$. 
Moreover, each topological space $X$ can be interpreted as an involutive space with respect to the trivial involution ${\rm Id}_X$. This leads to the following fact:

\begin{proposition}\label{prop:Rv1}
${\rm Vec}^m_{\R}(X)\simeq{\rm Vec}^m_{\rr{R}}(X,{\rm Id}_X)$ for all $m\in\N$.
\end{proposition}
\proof[{proof} (sketch of)]
Let $\rr{c}:{\rm Vec}_{\R}^m(X)\to{\rm Vec}_{\rr{R}}^m(X,{\rm Id}_X)$ be the \emph{complexification} of 
a real vector bundle $\rr{c}:\bb{E}\mapsto\bb{E}\otimes_\R\C$ with  $\C$ endowed with the standard complex conjugation. For the opposite direction let us define the \emph{realification} $\rr{r}:{\rm Vec}^m_{\rr{R}}(X,{\rm Id}_X)\to{\rm Vec}^m_{\R}(X)$ by the map
 $\rr{r}:\bb{E}\mapsto\bb{E}^\Theta$ where $\bb{E}^\Theta$ is the set of fixed points of $\bb{E}$ with respect $\Theta$ which has a real linear structure in each fiber.
The identity $\bb{E}^\Theta\otimes_\R\C\equiv\bb{E}$ proves the equality $\rr{c}\circ\rr{r}={\rm Id}$.
The inverse equality $\rr{r}\circ\rr{c}={\rm Id}$ follows from $(\bb{E}\otimes_\R\C)^\Theta\equiv \bb{E}$.
\qed

\medskip

This argument was initially sketched in  \cite[Section 1]{atiyah-66} and a $K$-theoretic version of it is
proved in \cite[Section 3.6]{luke-mishchenko-98}. 
We will provide a different (but equivalent) justification  in Remark \ref{rk:iso_categ}. Let us point out that
this result justifies  the name \virg{Real} vector bundle (in the category of involutive spaces) for elements in ${\rm Vec}_{\rr{R}}(X,\tau)$.

\medskip

Given a \virg{Real} bundle $(\bb{E},\Theta)$ over the involutive space $(X,\tau)$ we can \virg{forget} the $\rr{R}$-structure
and consider only the complex vector bundle $\bb{E}\to X$. This forgetting procedure goes through equivalence classes.
\begin{proposition}\label{prop:Rv101}
The process of forgetting the \virg{Real} structure defines a map
$$
\jmath\;:\;{\rm Vec}^m_{\rr{R}}(X,\tau)\;\longrightarrow \; {\rm Vec}^m_{\C}(X)
$$
such that $\jmath:0\to0$ where  $0$  denotes the trivial class.
\end{proposition}
\proof
The map $\jmath$ is well defined since  an $\rr{R}$-isomorphism between two $\rr{R}$-bundles is, in particular, an isomorphism of complex vector bundles plus an extra condition of equivariance which is lost under the process of forgetting the \virg{Real} structure. 	
\qed

\medskip

\begin{remark}\label{rk:Jmap}{\upshape
In general the map $\jmath$ is neither injective nor surjective depending on the nature of the involutive space $X$. In the case of TR-involution on $\n{S}^2$ a look at the Table \ref{tab:01}.1 shows that ${\rm Vec}^m_{\rr{R}}(\n{S}^2,\tau)=0$ and ${\rm Vec}^m_{\C}(\n{S}^2)\simeq\Z$ and so the map
$\jmath:{\rm Vec}^m_{\rr{R}}(\n{S}^2,\tau)\to {\rm Vec}^m_{\C}(\n{S}^2)$ cannot be surjective. On the other hand we know that
\beql{eq:comp_VB_S1}
{\rm Vec}^m_{\C}(\n{S}^1)\;=\;0\qquad\quad\forall\ m\in\N\;.
\eeq
Moreover, the study of the Picard group for real vector bundles provides the isomorphism
\begin{equation}\label{eq:clasLin_bunR}
w_1\;:\; {\rm Vec}^1_\R(X)\;\longrightarrow\;H^1(X,\Z_2)
\end{equation}
given by the \emph{Stiefel-Whitney
class}  $w_1$. Using the identification in {Proposition} \ref{prop:Rv1} one has
$$
{\rm Vec}^1_{\rr{R}}(\n{S}^1,{\rm Id}_X)\;\simeq\;{\rm Vec}^1_\R(\n{S}^1)\;\simeq\;H^1(\n{S}^1,\Z_2)\;\simeq\;\Z_2
$$
where the non-trivial element is the \emph{M\"{o}bius bundle}. A comparison with \eqref{eq:comp_VB_S1} shows that the map
$\jmath:{\rm Vec}^1_{\rr{R}}(\n{S}^1,{\rm Id}_X)\to {\rm Vec}^1_{\C}(\n{S}^1)$ cannot be injective in this case.
}\hfill $\blacktriangleleft$
\end{remark}

\subsection{\virg{Real} sections}\label{sect:Rs-vb_sect}
Let $\Gamma(\bb{E})$ be the set of  \emph{sections}
of an $\rr{R}$-bundle $(\bb{E},\Theta)$ over the involute space $(X,\tau)$. We recall that a section $s$
is a continuous map $s:X\to\bb{E}$ such that $\pi\circ s={\rm Id}_X$. The set $\Gamma(\bb{E})$ has the structure of a module over the algebra $C(X)$  \cite[Chapter 3, Proposition 1.6]{husemoller-94} 
and inherits from the \virg{Real} structure of $(\bb{E},\Theta)$  an \emph{anti}-linear involution $\tau_\Theta: \Gamma(\bb{E})\to \Gamma(\bb{E})$ defined by 
$$
\tau_\Theta(s)\;:=\;\Theta\;\circ\; s\;\circ\;\tau\;.
$$
 A {\virg{Real} section}, or
{$\rr{R}$-section}, is an element  $s\in\Gamma(\bb{E})$ such that $\tau_\Theta(s)=s$, namely it is a fixed point of $\tau_\Theta$. For this reason we denote  by $\Gamma(\bb{E})^{\Theta}$ the subset of  $\rr{R}$-sections. For a given $s\in \Gamma(\bb{E})$  the two combinations $s+\tau_\Theta(s)$ and ${\rm i}(s-\tau_\Theta(s))$ are in $\Gamma(\bb{E})^{\Theta}$ and they are independent over $C(X,\R)$ (the algebra of real-valued continuous function on $X$).  
Therefore, the full $C(X)$-module $\Gamma(\bb{E})$ coincides with the complexification of the $C(X,\R)$-module $\Gamma(\bb{E})^{\Theta}$.

\medskip

The product $\rr{R}$-bundle $(X\times \C^m,\Theta_0)$ over   $(X, \tau)$ has a special family of sections $\{r_1,\ldots,r_m\}$ given by $r_j:x\mapsto (x,{\rm e}_j)$ with ${\rm e}_j:=(0,\ldots,0,1,0,\ldots,0)$ the $j$-th vector of the canonical  basis of $\C^m$. Clearly, these are $\rr{R}$-sections since 
$$
\tau_{\Theta_{0}}(r_j)(x)\;=\;\Theta_{0}\big(\tau(x),{\rm e}_j\big)\;=\;(x,\overline{{\rm e}}_j)=\; r_j(x)
$$ 
where we exploited the reality  of the canonical basis ${\rm e}_j=\overline{{\rm e}}_j$. Moreover, 
for all $x\in X$ the vectors $\{r_1(x),\ldots,r_m(x)\}$ are a complete basis for the fiber $\{x\}\times\C^m$ over $x$. 
A set of $\rr{R}$-sections which provides a complete system of vectors in each fiber is called a \emph{global $\rr{R}$-frame}. 
For complex (resp. real) vector bundles the presence of a global frame of complex (resp. real) sections is equivalent to the 
triviality of the vector bundle \cite[Theorem 2.2]{milnor-stasheff-74}. A similar fact holds true also for $\rr{R}$-bundles.

\begin{theorem}[$\rr{R}$-triviality]
\label{theo:triv_glob}
An $\rr{R}$-bundle $(\bb{E},\Theta)$ over $(X,\tau)$ is $\rr{R}$-trivial if and only if it  admits a global $\rr{R}$-frame.
\end{theorem}
\proof
Let us assume that  $(\bb{E},\Theta)$ is $\rr{R}$-trivial. This means that there is an $\rr{R}$-isomorphism  $f:X\times \C^m\to\bb{E}$ between $(\bb{E},\Theta)$ and the product $\rr{R}$-bundle $(X\times \C^m,\Theta_0)$. Let us define sections  $s_j\in \Gamma(\bb{E})$ by $s_j:=f\circ r_j$ where $\{r_1,\ldots,r_m\}$ is the global $\rr{R}$-frame of $(X\times \C^m,\Theta_0)$. The fact that $f$ is an isomorphism implies that $\{{s}_1,\ldots,{s}_m\}$ spans each fiber of $\bb{E}$. Moreover, these are $\rr{R}$-sections since $\Theta\circ f=f\circ \Theta_{0}$ and so
$
\tau_{\Theta}(s_j)=\Theta\circ f\circ r_j\circ \tau=f\circ\Theta_{0}\circ r_j\circ\tau=f\circ r_j=s_j
$.

Conversely, let us assume that $(\bb{E},\Theta)$ has a  global $\rr{R}$-frame $\{{s}_1,\ldots,{s}_m\}$. For each $x\in X$ we can set the linear isomorphism $f_x:\{x\}\times\C^m\to \bb{E}_x$ defined by $f_x(x,{\rm e}_j):=s_j(x)$. The collection of $f_x$ defines an isomorphism $f:X\times\C^m\to \bb{E}$ between complex vector bundles \cite[Theorem 2.2]{milnor-stasheff-74}. Moreover, 
$\Theta\circ f(x,{\rm e}_j)=\Theta\circ s_j(x)=s_j(\tau(x))=f(\tau(x),{\rm e}_j)=f\circ\Theta_{0}(x,{\rm e}_j)$ for all $x\in X$
and all $j=1,\ldots,m$, namely  $f$ is an $\rr{R}$-isomorphism.
\qed

\subsection{Homotopy classification of \virg{Real} vector bundles}\label{sect:Rs-vb_homotopy}
The compactness assumption for the base space $X$ (\cf Assumption \ref{ass:1})   allows us to extend  usual homotopy properties valid  for complex vector bundles  to the category of $\rr{R}$-bundles.
\begin{lemma}[Extension]\label{lemma:ext}
Let $(X,\tau)$ be an involutive space and assume that $X$ verifies Assumption \ref{ass:1}. Let
 $(\bb{E},\Theta)$  be an $\rr{R}$-bundle over $(X,\tau)$ and $Y\subset X$  a closed subset such that $\tau(Y)=Y$.
 Then each  $\rr{R}$-section $\tilde{s}:Y\to \left.\bb{E}\right|_Y$ extends to an $\rr{R}$-section $s:X\to\bb{E}$.
\end{lemma}
\proof
The proof follows from the analogous statement for complex vector bundles  \cite[Lemma 1.1]{atiyah-bott-64} or \cite[Chapter 2, Theorem 7.1]{husemoller-94} which assures that $\tilde{s}$ can be extended to a section $t:X\to\bb{E}$. Now, in order to get an $\rr{R}$-section it is enough to consider $s:=\frac{1}{2}(t+\tau_\Theta(t))$. Evidently   $\left.s\right|_Y=\tilde{s}$.
\qed

\medskip

An $\rr{R}$-bundle is locally trivial in the category  of complex vector bundles by definition. Less obvious is that  an $\rr{R}$-bundle is also locally trivial in the category of vector bundles over an involutive space.
\begin{proposition}[Local $\rr{R}$-triviality]
\label{prop:loc_triv}
Let $(\bb{E},\Theta)$ be an $\rr{R}$-bundle
over the involutive space $(X,\tau)$ such that $X$ verifies Assumption \ref{ass:1}.
 Then $\pi:\bb{E}\to X$ is \emph{locally $\rr{R}$-trivial}, meaning that for all $x\in X$ there exists a $\tau$-invariant neighborhood $U$ of $x$ and an $\rr{R}$-isomorphism $h:\pi^{-1}(U)\to U\times\C^m$ where the product bundle $U\times\C^m\to U$ is endowed with the trivial $\rr{R}$-structure given by the complex conjugation. Moreover, if $x=\tau(x)$  the neighborhood $U$  can be chosen connected. If $x\neq\tau(x)$ then $U$ can be taken as the union of two disjoint open sets $U:={U}'\cup{U}''$ with $x\in {U}'$
and $\tau:{U}'\to {U}''$ a homeomorphism.
\end{proposition}

The proof of this result is explained in  \cite[pg. 374]{atiyah-66} and uses the Extension Lemma \ref{lemma:ext} together with
 the fact that each morphism $f:\bb{E}\to\bb{E}'$ between complex vector bundles over the same base space  $X$ can be identified with a section of the \emph{homomorphism bundle} ${\rm Hom}_\C(\bb{E},\bb{E}')\to X$.

\begin{remark}[Equivariant metric]\label{rk:eq_metric}{\upshape
Under the assumption of 
Proposition \ref{prop:loc_triv} let $\{U_i\}$ be  a $\tau$-invariant open covering of $(X,\tau)$ which trivializes locally the $\rr{R}$-bundle
$(\bb{E},\Theta)$. Let $\{{\rho}'_i\}$ be any \emph{partition of unity} subordinate to $\{U_i\}$. The collection $\{{\rho}_i\}$
given by ${\rho}_i(x):=\frac{1}{2}\big[{\rho}'_i(x)+{\rho}'_i(\tau(x))\big]$ provides a $\tau$\emph{-invariant} partition of unity subordinate to the same covering. Let $\bb{E}\times_X \bb{E}:=\{(p_1,p_2)\in\bb{E}\times \bb{E}\ |\ \pi(p_1)=\pi(p_2)\}$.
With the help of $\{{\rho}_i\}$ one can define a \emph{Hermitian metric} 
$\rr{m}:\bb{E}\times_X \bb{E}\to\C$
on the complex vector bundle $\pi:\bb{E}\to X$ following the standard construction in \cite[Chapter I, Theorem 8.7]{karoubi-97}. The $\tau$-invariance of the partition of unit implies $\rr{m}\big(\Theta(p_1),\Theta(p_2)\big)=\rr{m}\big(p_2,p_1\big)$ for all $(p_1,p_2)\in\bb{E}\times_X \bb{E}$. In this sense the involution $\Theta$ acts as an \virg{anti-unitary} map. The choice of such a metric is essentially unique up to isomorphisms (compare with \cite[Chapter I, Theorem 8.8]{karoubi-97}).
}\hfill $\blacktriangleleft$
\end{remark}

Let $(X_1,\tau_1)$ and $(X_2,\tau_2)$ be two involutive spaces. A map $\varphi:X_1\to X_2$ is called \emph{equivariant} if $\varphi\circ\tau_1=\tau_2\circ \varphi$. An \emph{equivariant homotopy} between  equivariant maps $\varphi_0$ and $\varphi_1$ is a continuous map $F:[0,1]\times X_1\to X_2$ such that each $\varphi_t(\cdot):=F(t,\cdot)$ is an equivariant map for all $t\in[0,1]$. The set of the equivalence classes of equivariant homotopic maps between $(X_1,\tau_1)$ and $(X_2,\tau_2)$ will be denoted by $[X_1,X_2]_{\rm eq}$.
\begin{theorem}[Homotopy property]
\label{theo:equi_homot1}
Let $(X_1,\tau_1)$ and $(X_2,\tau_2)$ be two involutive spaces with  $X_1$ verifying  Assumption \ref{ass:1}. Let $(\bb{E},\Theta)$ be a rank $m$ $\rr{R}$-bundle over $(X_2,\tau_2)$ and $F:[0,1]\times X_1\to X_2$ an equivariant homotopy between the equivariant maps $\varphi_0$ and $\varphi_1$. Then, the pullback bundles $\varphi_t^\ast\bb{E}\to X_1$ have an induced $\rr{R}$-structure for all $t\in[0,1]$ and
$\varphi_0^\ast\bb{E}\simeq \varphi_1^\ast\bb{E}$ in  ${\rm Vec}_{\rr{R}}^m(X_1,\tau_1)$.
\end{theorem}
This theorem can be proved by a suitable equivariant generalization of \cite[Proposition 1.3]{atiyah-bott-64} (see also \cite[Chapter 1, Theorem 8.15]{tom-dieck-87} for a more general context). We remark only that by definition  $\varphi_t^\ast\bb{E}|_x:= \{x\}\times\bb{E}|_{\varphi_t(x)}$ for all $x\in X_1$. Hence, the equivariance of $\varphi_t$  implies  $\varphi_t^\ast\bb{E}|_{\tau_1(x)}:= \{\tau_1(x)\}\times \bb{E}|_{\tau_2(\varphi_t(x))}$. Then, an $\rr{R}$-structure on $\varphi_t^\ast\bb{E}$ is induced via  pullback by the map $\Theta_t^\ast$ which acts anti-linearly between the fibers $\varphi_t^\ast\bb{E}|_x$ and $\varphi_t^\ast\bb{E}|_{\tau_1(x)}$ as the  product $\tau_1\times\Theta$.

\medskip

Theorem \ref{theo:equi_homot1} is the starting point for a 
homotopy classification   of $\rr{R}$-bundles in the spirit of \eqref{eq:class_compl_VB1}.  
In Example \ref{ex:grassman_inv} we showed how  the complex conjugation endows the Grassmann manifold
 with the structure of an involutive space. Similarly,  the {tautological $m$-plane bundle} $\pi:\bb{T}_m^\infty\to G_m(\C^\infty)$ introduced in Section \ref{sec:hom_vecc_compx}
 can be endowed 
  with the structure of an $\rr{R}$-bundle compatible with the involution $\varrho$. The \emph{anti}-linear map $\Xi:\bb{T}_m^\infty\to\bb{T}_m^\infty$ defined by $\Xi:(\Sigma,{\rm v})\mapsto \big(\varrho(\Sigma),\overline{\rm v}\big)$  verifies
$\pi\circ\Xi=\varrho\circ{\pi}$ and so $(\bb{T}_m^\infty,\Xi)$ becomes an $\rr{R}$-bundle over the involutive space  $(G_m(\C^\infty),\varrho)$. This  is the  \emph{universal} object for a  homotopy classification of $\rr{R}$-bundles.

\begin{theorem}[Homotopy classification \cite{edelson-71}]\label{theo:honotopy_class}
Let $(X,\tau)$ be an involutive space with $X$ verifying  Assumption \ref{ass:1}.
 Each $\rr{R}$-bundle $(\bb{E},\Theta)$ over $(X,\tau)$ can be obtained, up to isomorphisms, as a pullback $\bb{E}\simeq \varphi^\ast \bb{T}_m^\infty$ with respect to a map $\varphi:X\to G_m(\C^\infty)$ which is equivariant $\varphi\circ\tau=\varrho\circ \varphi$. Moreover, the homotopy property implies that $(\bb{E},\Theta)$ depends only on the equivariant homotopy class of $\varphi$, \ie one has  the isomorphism
$$
{\rm Vec}^m_\rr{R}(X,\tau)\;\simeq\;[X,G_m(\C^\infty)]_{\rm eq}\;.
$$
\end{theorem}

\medskip

As for the complex case (\cf Section \ref{sec:hom_vecc_compx}), one can consider the inductive limit $B\n{U}:=\bigcup_{m=0}^\infty G_m(\C^\infty)$ endowed with the induced involution $\varrho$. This space classifies \emph{stable} \virg{Real} vector bundles independently of the dimension of the fiber, \ie 
\beql{eq:class_compl_RB_stab}
{\rm Vec}_\rr{R}(X,\tau)\;\simeq\;[X,B\n{U}]_{\rm eq}\;.
\eeq
This result, proved in \cite{nagata-nishida-toda-82}, is relevant for the $KR$-theory (\cf Appendix \ref{sect:KR-theory}).

\medskip

\begin{remark}\label{rk:iso_categ}{\upshape
If one consider the trivial involutive space $(X,{\rm Id}_X)$, then the equivariant maps $\varphi:X\to G_m(\C^\infty)$ are characterized  by $\varphi(x)=\overline{\varphi(x)}$ for all $x\in X$.
Since the fixed point set of $G_m(\C^\infty)$ is parameterized by $G_m(\R^\infty)$ (\cf Example \ref{ex:grassman_inv}) one has that $[X,G_m(\C^\infty)]_{\rm eq}\simeq[X,G_m(\R^\infty)]$. 
However, $G_m(\R^\infty)$ is the classifying space for real vector bundles \cite[Chapter 5]{milnor-stasheff-74} hence
${\rm Vec}^m_{\rr{R}}(X,{\rm Id}_X)\simeq{\rm Vec}^m_{\R}(X)$ in agreement with
Proposition \ref{prop:Rv1}.
}\hfill $\blacktriangleleft$
\end{remark}

The homotopy classification provided by  Theorem \ref{theo:honotopy_class} 
can be directly used to classify $\rr{R}$-bundles over the TR-involutive space $\tilde{\n{S}}^1\equiv(\n{S}^1,\tau)$.   
\begin{proposition}
\label{prop_(i)}
$
{\rm Vec}^m_{\rr{R}}(\n{S}^1, \tau)\;=\;0\;.
$
\end{proposition}
\proof
In view of Theorem \ref{theo:honotopy_class} it is enough to  show that $[\tilde{\n{S}}^1,G_m(\C^\infty)]_{\rm eq}$ reduces to the equivariant homotopy class  of the constant map.  This fact follows directly from a more general result proved in Lemma
\ref{lemma:Z2-reduct} which is based on the $\Z_2$-skeleton decomposition of $\tilde{\n{S}}^1$.  The involutive space $\tilde{\n{S}}^1$ has 2 fixed cells of dimension $0$ and 1 free cell of dimension 1 (\cf Example \ref{ex:sphere_TR_CW}). Moreover, we know that $\pi_0(G_m(\C^\infty))\simeq \pi_1(G_m(\C^\infty))\simeq 0$ and the identification between the fixed point set $G_m(\C^\infty)^\varrho$ with $G_m(\R^\infty)$ (\cf Example \ref{ex:grassman_inv}) implies also $\pi_0(G_m(\C^\infty)^\varrho)\simeq0$.
These data are sufficient for the application of Lemma
\ref{lemma:Z2-reduct}.
\qed

\begin{remark}\label{rk:gen_prof_int}{\upshape
The same proof works exactly in the same way if in Proposition \ref{prop_(i)} we replace the TR-involutive sphere
$(\n{S}^1, \tau)$ with any other involutive space $(X, \tau)$ with cells of maximal dimension 1 and fixed cells only in dimension $0$.
}\hfill $\blacktriangleleft$
\end{remark}

In the case of  higher dimensional spheres with TR-involution $\tilde{\n{S}}^d\equiv({\n{S}}^d,\tau)$ we can generalize  \cite[Lemma 1.4.9]{atiyah-67} to the \virg{Real} category. We recall that $\hat{\n{S}}^d\equiv({\n{S}}^d,\vartheta)$ denotes the sphere with free antipodal action (\cf Example \ref{ex:sphere_antipodal}).

\begin{proposition}[Equivariant clutching construction]
\label{prop:eq_clut_constr}
For any $m,d\in\N$ there is a natural bijection:
$$
{\rm Vec}^m_{\rr{R}}(\n{S}^d, \tau)
\;\simeq\;[\hat{\n{S}}^{d-1}, {\rm S}\n{U}(m)]_{\rm eq}\;/\; \Z_2.
$$
In the above, the set of homotopy classes of $\Z_2$-equivariant maps is defined with respect to the involution on the special unitary group ${\rm S}\n{U}(m)$ given by the complex conjugation  and the  $\Z_2$-action on $[\hat{\n{S}}^{d-1}, {\rm S}\n{U}(m)]_{\rm eq}$ is defined by $[\varphi] \mapsto [\epsilon \varphi \epsilon]$, where $\epsilon = \mathrm{diag}(-1, 1 \ldots, 1)$.
\end{proposition}

For the proof we need the suspension relation $\tilde{\n{S}}^d\simeq \s{S}(\hat{\n{S}}^{d-1})$ discussed in Example \ref{ex:homot_spher} which is the same as the decomposition
$$
\tilde{\n{S}}^d\;\simeq \;\s{C}_+(\hat{\n{S}}^{d-1})\;\cup\; \s{C}_-(\hat{\n{S}}^{d-1})
$$
where $\s{C}_+(\hat{\n{S}}^{d-1}):=(\hat{\n{S}}^{d-1}\times[0,\frac{1}{2}])/(\hat{\n{S}}^{d-1}\times\{0\})$ and 
 $\s{C}_-(\hat{\n{S}}^{d-1}):=(\hat{\n{S}}^{d-1}\times[\frac{1}{2},1])/(\hat{\n{S}}^{d-1}\times\{1\})$
are the two cones which  
form the suspension of $\hat{\n{S}}^{d-1}=\s{C}_+(\hat{\n{S}}^{d-1})\cap\s{C}_-(\hat{\n{S}}^{d-1})$. 
The group of equivariant homotopic maps $[\s{C}_\pm(\hat{\n{S}}^{d-1}), \n{U}(m)]_{\rm eq}$ and $[\hat{\n{S}}^{d-1}, \n{U}(m)]_{\rm eq}$ are understood with respect to 
 the involution on $\n{U}(m)$ given by the complex conjugation. 
We prove first two preliminary results.

\begin{lemma}
\label{lemma:doub_coset}
There is a natural bijection 
$$
{\rm Vec}^m_{\rr{R}}(\n{S}^d, \tau)
\;\simeq\;
[\s{C}_+(\hat{\n{S}}^{d-1}), \n{U}(m)]_{\rm eq}\; \backslash\;
[\hat{\n{S}}^{d-1}, \n{U}(m)]_{\rm eq}\; /\;
[\s{C}_-(\hat{\n{S}}^{d-1}), \n{U}(m)]_{\rm eq}
$$
where on the right-side there is the 
 double coset of $[\hat{\n{S}}^{d-1}, \n{U}(m)]_{\rm eq}$ given by  the equivalence relation $[\varphi] \sim [\psi_+\cdot \varphi \cdot\psi_-^{-1}]$
for  some pair
$[\psi_\pm] \in [\s{C}_\pm(\hat{\n{S}}^{d-1}), \n{U}(m)]_{\rm eq}$. 
\end{lemma}

\begin{proof}
According to Remark 
\ref{rk:eq_metric} we can assume without loss of generality that equivariant  Hermitian metrics are given on $\rr{R}$-bundles.
First we construct a map
$$
\xi: \
[\s{C}_+(\hat{\n{S}}^{d-1}), \n{U}(m)]_{\rm eq}\; \backslash\;
[\hat{\n{S}}^{d-1}, \n{U}(m)]_{\rm eq} \;/\;
[\s{C}_-(\hat{\n{S}}^{d-1}), \n{U}(m)]_{\rm eq}\;
\longrightarrow\;
{\rm Vec}^m_{\rr{R}}(\n{S}^d, \tau)\;,
$$
mimicking the standard clutching construction \cite[Lemma 1.4.9]{atiyah-67}. Let $\varphi : \hat{\n{S}}^{d-1} \to \n{U}(m)$ be an equivariant map representing an element in the double coset. Then we glue together the product $\rr{R}$-bundles on $\s{C}_\pm(\hat{\n{S}}^{d-1})$ (with the standard Hermitian metric)  to get an $\rr{R}$-bundle on $\tilde{\n{S}}^d$
\beql{eq:clut_constr_01}
\big(\s{C}_+(\hat{\n{S}}^{d-1})\times\C^m\big)\;\cup_\varphi\; \big(\s{C}_-(\hat{\n{S}}^{d-1})\times\C^m\big)\;.
\eeq
If we replace $\varphi$ by $\varphi' = \psi_+\cdot \varphi\cdot \psi_-^{-1}$ using equivariant maps $\psi_\pm : \s{C}_\pm(\hat{\n{S}}^{d-1}) \to \n{U}(m)$ we obtain an isomorphism of $\rr{R}$-bundles
$$
\big(\s{C}_+(\hat{\n{S}}^{d-1})\times\C^m\big)\;\cup_\varphi\; \big(\s{C}_-(\hat{\n{S}}^{d-1})\times\C^m\big) \to
\big(\s{C}_+(\hat{\n{S}}^{d-1})\times\C^m\big)\;\cup_{\varphi'}\; \big(\s{C}_-(\hat{\n{S}}^{d-1})\times\C^m\big)\;.
$$
Moreover, due to the homotopy property of $\rr{R}$-bundles, the isomorphism class of the $\rr{R}$-bundle \eqref{eq:clut_constr_01}  is not altered under the changes of $\psi_\pm$ and $\varphi$ by homotopy. Hence $\xi$ is well defined. 

On the other hand, we can as well construct a map
$$
\nu\;: \;
{\rm Vec}^m_{\rr{R}}(\n{S}^d, \tau)\; \longrightarrow\;
[\s{C}_+(\hat{\n{S}}^{d-1}), \n{U}(m)]_{\rm eq}\; \backslash\;
[\hat{\n{S}}^{d-1}, \n{U}(m)]_{\rm eq} \;/\;
[\s{C}_-(\hat{\n{S}}^{d-1}), \n{U}(m)]_{\rm eq}\;.
$$
 Let $\bb{E} \to \tilde{\n{S}}^d$ be a given $\rr{R}$-bundle. Since the cones $\s{C}_\pm(\hat{\n{S}}^{d-1})$ are $\tau$-invariantly contractible, the restrictions $\bb{E}|_{\s{C}_\pm(\hat{\n{S}}^{d-1})}$ are $\rr{R}$-trivial. Then a choice of trivializations $h_\pm : \bb{E}|_{\s{C}_\pm(\hat{\n{S}}^{d-1})} \simeq \s{C}_\pm(\hat{\n{S}}^{d-1}) \times \C^m$ induces an equivariant map $\varphi : \hat{\n{S}}^{d-1} \to \n{U}(m)$ as well as an isomorphism of $\rr{R}$-bundles
$$
\bb{E}\;\simeq\;\big(\s{C}_+(\hat{\n{S}}^{d-1})\times\C^m\big)\;\cup_\varphi\; \big(\s{C}_-(\hat{\n{S}}^{d-1})\times\C^m\big).
$$
If we choose other trivializations $h'_\pm$, they differ from $h_\pm$ only by equivariant maps $\psi_\pm : \s{C}_\pm(\hat{\n{S}}^{d-1}) \to \n{U}(m)$, and the resulting clutching function $\varphi'$ is expressed as $\varphi' = \psi_+\cdot \varphi\cdot \psi_-^{-1}$. Thus the assignment $\nu:[\bb{E}] \mapsto [\varphi]$ gives rise to a well-defined map. By the constructions, $\xi$ and $\nu$ are inverse to each other. 
\end{proof}

\begin{lemma}
\label{lemma:semidirect_product}
For any involutive space $(X,\tau)$, there is an isomorphism of groups
$$
[X, \n{U}(m)]_{\rm eq}\;\simeq\;[X, {\rm S}\n{U}(m)]_{\rm eq}\;\rtimes\;[X, \n{U}(1)]_{\rm eq}
$$
where $\rtimes$ denotes the \emph{semidirect} product.
\end{lemma}

\begin{proof}
Let us start with the sequence of groups
\beql{eq:exact_seq_00}
1\;\to\; \bb{M}\big(X, {\rm S}\n{U}(m)\big)\;\stackrel{\imath}{\hookrightarrow}\; \bb{M}\big(X, \n{U}(m)\big)\;\stackrel{\jmath}{\to}\; \bb{M}\big(X, \n{U}(1)\big)\;\to\;1
\eeq
 where $\bb{M}\big(X, Y\big)$ is the group of maps from $X$ to $Y$, $\imath$ is the inclusion induced by ${\rm S}\n{U}(m)\hookrightarrow
\n{U}(m)$ and $\jmath$ is induced by ${\rm det}:\n{U}(m)\to \n{U}(1)$. This sequence is exact and splits on $\jmath$ since there is a map $s:  \bb{M}\big(X, \n{U}(1)\big)\to  \bb{M}\big(X, \n{U}(m)\big)$ such that $\jmath\circ s={\rm Id}$. Such a map is explicitly realized by
$$
s(u)\;:=\;
\left(\begin{array}{cc}u & 0 \\0 & \n{1}_{m-1}\end{array}\right)\;\qquad \quad u\in\bb{M}\big(X, \n{U}(1)\big)\;.
$$
The sequence \eqref{eq:exact_seq_00} induces a similar sequence at level of equivalence classes of homotopy maps. Moreover, the complex conjugation on $\n{U}(m)$ preserves the group product, hence the sequence can be generalized to equivariant maps. This means that we have an exact sequence of groups
\beql{eq:exact_seq_01}
1\;\to\; [X, {\rm S}\n{U}(m)]_{\rm eq}\;\stackrel{\imath}{\hookrightarrow}\; [X, \n{U}(m)]_{\rm eq}\;\stackrel{\jmath}{\to}\; [X, \n{U}(1)]_{\rm eq}\;\to\;1
\eeq
which is splitting exact in $\jmath$. Then a general argument (splitting Lemma) completes the proof.
\end{proof}

\proof[Proof of Proposition \ref{prop:eq_clut_constr}]
We have
\begin{align*}
[\s{C}_\pm(\hat{\n{S}}^{d-1}), {\rm S}\n{U}(m)]_{\rm eq}
&\;\simeq\; 
[\{\ast\}, {\rm S}\n{U}(m)]_{\rm eq}
\;\simeq\;
[\{\ast\}, {\rm S}\n{O}(m)] \;\simeq\; 1, \\
[\s{C}_\pm(\hat{\n{S}}^{d-1}), \n{U}(1)]_{\rm eq}
&\;\simeq\; 
[\{\ast\}, \n{U}(1)]_{\rm eq}
\;\simeq\;
[\{\ast\}, \Z_2] \simeq \Z_2
\end{align*}
which, in view of Lemma \ref{lemma:semidirect_product}, means that we can replace $[\s{C}_\pm(\hat{\n{S}}^{d-1}), \n{U}(m)]_{\rm eq}$ by the semidirect product
 $1 \rtimes[\s{C}_\pm(\hat{\n{S}}^{d-1}), \n{U}(1)]_{\rm eq}$ 
in the definition of the double coset introduced in Lemma \ref{lemma:doub_coset}.

Using also the splitting 
$[\hat{\n{S}}^{d-1}, \n{U}(m)]_{\rm eq}\simeq[\hat{\n{S}}^{d-1}, {\rm S}\n{U}(m)]_{\rm eq}\rtimes[\hat{\n{S}}^{d-1}, \n{U}(1)]_{\rm eq}$ we can rewrite the  coset as
\begin{multline*}
\Big([\hat{\n{S}}^{d-1}, {\rm S}\n{U}(m)]_{\rm eq}\;/\;\Z_2\Big)
 \;{\times}\;\Big([\s{C}_+(\hat{\n{S}}^{d-1}), \n{U}(1)]_{\rm eq}\; \backslash\;
[\hat{\n{S}}^{d-1}, \n{U}(1)]_{\rm eq}\; /\;
[\s{C}_-(\hat{\n{S}}^{d-1}), \n{U}(1)]_{\rm eq}\Big)\;
\end{multline*}
where the $\Z_2$ action on $[\hat{\n{S}}^{d-1}, {\rm S}\n{U}(m)]_{\rm eq}$ is given by the  non-trivial elements $-1 \in [\s{C}_\pm(\hat{\n{S}}^{d-1}), \n{U}(1)]_{\rm eq} $ according to  $[\varphi] \mapsto [s(-1) \varphi s(-1)^{-1}] = [\epsilon \varphi \epsilon]$. The double coset in the direct product above is isomorphic to ${\rm Vec}^1_{\rr{R}}(\n{S}^d, \tau)$ in view Lemma \ref{lemma:doub_coset}. Moreover, in Proposition \ref{prob:R_linS} is proved that ${\rm Vec}^1_{\rr{R}}(\n{S}^d, \tau)\simeq 0$ and this concludes the proof.
\qed

\begin{remark}{\upshape
One is tempted to derive Proposition \ref{prop:eq_clut_constr}
from a direct application of Lemma \ref{lemma:semidirect_product}.
Unfortunately (and surprisingly!)
the group $[\hat{\n{S}}^{d-1}, \n{U}(1)]_{\rm eq}$ is non-trivial. Concretely  $[\hat{\n{S}}^{d-1}, \n{U}(1)]_{\rm eq} \simeq \Z_2$ and a proof of this fact follows from the identification between $[\hat{\n{S}}^{d-1}, \n{U}(1)]_{\rm eq}$ and the equivariant cohomology $H^1_{\Z_2}(\hat{\n{S}}^{d-1}, \Z(1))$ 
which will be introduced in Section \ref{ssec:borel_constr}.
For the computation of the cohomology one uses the isomorphism 
$H^1_{\Z_2}(\hat{\n{S}}^{d-1}, \Z(1))\simeq H^1(\R P^{d-1}, \Z(1))$ and equation \eqref{eq:cohomPR_Z1} or one can employ the \emph{Gysin sequence} as explained in \cite{gomi-13}. 
}\hfill $\blacktriangleleft$
\end{remark}


\subsection{Stable range condition and equivariant CW-decomposition}
\label{ssec:stab_reng_equi_constr}

In topology  spaces which are homotopy equivalent to CW-complexes are very important. A similar notion can be extended to $\Z_2$-spaces like
 $(\n{S}^d,\tau)$ and $(\n{T}^d,\tau)$.
These involutive spaces have the structure of a $\Z_2$- CW-complex according to  \cite[Definition 1.1.1]{allday-puppe-93}. We sketch briefly this notion:
\begin{itemize}
\item[-] {\bf $\Z_2$-cell.} The group $\Z_2\equiv\{+1,-1\}$, which is relevant for spaces with involution,  has only two subgroups: the unit $\{+1\}$ and the full group $\Z_2$ itself.
Let
$
\n{D}^n:=\{k\in\R^n\,|\; \|k\|\leqslant 1\}
$
 be the closed unit ball (or \emph{disk}) in $\R^n$.
A \emph{fixed} cell of dimension $n$ is a  $\Z_2$-space
${\bf e}^n:=\{+1\}\times \n{D}^n\simeq\n{D}^n$ on which the action of the involution is  trivial.
A \emph{free} cell of dimension $n$ is a  $\Z_2$-space
$\tilde{{\bf e}}^n:=\Z_2\times \n{D}^n$ on which the involution acts trivially on the ball $\n{D}^n$ and by permutation on $\Z_2$. The boundary of $\n{D}^n$ is the sphere $\n{S}^{n-1}$. 
The \emph{$\Z_2$-boundary} of a fixed cell ${\bf e}^n$ is $\partial{\bf e}^n:=\{+1\}\times \n{S}^{n-1}\simeq\n{S}^{n-1}$. Similarly for a free cell $\tilde{{\bf e}}^n$ the  {$\Z_2$-boundary} is $\partial\tilde{{\bf e}}^n:=\Z_2\times \n{S}^{n-1}$.\vspace{0.8mm}

\item[-] {\bf $0$-skeleton.} A disjoint union of $\Z_2$-cells of dimension $0$
\beql{eq:zero_skel}
X^0\;:=\;\left(\coprod_{j=1}^{N_0}\; {\bf e}^0_j\right)\; \amalg\; \left(\coprod_{j=1}^{\tilde{N}_0}\; \tilde{{\bf e}}^0_j\right)
\eeq
is a $0$-skeleton with $N_0$ fixed $0$-cells ${\bf e}^0_j\simeq\{\ast\}$ and $\tilde{N}_0$ free
$0$-cells $\tilde{{\bf e}}^0_j\simeq\Z_2$. The numbers $N_0$ and $\tilde{N}_0$ are positive integers
which can also take the value $\infty$. We use $N_0=0$ or $\tilde{N}_0=0$ as a convention to indicate the absence 
of fixed cells or free cells, respectively. 
\vspace{0.8mm}

\item[-] {\bf $\Z_2$-CW-complex.} A $\Z_2$-space $X$ has a \emph{$\Z_2$-skeleton decomposition} if it can be written as a {colimit} of inclusions $X^0\subset X^1\subset\ldots\subset X^{n-1}\subset X^{n}\subset\ldots$, where $X^0$ is a $0$-skeleton and the \emph{$n$-skeleton} $X^{n}$ is obtained from the {$n-1$-skeleton} $X^{n-1}$ by attaching a disjoint union  of $\Z_2$-cells of dimension $n$ along the boundaries via \emph{equivariant} attaching  maps $\phi_j:\{+1\}\times \n{S}^{n-1}\to X^{n-1}$ (for fixed cells) and
$\tilde{\phi}_j:\Z_2\times \n{S}^{n-1}\to X^{n-1}$ (for free cells). More precisely one defines
$$
X^{n}\;:=\; X^{n-1}\;\bigcup_{\phi_j}\;\left(\coprod_{j=1}^{N_n}\; {\bf e}^n_j\right)\;\bigcup_{\tilde{\phi_j}}\;\left(\coprod_{j=1}^{\tilde{N}_n}\; \tilde{{\bf e}}^n_j\right) 
$$
where the symbols $\bigcup_{\phi_j}$ and $\bigcup_{\tilde{\phi}_j}$ denote the union modulo  the identification
$\phi_j\big(a\big)\sim a$ for all $a\in\partial{{\bf e}}^n_j$  and 
$\tilde{\phi}_j\big(a\big)\sim a$ for all $a\in\partial{\tilde{\bf e}}^n_j$.
A \emph{finite} $\Z_2$-CW-complex is a $\Z_2$-space $X$ with a  skeleton decomposition  given by a finite number of $\Z_2$-cells.
\vspace{0.8mm}

\item[-] {\bf Orbit space.} The orbit space $X/\Z_2$ of a $\Z_2$-CW-complex inherits, in obvious way, the  CW-complex structure. The (non-equivariant) cells of $X/\Z_2$ are just the orbit spaces of the $\Z_2$-cells of $X$ and the attaching maps of $X/\Z_2$ are those induced from the equivariant maps of $X$ by forgetting the $\Z_2$-action.
\end{itemize}

 \medskip
 
The above construction is modeled after the usual definition of CW-complex, replacing the \virg{point} by \virg{$\Z_2$-point}. For this reason many topological and homological properties of CW-complexes have their \virg{natural} counterparts in the equivariant setting. In the following examples we describe explicitly  $\Z_2$-CW-complex structures of the involutive spaces $(\n{S}^d,\tau)$ and $(\n{T}^d,\tau)$.

\begin{example}[$\Z_2$-CW decomposition for  spheres with TR-involution]\label{ex:sphere_TR_CW}{\upshape
We recall the notation $\tilde{\n{S}}^d\equiv(\n{S}^{d},\tau)$.
Let us start with the one-dimensional case  $\tilde{\n{S}}^1$ which will be relevant also for the torus. In this case the $0$-skeleton 
of $\tilde{\n{S}}^1$ is made by two fixed $0$-cells ${\bf e}^0_\pm$ identified with the two fixed point $(\pm 1,0)\in\n{S}^1$, respectively. Attaching to the $0$-skeleton the free 1-cell $\tilde{\bf e}^1$ identified with the pair $\{\ell_+,\ell_-\}$ of conjugate (1-dimensional) hemispheres $\ell_\pm:=\{(k_0,k_1)\in\n{S}^1\ |\ \pm k_1\geqslant0\}$, one obtains $\tilde{\n{S}}^1$.
A $\Z_2$-skeleton decomposition of $\tilde{\n{S}}^d$ can be defined inductively observing that 
$\tilde{\n{S}}^d$ can be obtained from $\tilde{\n{S}}^{d-1}$ adding a free $d$-cell $\tilde{\bf e}^d$ 
 identified with the pair $\{\ell^d_+,\ell^d_-\}$ of conjugate $d$-dimensional hemispheres $\ell^d_\pm:=\{(k_0,k_1,\ldots,k_d)\in\n{S}^d\ |\ \pm k_d\geqslant0\}$. Summarizing, $\tilde{\n{S}}^d$ has two fixed cells in $0$-dimension  (\ie $N_0=2$ and $\tilde{N}_0=0$) and one free cell in any dimension $1\leqslant n\leqslant d$ (\ie $N_n=0$ and $\tilde{N}_n=1$).
}\hfill $\blacktriangleleft$
\end{example}

\medskip

\begin{example}[$\Z_2$-CW decomposition for tori with TR-involution]\label{ex:tori_TR_CW}{\upshape
The construction of a $\Z_2$-skeleton decomposition of $\tilde{\n{T}}^d\equiv(\n{T}^{d},\tau)$ can be derived from that of 
$\tilde{\n{S}}^1$ and the Cartesian product structure $\tilde{\n{T}}^d=\tilde{\n{S}}^1\times\ldots\times\tilde{\n{S}}^1$.
The $0$-cells are products of the two fixed $0$-cells ${\bf e}^0_\pm$ of $\tilde{\n{S}}^1$, namely
${\bf e}^0_{\sigma_1,\ldots,\sigma_d}:=({\bf e}^0_{\sigma_1},\ldots,{\bf e}^0_{\sigma_d})={\bf e}^0_{\sigma_1}\times\ldots\times{\bf e}^0_{\sigma_d}$ where $\sigma_j\in\{+,-\}$ for all $1\leqslant j\leqslant d$. Hence, one has $N_0=2^d$ fixed $0$-cells and $\tilde{N}_0=0$ free $0$-cells.
The $1$-cells are obtained by replacing the free 1-cell $\tilde{\bf e}^1$ of $\tilde{\n{S}}^1$ in place of one of the fixed $0$-cells, \ie
$$
\tilde{\bf e}^1_{\sigma_1,\ldots,\underline{\sigma_j},\ldots,\sigma_d}\;:=\;({\bf e}^0_{\sigma_1},\;\ldots\;,{\bf e}^0_{\sigma_{j-1}},\tilde{\bf e}^1,{\bf e}^0_{\sigma_{j+1}},\;\ldots\;,{\bf e}^0_{\sigma_d})\;.
$$
Then, one has $N_1=0$ fixed $1$-cells and $\tilde{N}_1=d\;2^{d-1}$ free $1$-cells. For the $2$-cells one has that
for each pair $i<j$ the cells
$$
\begin{aligned}
\tilde{\bf e}^2_{\sigma_1,\ldots,\underline{\sigma_i},\ldots,\underline{\sigma_j},\ldots,\sigma_d}\;&:=\;({\bf e}^0_{\sigma_1},\;\ldots\;,{\bf e}^0_{\sigma_{i-1}},\tilde{\bf e}^1,{\bf e}^0_{\sigma_{i+1}},\;\ldots\;,{\bf e}^0_{\sigma_{j-1}},\tilde{\bf e}^1,{\bf e}^0_{\sigma_{j+1}},\;\ldots\;,{\bf e}^0_{\sigma_d})\\
{{\tilde{\bf e}}'}{^2}_{\sigma_1,\ldots,\underline{\sigma_i},\ldots,\underline{\sigma_j},\ldots,\sigma_d}\;&:=\;({\bf e}^0_{\sigma_1},\;\ldots\;,{\bf e}^0_{\sigma_{i-1}},\tilde{\bf e}^1,{\bf e}^0_{\sigma_{i+1}},\;\ldots\;,{\bf e}^0_{\sigma_{j-1}},\tau(\tilde{\bf e}^1),{\bf e}^0_{\sigma_{j+1}},\;\ldots\;,{\bf e}^0_{\sigma_d})
\end{aligned}
$$
behave differently under the $\Z_2$-action of the involution $\tau$ and a simple combinatorial computation leads to 
$N_2=0$ and $\tilde{N}_2=d(d-1)\;2^{d-2}$.
The construction of the higher dimensional cells follow similarly and one obtains that $\tilde{\n{T}}^d$ has no fixed cells of dimension bigger than $0$ (\ie $N_n=0$ for all $n>0$) and $\tilde{N}_n=\binom {d} {n}\;2^{d-1}$ free $n$-cells for all $1\leqslant n\leqslant d$.
}\hfill $\blacktriangleleft$
\end{example}

 \medskip

In both the examples of interest $(\n{S}^d,\tau)$ and $(\n{T}^d,\tau)$ the $\Z_2$-CW-complex structure has fixed cells only in dimension 0. 
This fact  plays an important role for the next result which is the equivariant generalization of  \cite[Chapter 2, Theorem 7.1]{husemoller-94}.
\begin{proposition}[existence of a global $\rr{R}$-section]
\label{prop:stab_ran_R}
Let $(X,\tau)$ be an involutive space such that $X$ has a finite $\Z_2$-CW-complex decomposition with fixed cells only in dimension 0. 
Let us denote with $d$ the dimension of $X$ (\ie the maximal dimension of the cells in the decomposition of $X$). Let $(\bb{E},\Theta)$ be an $\rr{R}$-vector bundle over $(X,\tau)$ with fiber of rank $m$. If $d<2m$ there exists an $\rr{R}$-section $s:X\to\bb{E}$ which is global in the sense that $s(x)\neq 0$ for all $x\in X$.  
\end{proposition}

\proof

 The \emph{zero} section $s_0(x)=0\in\bb{E}_x$ for all
$x\in X$  is an  $\rr{R}$-section since $\Theta$ is an anti-linear isomorphism in each fiber. Let $\bb{E}^{\times}\subset\bb{E}$ be the subbundle of nonzero vectors. The fibers $\bb{E}^{\times}_x$ are all isomorphic to $(\C^m)^\times:=\C^m\setminus\{0\}$
which is a $2(m-1)$-connected space (\ie the first $2(m-1)$  homotopy groups $\pi_i$ vanish identically). The involution $\Theta$ endows $\bb{E}^{\times}$ with an $\rr{R}$-structure over  $(X,\tau)$. An $\rr{R}$-section of $\bb{E}^{\times}$ can be seen as an 
everywhere-nonzero (\ie global) $\rr{R}$-section for the $\rr{R}$-vector bundle $(\bb{E},\Theta)$. We will show that if $d<2m$ such a section of $\bb{E}^{\times}$ always exists. We notice that each fiber $\bb{E}^{\times}_x$ is $(j-1)$-connected for all $j\leqslant d$.

We prove the claim by induction on the dimension of the skeleton.
 This is the case for $X^0$ which is a finite collection of fixed points $\{x_j\}_{j=1,\ldots,N_0}$ and conjugated pairs $\{x_j, \tau(x_j)\}_{j=1,\ldots,\tilde{N}_0}$.
  In this case a global $\rr{R}$-section is given by setting $s'(x_j)\in\bb{E}^{\Theta}_x\cap \bb{E}^{\times}_x\simeq\R^m\setminus\{0\}$ for the fixed points $x_j$ and  $s'(x_j)\in\bb{E}^{\Theta}_x\simeq\C^m\setminus\{0\}$ for the free pairs $\{x_j, \tau(x_j)\}$
 together with the equivariant constraints $s'(\tau(x_j)):=\Theta(s'(x_j))$. 
Assume now that, for $j$ such that $1 \leqslant j \leqslant d$, the claim is true for the $\Z_2$-CW-subcomplex $X^{j-1}$ of dimension $j-1$. By the inductive hypothesis we have an $\rr{R}$-section $s'$ of the restricted bundle $\bb{E}^{\times}|_{X^{j-1}}$.
Let $Y\subset X$ be a free $j$-cell of $X$ with equivariant attaching map $\phi:\Z_2\times \n{D}^j\to Y\subset X$.
The pullback bundle $\phi^\ast(\bb{E}^{\times})\to \Z_2\times \n{D}^j$ has an $\rr{R}$-structure since $\phi$ is equivariant
and it is locally $\rr{R}$-trivial.
 The $\rr{R}$-section $s'$ defines an $\rr{R}$-section $\sigma':=s'\circ\phi$ on $\phi^\ast(\bb{E}^{\times})|_{\Z_2\times\partial \n{D}^j}$. 
Since $\n{D}^j$ is contractible we know that
$\phi^\ast(\bb{E})$ is trivial as a complex vector bundle. This implies that $\phi^\ast(\bb{E}^{\times})|_{\Z_2\times \n{D}^j}$ is isomorphic to $(\Z_2\times \n{D}^j)\times (\C^m)^\times$ with involution induced by certain anti-linear maps $\{ \pm 1 \} \times \C^m \to \{ \mp 1 \} \times \C^m$. Then, 
the $\rr{R}$-section $\sigma'$
 defined on $\Z_2\times\partial \n{D}^j$ can be  identified with an equivariant map  $\Z_2\times\partial \n{D}^j\to (\C^m)^\times$.
Because $j - 1 \le 2(m-1)$ and $\pi_{j-1}((\C^m)^\times) \simeq 0$, its restriction $\{ 1 \} \times \partial \n{D}^j \to (\C^m)^\times$ prolongs to a map $\{ 1 \} \times \n{D}^j \to (\C^m)^\times$, which, by the equivariant constraint again, prolongs to an equivariant map $\Z_2\times \n{D}^j \to (\C^m)^\times$. 
This prolonged map yields an $\rr{R}$-section $\sigma$ of $\phi^\ast(\bb{E}^{\times})$. Using the natural morphism $\hat{\phi}:\phi^\ast(\bb{E}^{\times})\to \bb{E}^{\times}$ over $\phi$ (defined by the pullback construction) we have a unique $\rr{R}$-section $s_Y$ of $\bb{E}^{\times}|_{\overline{Y}}$ defined by $\hat{\phi}\circ\sigma=s_Y\circ \phi$
such that $s_Y=s'$ on $X^{j-1}\cap \overline{Y}$. 
Now, one defines a global section $s$ of $\bb{E}^{\times}|_{X^{j-1} \cup Y}$
 by the requirements that $s|_{X^{j-1}}\equiv s'$  and $s|_Y \equiv s_Y$ for the free $j$-cell $Y$.   By the weak topology property of CW-complex, $s$ is also continuous. This argument applies to other free $j$-cells, and the claim is true on $X^j$ and eventually on $X^d = X$.
  \qed

\medskip

\begin{remark}{\upshape
The condition that the involutive space $(X,\tau)$ must have fixed cells only in dimension 0 cannot be relaxed. Indeed, as a matter of fact, in the case of trivial involution $(X,{\rm Id}_X)$ the fibers of an $\rr{R}$-bundle are real vector space (\cf Proposition \ref{prop:Rv1}) and 
 $\bb{E}^{\times}_x\simeq\R^m\setminus\{0\}$  is only $(m-2)$-connected.
}\hfill $\blacktriangleleft$
\end{remark}

\medskip

The next theorem generalizes equation \eqref{eq:stab_rank} for the case of \virg{Real} vector bundles.

\begin{theorem}[Stable range]
\label{theo:stab_ran_R}
Let $(X,\tau)$ be an involutive space such that $X$ has a finite $\Z_2$-CW-complex decomposition of dimension $d$ with fixed cells only in dimension 0. 
Each rank $m$ $\rr{R}$-vector bundle $(\bb{E},\Theta)$  over $(X,\tau)$ such that $d<2m$
splits as
\beql{eq:stab_rank_R}
\bb{E}\;\simeq\;\bb{E}_0\;\oplus\;(X\times\C^{m-\sigma})
\eeq
where $\bb{E}_0$ is an $\rr{R}$-vector bundle  over $(X,\tau)$, $X\times\C^{m-\sigma}\to X$ is the trivial $\rr{R}$-vector bundle  over $(X,\tau)$  and $\sigma:=[d/2]$ (here $[x]$ denotes the integer part of $x\in\R$).
\end{theorem}
\proof
By Proposition \ref{prop:stab_ran_R} there is a global $\rr{R}$-section $s:X\to\bb{E}$ such that $s(x)\neq0$
This section determines a monomorphism $f:X\times\C\to \bb{E}$ given by $f(x,a):=as(x)$. This mono-morphism is equivariant, \ie
$f(\tau(x),\bar{a})=\bar{a}s(\tau(x))=\Theta(as(x))$. Let $\bb{E}'$ be the \emph{coker} of $f$ in $\bb{E}$, namely
$\bb{E}'$ is the quotient of $\bb{E}$ by the following relation $p\sim p'$ if $p$ and $p'$ are in the same fiber of $\bb{E}$ and if $p-p'\in{\rm Im}(f)$. The map $\bb{E}'\to X$ is a vector bundle of rank $m-1$ (\cf \cite[Chapter 3, Corollary 8.3]{husemoller-94}) which inherits an $\rr{R}$-structure from $\bb{E}$ and the equivariance of $f$.
Since $X$ is  compact, by \cite[Chapter 3, Theorem 9.6]{husemoller-94} there is an isomorphism of $\rr{R}$-bundles between 
$\bb{E}$ and $\bb{E}'\;\oplus\;(X\times\C)$. If $d<2(m-1)$ one can repeat the argument for $\bb{E}'$ and iterating this procedure  one gets \eqref{eq:stab_rank_R}.
\qed

\medskip

As  immediate consequence we have the following important result:

\begin{corollary}
\label{corol_(iii)}
Let $(X,\tau)$ be an involutive space such that $X$ has a finite $\Z_2$-CW-complex decomposition of dimension $2\leqslant d\leqslant 3$  with fixed cells only in dimension 0. 
Then the condition ${\rm Vec}^1_{\rr{R}}(X, \tau)=0$ implies ${\rm Vec}^m_{\rr{R}}(X, \tau)=0$ for all $m\in\N$.
\end{corollary}
\proof
According to Theorem \ref{theo:stab_ran_R} each $\rr{R}$-vector bundle $(\bb{E},\Theta)$ splits   as the sum $\bb{E}_0\oplus (X\times\C^{m-1})$
and $\bb{E}_0$ has rank 1.
\qed

\medskip

We end this section with the proof of a useful result 
based on the notion of $\Z_2$-CW-complex decomposition
that is  repeatedly applied in this work.

\begin{lemma}[$\Z_2$-homotopy reduction]
\label{lemma:Z2-reduct}
Let $(X_1,\tau_1)$ and $(X_2,\tau_2)$ be two involutive spaces. Assume that: 
\begin{enumerate}
\item[a)] $X_1$ is a finite $\Z_2$-CW complex;
\item[b)] The fixed point set $X_2^{\tau_2}\subset X_2$
is non empty;
\item[c)] There exist $d_1,d_2$ such that 
$\pi_k(X_2^{\tau_2})=0$ for all $0\leqslant k\leqslant d_1$ and $\pi_k(X_2)=0$ for all $0\leqslant k\leqslant d_2$.
\end{enumerate}
Let $Y\subset X_1$ be a $\Z_2$-subcomplex which contains fixed cells of 
dimension at most $d_1$ and free cells of dimension at most $d_2$
Then, each equivariant map $\varphi:X_1\to X_2$ is
$\Z_2$-equivariantly homotopic to a map $\varphi'$ which has a constant value on $Y$, \ie $\varphi'(Y)=\ast\in X_2^{\tau_2}$.
\end{lemma}

\begin{proof}
Throughout the proof, we will write $c$ for the constant map with value $\ast \in X_2^{\tau_2}$ defined on any space. The lemma is proved by an inductive argument. 
For this aim, let $Y^0 \subset Y^1 \subset \ldots \subset Y^n = Y$ be a $\Z_2$-skeleton
decomposition of $Y$, so that $Y^{j+1}$ is given by attaching $(j+1)$-dimensional
$\Z_2$-cells to $Y^j$ and fix $n={\rm max}\{d_1,d_2\}$.

As  initial step we prove that there exists an equivariant homotopy from $\varphi$ to a map $\varphi_0$ such that $\varphi_0(Y^0) = \ast$.
The 0-skeleton of $Y$
 can be decomposed according to \eqref{eq:zero_skel}
in a disjoint union of fixed 0-cells ${\bf e}^0_j$ and free 0-cells $\tilde{\bf e}^0_j$. Because $\varphi|_{{\bf e}_j^0} \in X_2^{\tau_2}$ and $\pi_0(X_2^{\tau_2}) = 0$, there is a $\Z_2$-equivariant homotopy from $\varphi|_{{\bf e}_j^0}$ to $c$. If we express a free $\Z_2$-cell as $\tilde{\bf e}^0_j = \{ \pm 1 \} \times \n{D}^0_j$, then $\pi_0(X_2) = 0$ gives a homotopy from $\varphi|_{ \{ 1 \} \times \n{D}^0}$ to $c$. By means of the involution $\tau_2$, this homotopy extends to a  $\Z_2$-equivariant homotopy from $\varphi|_{\tilde{\bf e}^0_j}$ to $c$. Assembling these equivariant homotopies on $0$-cells together we get an equivariant homotopy from $\varphi|_{Y^0}$ to $c$. We can extend this to a $\Z_2$-equivariant homotopy from $\varphi$ to $\varphi_0 : X_1 \to X_2$ such that $\varphi_0(Y^0) = \ast$, applying the $\Z_2$-homotopy extension property \cite{matumoto-71} to the subcomplex $Y^0 \subset X_1$.

Now, let us assume that there exists an equivariant homotopy from $\varphi$ to $\varphi_j : X_1 \to X_2$ such that $\varphi_j(Y^j) = \ast$. By definition, $Y^{j+1}$ is obtained from $Y^j$ by attaching 
fixed and free $\Z_2$-cells of dimension $j+1$. Let ${\bf e}^{j+1}$ be a fixed cell and consider the  subcomplex $Y^j \cup {\bf e}^{j+1}$.
The restriction of $\varphi_j$ to ${\bf e}^{j+1}$ defines an element $[\varphi_j|_{{\bf e}^{j+1}}] \in \pi_{j+1}(X_2^{\tau_2})$. Then, if $j+1 \le d_1$ from the hypothesis it follows that $\pi_{j+1}(X_2^{\tau_2}) = 0$ and there is a $\Z_2$-equivariant homotopy from $\varphi_j|_{{\bf e}^{j+1}}$ to $c$. This defines an equivariant homotopy from $\varphi_j|_{Y^j \cup {\bf e}^{j+1}}$ to $c$. Applying the $\Z_2$-homotopy extension property, this homotopy further extends to an equivariant homotopy from $\varphi_j$ to a map $\varphi_{j}' : X_1 \to X_2$ such that $\varphi_{j}'(Y^j \cup {\bf e}^{j+1}) = \ast$.
This construction can be repeated verbatim adding to $Y^j$ all other fixed cells needed to build $Y^{j+1}$.
Now, we can deal with the case of a
  free cell $\tilde{{\bf e}}^{j+1} = \{ \pm 1 \} \times \n{D}^{j+1}$.
 Let ${Y^j}':=Y^j\cup{\bf e}^{j+1}\cup\ldots\cup {\bf e}^{j+1}$ be
  the $j$-skeleton plus all the fixed cells of the $j+1$ skeleton.
  Let us consider the subcomplex ${Y^j}'\cup \tilde{{\bf e}}^{j+1}$. 
  The restriction of $\varphi'_j$ to $\{ 1 \} \times \n{D}^{j+1}$ defines an element $[\varphi'_j|_{\{ 1 \} \times \n{D}^{j+1}}] \in \pi_{j+1}(X_2)$. Hence,
  if $j+1 \le d_2$ from the hypothesis it follows that $\pi_{j+1}(X_2) = 0$ and there is a homotopy from $\varphi'_j|_{\{ 1 \} \times \n{D}^{j+1}}$ to $c$.
   We use the involution $\tau_2$ to extend this homotopy to an equivariant homotopy from $\varphi'_j|_{\tilde{\bf e}^{j+1}}$ to $c$
 and consequently we get an equivariant homotopy between  $\varphi'_j|_{{Y^j}'\cup \tilde{{\bf e}}^{j+1}}$ to $c$. Certainly 
we can repeat this construction for all the other $(j+1)$-dimensional free cells in $Y^{j+1}$ showing that there exists an equivariant homotopy between 
$\varphi'_j|_{Y^{j+1}}$ and $c$. By means of the $\Z_2$-homotopy extension property adapted to $Y^{j+1} \subset X_1$
we can construct an equivariant homotopy from $\varphi'_j$ to $\varphi_{j+1} : X_1 \to X_2$ such that $\varphi_{j+1}(Y^{j+1}) = \ast$.

An iterated application of this construction gives  finally an equivariant homotopy from $\varphi$ to $\varphi' := \varphi_n$ such that $\varphi'(Y) = \ast$.
\end{proof}

\section{Cohomology classification for \virg{Real} vector bundles}
\label{sec:coomol_class-Rbund}
Chern classes are an important tool for the study of complex vector bundles. In a similar way, for $\rr{R}$-vector bundles one can define a proper theory of characteristic classes: the \emph{\virg{Real} Chern classes}. This cohomology theory has been introduced for the first time by B. Kane in \cite{kahn-59} and subsequently studied in \cite{krasnov-92,pitsch-scherer-13}.

\subsection{Equivariant cohomology and Borel construction}
\label{ssec:borel_constr}
In this section we recall the 
 \emph{Borel construction} for the equivariant cohomology
 \cite{borel-60}.  For a more technical introduction to this subject  we refer to
\cite[Chapter 3]{hsiang-75} and \cite[Chapter 1]{allday-puppe-93}.

\medskip

The \emph{homotopy quotient} of an involutive space   $(X,\tau)$ is the orbit space
$$
{X}_{\sim\tau}\;:=\;X\times\hat{\n{S}}^\infty /( \tau\times \vartheta)\;.
$$
Here $\vartheta$ is the antipodal map on the infinite sphere $\n{S}^\infty$ (\cf Example \ref{ex:sphere_antipodal}).
We point out that the product space $X\times{\n{S}}^\infty$ (forgetting for a moment the $\Z_2$-action) has the \emph{same} homotopy type of $X$ 
since $\n{S}^\infty$ is contractible. Moreover, since $\vartheta$ is a free involution,  also the composed involution $\tau\times\vartheta$ is free, independently of $\tau$.
Let $\s{R}$ be any commutative ring (\eg, $\R,\Z,\Z_2,\ldots$). The \emph{equivariant cohomology} ring 
of $(X,\tau)$
with coefficients
in $\s{R}$ is defined as
$$
H^\bullet_{\Z_2}(X,\s{R})\;:=\; H^\bullet({X}_{\sim\tau},\s{R})\;.
$$
More precisely each equivariant cohomology group $H^j_{\Z_2}(X,\s{R})$ is given by the
 singular cohomology group  $H^j({X}_{\sim\tau},\s{R})$ of the  homotopy quotient ${X}_{\sim\tau}$ with coefficients in $\s{R}$ and the ring structure is given as usual by the {cup product}.

\medskip

Since the $\Z_2$-action  is free,  the canonical  projection 
$
X\times \hat{\n{S}}^\infty /( \tau\times \vartheta)\to\hat{\n{S}}^\infty / \vartheta\simeq \R P^\infty
$
has fiber  $X$ and the sequence
\begin{equation}\label{eq:ex_hom1}
X\;\hookrightarrow\;{X}_{\sim\tau}\;\longrightarrow\;\R P^\infty
\end{equation}
is  known as \emph{canonical} fibration. The projection ${X}_{\sim\tau}\to\R P^\infty$ makes $H^\bullet_{\Z_2}(X,\s{R})$ into a $\s{K}_{\s{R}}$-module where
\beql{eq:KR-coeffic}
\s{K}_{\s{R}}\;:=\;H^\bullet_{\Z_2}(\{\ast\},\s{R})\simeq H^\bullet(\R P^\infty,\s{R})
\eeq
is  the \emph{coefficient ring} associated to the group action of $\Z_2$. 
A space $X$ with a $\Z_2$-action is called \emph{equivariantly formal} if the 
 associated equivariant cohomology ring with coefficient in $\s{R}$ can be expressed as
\begin{equation}\label{eq:eq_torus_1}
H^\bullet_{\Z_2}(X,\s{R})\; =\; H^\bullet(X,\s{R})\;\otimes \;\s{K}_{\s{R}}\;. 
\end{equation}
The cases $\s{R}=\Z$ and $\s{R}=\Z_2$ are of particular interest for us. From \cite[Proposition 1.3.1 and Proposition 1.3.2]{allday-puppe-93} one has
\begin{equation}\label{eq_proj_hom}
H^k(\R P^\infty,\Z)\;\simeq\;
\left\{
\begin{aligned}
&\Z&&\text{if}\ \ k=0\\
&\{0\}&&\text{if}\ \ k\ \ \text{odd}\\
&\Z_2&&\text{if}\ \ k\ \ \text{even}
\end{aligned}
\right.\qquad\text{and}\qquad H^k(\R P^\infty,\Z_2)\;\simeq\;\Z_2\;,\quad\forall\; k\geqslant0\;\;.
\end{equation}
In particular, the ring structures induced by the cup product are described by
$$
\s{K}_{\Z}\;:=\;H^\bullet(\R P^\infty,\Z)\; =\; \Z[\kappa]/(2\kappa)\;\qquad\text{and}\qquad \s{K}_{\Z_2}\;:=\;H^\bullet(\R P^\infty,\Z_2)\; =\; \Z_2[\xi]
$$
where 
$\kappa\in H^2(\R P^\infty,\Z)$ such that $2\kappa=0$ and $\xi\in H^1(\R P^\infty,\Z_2)$.

\medskip

This {canonical fibration} \eqref{eq:ex_hom1} induces a long exact sequence of homotopy groups:
\begin{equation}\label{eq:ex_hom2}
\ldots \;\longrightarrow\;\pi_k(X)\;\longrightarrow\;\pi_k({X}_{\sim\tau})\;\longrightarrow\;\pi_k(\R P^\infty)\;\longrightarrow\;
\pi_{k-1}(X)\;\longrightarrow\;\ldots\;.
\end{equation}
Since $\pi_1(\R P^\infty)=\Z_2$ and $\pi_k(\R P^\infty)=0$ if $k\geqslant 2$ one has
$\pi_k(X)=\pi_k({X}_{\sim\tau})$ for $k\geqslant 2$ and
\begin{equation}\label{eq:ex_hom3}
1\;\longrightarrow\;\pi_1(X)\;\stackrel{\imath}{\longrightarrow}\;\pi_1({X}_{\sim\tau})\;\stackrel{w}{\longrightarrow}\;\Z_2\;\longrightarrow\;1\;.
\end{equation}
Hence, $\pi_1({X}_{\sim\tau})$ is a \emph{group extension} of $\Z_2$ by $\pi_1(X)$.

\medskip

As the coefficients of
the usual singular cohomology are generalized to \emph{local coefficients} (see \eg \cite[Section 3.H]{hatcher-02} or
\cite[Section 5]{davis-kirk-01}), the coefficients of Borel's equivariant cohomology are also
generalized to local coefficients. Given an involutive space $(X,\tau)$ let us consider the homotopy group $\pi_1({X}_{\sim\tau})$
and the associated  \emph{group ring} $\Z[\pi_1({X}_{\sim\tau})]$. Each module $\s{Z}$ over the group $\Z[\pi_1({X}_{\sim\tau})]$ is, by definition,
a \emph{local system} on $X_{\sim\tau}$.  Using this local system one defines, as usual, the equivariant cohomology with local coefficients in $\s{Z}$:
$$
H^\bullet_{\Z_2}(X,\s{Z})\;:=\; H^\bullet({X}_{\sim\tau},\s{Z})\;.
$$
We are particularly interested in modules $\s{Z}$ whose underlying groups are identifiable with $\Z$.
In
this case, the module structure on $\s{Z}$ defines a homomorphism of groups $\pi_1({X}_{\sim\tau})\to{\rm Aut}_\Z(\Z)=\Z_2$, and vice verse. Such homomorphisms are in one to one correspondence
with elements in $H^1_{\Z_2}(X,\Z_2):=H^1({X}_{\sim\tau},\Z_2)\simeq{\rm Hom}(\pi_1({X}_{\sim\tau}),\Z_2)$
(see the proof of \cite[Corollary 5.6]{davis-kirk-01}). Thus, we can  identify each local system $\s{Z}$ (with underlying group $\Z$) with an
element $w\in H^1_{\Z_2}(X,\Z_2)$. This group classifies also \emph{$\Z_2$-equivariant} real line bundles.

\medskip

On a space  $(X,\tau)$, there always exists a particular local system $\Z(m)$
for each $m\in\Z$. This  is defined
to be the group $\Z$ on that $\pi_1(X_{\sim\tau})$ acts by
the natural homomorphism
in the homotopy exact sequence \eqref{eq:ex_hom3}. More precisely,  one has that 
$w[\gamma]$ is the identity over $\Z$ if $\gamma=\imath(\gamma')$ for a loop $\gamma'\in\pi_1(X)$, otherwise
$w[\gamma]$ is the multiplication  $(-1)^m:\Z\to\Z$. Similarly, $\Z(m)$ can be identified with the {$\Z_2$-equivariant line bundle}  $X\times\R\to X$ with $\Z_2$-action $(x,r)\mapsto(\tau(x),(-1)^mr)$.
Because the module structure depends only on the parity of $m$, we consider only the $\Z_2$-modules ${\Z}(0)$ and ${\Z}(1)$. Since ${\Z}(0)$ corresponds to the case of the trivial action of $\pi_1(X_{\sim\tau})$ on $\Z$ (equivalently to the trivial element in $H^1_{\Z_2}(X,\Z_2)$) one has $H^k_{\Z_2}(X,\Z(0))\simeq H^k_{\Z_2}(X,\Z)$ \cite[Section 5.2]{davis-kirk-01}. An important ingredient for the calculation of the equivariant cohomology with local coefficients $\Z(1)$ is the determination of the {coefficient modules} $\s{K}_{\Z(1)}\;:=\;H^\bullet_{\Z_2}(\{\ast\},\Z(1))$. In this case the local system is associated with a non-trivial  element of $H^1_{\Z_2}(\{\ast\},\Z_2)=H^1(\R P^\infty,\Z_2)$ called the \emph{orientation character} of $\R P^\infty$ and 
one has \cite[Chapter II, Example 11.3]{bredon-97}
\beql{eq:cohomPR_Z1}
H^k\big(\R P^\infty,\Z(1)\big)\;\simeq\;
\left\{
\begin{aligned}
&\{0\}&&\text{if}\ \ k=0\ \ \text{or}\ \ k\ \ \text{even}\\
&\Z_2&&\text{if}\ \ k\ \ \text{odd}\\
\end{aligned}
\right.
\eeq
with module structure (over $\s{K}_{\Z}$) given by 
\beql{eq:KR-coefficZ1}
\s{K}_{\Z(1)}\;:=\; H^\bullet(\R P^\infty,\Z(1))\;=\;\alpha_1\ \Z[\alpha_2]/(2\alpha_1,2\alpha_2)\;,
\eeq
where $\alpha_j\in H^j(\R P^\infty,\Z(1))$ such that $2\alpha_j=0$, $j=1,2$. 

\medskip

The following recursive formulas are extremely useful for the computation of equivariant cohomology groups: 
\beql{eq:recurs01}
\begin{aligned}
H^k_{\Z_2}(X\times\tilde{\n{S}}^1,\Z(1))&\;\simeq\;H^k_{\Z_2}(X,\Z(1))\;\oplus\; H^{k-1}_{\Z_2}(X,\Z) \\
H^k_{\Z_2}(X\times\tilde{\n{S}}^1,\Z)&\;\simeq\;H^k_{\Z_2}(X,\Z)\;\oplus\; H^{k-1}_{\Z_2}(X,\Z(1))\;.
\end{aligned}
\eeq
Here, $X$ is any space  with a $\Z_2$-involution and $\tilde{\n{S}}^1$ is the circle with  TR-involution as defined in Example \ref{ex:homot_spher}. These formulas can be obtained by a suitable application of the \emph{Gysin sequence}. The details of the proof can be found in \cite[Section 2.6]{gomi-13}. We point out that under the conditions stated in Assumption \ref{ass:1} the isomorphisms \eqref{eq:recurs01}  are valid for all $k\in\Z$
with $H^k_{\Z_2}(X,\Z)=H^k_{\Z_2}(X,\Z(1))=0$ for all $k<0$,
$H^0_{\Z_2}(X,\Z)=\Z$ (one path connected component) and $H^0_{\Z_2}(X,\Z(1))=0$.

\medskip

Because of its definition via  Borel construction, the cohomology ring $H^\bullet_{\Z_2}(X,\Z(m))$ defines a $\Z_2$-equivariant cohomology theory  \cite[Section 2.4]{gomi-13}. This means that the homotopy, excision, exactness and additivity axioms are satisfied. There also exists the corresponding reduced theory $\tilde{H}^\bullet_{\Z_2}(X,\Z(m))$ defined as the kernel of the homomorphism
${H}^\bullet_{\Z_2}(X,\Z(m))\to {H}^\bullet_{\Z_2}(\{\ast\},\Z(m))$ induced from the inclusion $\ast\to X$ where $\ast$ is a fixed point of $(X,\tau)$. Then one has the usual decomposition
\beql{eq:susp1}
{H}^k_{\Z_2}(X,\Z(m))\;\simeq\;\tilde{H}^k_{\Z_2}(X,\Z(m))\;\oplus\; {H}^k_{\Z_2}(\{\ast\},\Z(m))\;.
\eeq
Let us denote with $\tilde{I}:=[-1,1]$ the involutive space endowed with the $\Z_2$-action $x\mapsto -x$. The space
\beql{eq:susp2}
\tilde{\Sigma}X\;:=\;(X\times\tilde{I})/(X\times \partial\tilde{I}\;\cup\;\{\ast\}\times\tilde{I})\;,
\eeq
called 
{equivariant} reduced suspension, has a well defined $\Z_2$-action and the \emph{equivariant suspension formula}
\beql{eq:susp3}
\tilde{H}^k_{\Z_2}(\tilde{\Sigma}X,\Z(m))\;\simeq\;\tilde{H}^{k-1}_{\Z_2}(X,\Z(m-1))
\eeq
holds true.

\subsection{The Picard group of \virg{Real} line bundles}\label{ssect:Picard-R}
As for the complex case (\cf Section \ref{ssec:comp-Picard}), 
 the tensor product induces a group structure  on the set ${\rm Vec}^1_\rr{R}(X,\tau)$ \cite[Section 2]{edelson-71_2}. Moreover,  ${\rm Vec}^1_\rr{R}(X,\tau)$ can be identified with ${\rm Pic}(X,\tau)$, the \emph{Picard group} of isomorphism classes of $\tau$-equivariant {invertible sheaves} on $X$. This leads to a classification of ${\rm Vec}^1_\rr{R}(X,\tau)$ of the type \eqref{eq:clasLin_bun}.  

\begin{theorem}[Classification of  $\rr{R}$-line bundles \cite{kahn-59,krasnov-92}]
\label{theo:clasR-lin}
Let $(X,\tau)$ an involutive space with $X$ that verifies Assumption \ref{ass:1}. There is an isomorphism of groups
\begin{equation}\label{eq:clasLin_bun2}
\tilde{c}_1\;:\; {\rm Vec}^1_\rr{R}(X,\tau)\;\longrightarrow\;H^2_{\Z_2}\big(X,\Z(1)\big)
\end{equation}
and the map $\tilde{c}_1$ is called first \emph{\virg{Real} Chern class}.
\end{theorem}
We do not give the details of the proof of this theorem that was firstly  proved in \cite{kahn-59} and \cite{krasnov-92}. A more detailed (and modern) proof can be found in \cite[Appendix A]{gomi-13}.
Nevertheless, we point out that the strategy of the proof is very close to the analogous result
 for  complex line bundles. Also in this case one first proves 
${\rm Vec}^1_\rr{R}(X,\tau)\simeq \check{H}^1(X;\Z_2,\bb{O}_\C^*)
$ where the right-hand side
is the 
\emph{equivariant  Galois-Grothendieck cohomology} group with coefficients in the  {invertible sheaf} $\bb{O}^*_\C$ over $X$ (\cf \cite[Proposition 1.1.1]{krasnov-92}). This sheaf supports a canonical $\Z_2$-action associated with $\tau$ given by $f\mapsto \overline{f\circ \tau}$ for each  $f\in \bb{O}^*_\C$. Following \cite[Example 3.4]{bredon-97}, let $\underline{\Z}:=X\times \Z$ be the  \emph{constant integral} sheaf over $X$ with stalks $\Z$. For each integer $m\in\Z$ one can define a \emph{twist} of this sheaf by the $\Z_2$-action  $(x,n)\mapsto(\tau(x),(-1)^m n)$ induced by $\tau$. This produces a  $\Z_2$-sheaves with stalks $\Z$  over  the involutive space $(X,\tau)$
denoted by $\underline{\Z}(m)$.
The exact sequence of $\Z_2$-sheaves 
 $0\to\underline{\Z}(1)\to \bb{O}_\C\to\bb{O}_\C^\ast\to1$
given by the exponential map induces an isomorphism
$\check{H}^1(X;\Z_2,\bb{O}_\C^*)\simeq\check{H}^2(X;\Z_2,\underline{\Z}(1))
$. The last step is to prove   the equivalence between the  
Galois-Grothendieck cohomology group $\check{H}^2(X;\Z_2,\underline{\Z}(1))$ and the equivariant cohomology with local coefficients $H^2_{\rm eq}(X,\Z(1))$ defined in Section \ref{ssec:borel_constr}. This identification is explained in detail in \cite{kahn-59} and  \cite{gomi-13}.

\subsection{\virg{Real} line bundle over TR-tori}
\label{sect:eq_cohom_tor}
The aim of this section is to classify ${\rm Vec}^1_\rr{R}(\T^d,\tau)$. In view of Theorem \ref{theo:clasR-lin} we can achieve this goal by calculating $H^2_{\Z_2}(\T^d,\Z(1))$.

\medskip

In the one-dimensional case $(\T^1,\tau)=(\n{S}^1,\tau)\equiv \tilde{\n{S}}^1$,  replacing $\{\ast\}\times\tilde{\n{S}}^1\simeq\tilde{\n{S}}^1$ in the recursive equations \eqref{eq:recurs01} one obtains
\begin{align}
H^k_{\Z_2}(\tilde{\n{S}}^1,\Z(1))&\;\simeq\;H^k(\R P^\infty,\Z(1))\;\oplus\; H^{k-1}(\R P^\infty,\Z)\label{eq:recurs1b}\\
H^k_{\Z_2}(\tilde{\n{S}}^1,\Z)&\;\simeq\;H^k(\R P^\infty,\Z)\;\oplus\; H^{k-1}(\R P^\infty,\Z(1))\label{eq:recurs2b}\;
\end{align}
where we used $H^k_{\Z_2}(\{\ast\},\Z)=H^k(\R P^\infty,\Z)$ and $H^k_{\Z_2}(\{\ast\},\Z(1))=H^k(\R P^\infty,\Z(1))$. From \eqref{eq_proj_hom} and \eqref{eq:cohomPR_Z1} one gets
\beql{eq:equi_cohom1dA}
H^k_{\Z_2}(\tilde{\n{S}}^1,\Z)\;\simeq\;
\left\{
\begin{aligned}
&\Z&&\text{if}\ \ k=0\\
&\{0\}&&\text{if}\ \ k\ \ \text{odd}\\
&\Z_2\oplus\Z_2&&\text{if}\ \ k\ \ \text{even}
\end{aligned}
\right.
\eeq
and 

\beql{eq:equi_cohom1dB}
H^k_{\Z_2}(\tilde{\n{S}}^1,\Z(1))\;\simeq\;
\left\{
\begin{aligned}
&\Z\oplus\Z_2&&\text{if}\ \ k=1\\
&\Z_2\oplus\Z_2&&\text{if}\ \ k\ \ \text{odd}\ \ k>1\\
&\{0\}&&\text{if}\ \ k=0\ \ \text{or}\ \ k\ \ \text{even}
\end{aligned}
\right.
\eeq
\begin{remark}\label{equi_cohom1d}{\upshape
The cohomology groups in \eqref{eq:equi_cohom1dA} and \eqref{eq:equi_cohom1dB} can be computed also 
directly without the use of the recursive formulas \eqref{eq:recurs01}. In fact the {homotopy quotient} $\n{S}^1_{\sim\tau}$ can be explicitly computed: 
\beql{eq:homo_quotS1}
\n{S}^1_{\sim\tau}\;:=\;\tilde{\n{S}}^1\times\hat{\n{S}}^\infty /( \tau\times \vartheta)\;\simeq\;\R P^\infty\;\vee\;\R P^\infty
\eeq
where $\vee$ denotes the \emph{wedge sum} of topological spaces. More explicitly,  the {homotopy quotient} $\n{S}^1_{\sim\tau}$
 is topologically isomorphic to two copies of
$\R P^\infty$ glued together at a common point identifiable with their zero dimensional cell. The identification \eqref{eq:homo_quotS1}
 can be proved in two different ways. The first way  uses the topological identification $\tilde{\n{S}}^1\times\hat{\n{S}}^n /( \tau\times \vartheta)\simeq \R P^{n+1}\;\sharp\; \R P^{n+1}$ (see \cite{jahren-kwasik-11} and references therein). Here $\sharp$ denotes the \emph{connected sum}, namely the two copies of $\R P^{n+1}$ are identified along a copy of  $\n{S}^n$. When $n\to \infty$ the identification becomes along a copy of $\n{S}^\infty$ which is contractible to a point. The second way is to consider the canonical fibration
\eqref{eq:ex_hom1} which says that  $\n{S}_{\sim\tau}\to \R P^\infty$ is a bundle with typical fiber $\tilde{\n{S}}^1$. We can think of $\tilde{\n{S}}^1$ as the union of the two $\tau$-invariant
 hemispheres 
$\ell'_\pm:=\{(k_0,k_1)\in\n{S}^1\ |\ \pm k_0\geqslant0\}$, glued together at $p_-=(0,-1)$ and $p_+=(0,+1)$.
Let us consider the  subbundles $\s{L}_\pm\to \R P^\infty$ with fibers given by the respective intervals $\ell'_\pm$. Since the fibers are contractible one has 
 $\s{L}_\pm\simeq \R P^\infty$. The subbundle $\s{I}\to \R P^\infty$ with fibers given by the respective intersections  $\ell'_-\cap\ell'_+=\{p_-,p_+\}$ is a $\Z_2$-bundle over $\R P^\infty$. Thus, $\s{I}$ has to be contractible by definition since $\R P^\infty$ is the classifying space for the cyclic group $\Z_2$. Therefore, $\n{S}^1_{\sim\tau}$ is topologically equivalent to the two copies $\s{L}_\pm\simeq \R P^\infty$ identified along $\s{I}\simeq\{\ast\}$. Once formula \eqref{eq:homo_quotS1} has been established, one can compute $H^k_{\Z_2}(\tilde{\n{S}}^1,\Z)=H^k(\R P^\infty\vee\R P^\infty,\Z)$ using \cite[Corollary 2.25]{hatcher-02}. This gives rise to
 the same result displayed in equation \eqref{eq:equi_cohom1dA}. The computation of the cohomology with local coefficients $\Z(1)$ can be done with the application of the \emph{Poincar\'{e} duality} and of the \emph{universal coefficient Theorem} \cite[Section 5.1]{jahren-kwasik-11} and produces the same of  \eqref{eq:equi_cohom1dB}.
 }\hfill $\blacktriangleleft$
\end{remark}

Now we are in position to provide the classification of $\rr{R}$-line bundles over $(\T^d,\tau)$.

\begin{proposition}[Classification of $\rr{R}$-line bundle over TR-tori]
\label{prob:R_linT}
Let $(\T^d,\tau)$ be  the involutive space described in 
Definition \ref{def:per_ferm}. Then
$$
{\rm Vec}^1_\rr{R}(\T^d,\tau)\;=\; 0
$$
in any dimension $d\in\N$.
\end{proposition}
\proof
We use the short notation $\tilde{\n{T}}^d\equiv(\T^d,\tau)$ introduced in Example \ref{ex:tori_inv}.
Since Theorem \ref{theo:clasR-lin}, it suffices to show that $H^2_{\Z_2}(\tilde{\n{T}}^d,\Z(1))=0$ for all $d$. Equation \eqref{eq:equi_cohom1dB} shows that the claim is true for $d=1$. Assume that $H^2_{\Z_2}(\tilde{\n{T}}^j,\Z(1))=0$ for all $1\leqslant j\leqslant d$. Since $\tilde{\n{T}}^{j+1}=\tilde{\n{T}}^j\times\tilde{\n{S}}^1$ we can use the recursive formulas \eqref{eq:recurs01} which provide
$$
\begin{aligned}
H^2_{\Z_2}(\tilde{\n{T}}^{d+1},\Z(1))\;=\;H^{1}_{\Z_2}(\tilde{\n{T}}^{d},\Z) \;&=\;H^{1}_{\Z_2}(\tilde{\n{T}}^{d-1},\Z)\;\oplus\; H^{0}_{\Z_2}(\tilde{\n{T}}^{d-1},\Z(1))\\
 &=\;H^{1}_{\Z_2}(\tilde{\n{T}}^{d-2},\Z)\;\oplus\; H^{0}_{\Z_2}(\tilde{\n{T}}^{d-2},\Z(1))\;\oplus\; H^{0}_{\Z_2}(\tilde{\n{T}}^{d-1},\Z(1))\\
 &=\;\ldots\\
 &=\;H^{1}_{\Z_2}(\tilde{\n{S}}^{1},\Z)\;\oplus\left(\bigoplus_{j=1}^{d-1}H^{0}_{\Z_2}(\tilde{\n{T}}^{j},\Z(1))\right)\;=\;0
\end{aligned}
$$
since
$H^{1}_{\Z_2}(\tilde{\n{S}}^{1},\Z)=0$ (see equation \eqref{eq:equi_cohom1dA}) and $H^0_{\Z_2}(\tilde{\n{T}}^{j},\Z(1))=0$ for all $j\in\N$. The last fact follows again from  
 \eqref{eq:recurs01} which implies  
$H^0_{\Z_2}(\tilde{\n{T}}^{j+1},\Z(1))=H^0_{\Z_2}(\tilde{\n{T}}^{j},\Z(1))$ and equation \eqref{eq:equi_cohom1dB} where   $H^0_{\Z_2}(\tilde{\n{S}}^{1},\Z(1))=0$.
\qed

\begin{remark}\label{triv_lin_bun}{\upshape
Proposition \ref{prob:R_linT} generalizes a previous result by B.~Helffer and J.~Sj\"{o}strand. In \cite[Lemme 1.1]{helffer-sostrand-90} the authors proved the existence of a \emph{single} global $\rr{R}$-section under the assumption  of time reversal symmetry.
 }\hfill $\blacktriangleleft$
\end{remark}

\begin{remark}\label{even_odd_torus}{\upshape
With an inductive argument similar to that in the proof of  Proposition \ref{prob:R_linT} one can prove also
\beql{eq:even_odd_torus}
H^{\rm even}_{\Z_2}(\tilde{\n{T}}^d,\Z(1))\;\simeq\;0\;,\qquad\quad H^{\rm odd}_{\Z_2}(\tilde{\n{T}}^d,\Z)\;\simeq\;0
\eeq
for all $d\in\N$. When $d=1$ the relations \eqref{eq:even_odd_torus} agree with  \eqref{eq:equi_cohom1dA} and \eqref{eq:equi_cohom1dB}. When $d>1$ the \eqref{eq:even_odd_torus} can be proved iterating equations \eqref{eq:recurs01}.
 }\hfill $\blacktriangleleft$
\end{remark}

\subsection{\virg{Real} line bundle over TR-spheres}\label{sect:eq_cohom_sphe}
In Example \ref{ex:homot_spher} we explained that the  space $\tilde{\n{S}}^d\equiv(\n{S}^d,\tau)$ can be identified with  the {suspension} of $\hat{\n{S}}^{d-1}$
along the $\tau$-invariant direction $k_0\in[-1,1]$. Since the antipodal action on  $\hat{\n{S}}^{d-1}$ acts freely, the absence of fixed points prevents the possibility to identify $\tilde{\n{S}}^d$ with the reduced suspension  of 
$\hat{\n{S}}^{d-1}$ as $\Z_2$-spaces. On the other hand, $\tilde{\n{S}}^d$ can be see also as the {suspension} of $\tilde{\n{S}}^{d-1}$
along the direction $k_d\in\tilde{I}$ where $\tilde{I}:=[-1,1]$
inherits the $\Z_2$-action $k_d\mapsto -k_d$. Since $\tilde{\n{S}}^{d-1}$ has two fixed points, the equivariant reduced suspension \eqref{eq:susp2} makes sense and one has the topological equivalence
\beql{eq:sphere_susp}
\tilde{\n{S}}^{d}\;\simeq\;\tilde{\Sigma}\tilde{\n{S}}^{d-1}\;:=\; \big(\tilde{\n{S}}^{d-1}\times\tilde{I}\big)/\big(\tilde{\n{S}}^{d-1}\times \partial\tilde{I}\;\cup\;\{\ast\}\times\tilde{I}\big)
\eeq
as spaces endowed with the TR-involution $\tau$. 
Equation \eqref{eq:sphere_susp} is a key formula for the computation of the equivariant cohomology groups for spheres with TR-involution.
\begin{lemma}
Let $\tilde{\n{S}}^d\equiv(\n{S}^d,\tau)$ be the $d$ dimensional sphere with the TR-involution described in Definition \ref{def:free_ferm}. Then:
\begin{itemize}
\item  If $d$ is even and for all $k\in\Z$
\begin{align}
H^k_{\Z_2}(\tilde{\n{S}}^d,\Z(1))&\;\simeq\;H^k_{\Z_2}(\{\ast\},\Z(1))\;\oplus\; H^{k-d}_{\Z_2}(\{\ast\},\Z(1)) \label{eq:recurs1ev}\\
H^k_{\Z_2}(\tilde{\n{S}}^d,\Z)&\;\simeq\;H^k_{\Z_2}(\{\ast\},\Z)\;\oplus\; H^{k-d}_{\Z_2}(\{\ast\},\Z)\label{eq:recurs2ev}\;
\end{align}

\item  If $d$ is odd and for all $k\in\Z$
\begin{align}
H^k_{\Z_2}(\tilde{\n{S}}^d,\Z(1))&\;\simeq\;H^k_{\Z_2}(\{\ast\},\Z(1))\;\oplus\; H^{k-d}_{\Z_2}(\{\ast\},\Z) \label{eq:recurs1odd}\\
H^k_{\Z_2}(\tilde{\n{S}}^d,\Z)&\;\simeq\;H^k_{\Z_2}(\{\ast\},\Z)\;\oplus\; H^{k-d}_{\Z_2}(\{\ast\},\Z(1))\label{eq:recurs2odd}\;.
\end{align}
\end{itemize}
Here $\{\ast\}$ is any fixed point of $\tilde{\n{S}}^d$.
\end{lemma}
\proof
After iterating equation \eqref{eq:sphere_susp} in the  {equivariant suspension formula} \eqref{eq:susp3} one gets
$$
H^k_{\Z_2}(\tilde{\n{S}}^d,\Z(m))\;\simeq\;H^k_{\Z_2}(\{\ast\},\Z(m))\;\oplus\;\tilde{H}^{k-d+1}_{\Z_2}(\tilde{\n{S}}^1,\Z(m-d+1))\;.
$$
The isomorphism $\tilde{H}^{k-d+1}_{\Z_2}(\tilde{\n{S}}^1,\Z(m-d+1))\simeq {H}^{k-d}_{\Z_2}(\{\ast\},\Z(m-d))$, which follows from \eqref{eq:recurs1b} and \eqref{eq:recurs2b}, concludes the proof.
\qed

\medskip

Since $H^k_{\Z_2}(\{\ast\},\Z)=H^k(\R P^\infty,\Z)$ and $H^k_{\Z_2}(\{\ast\},\Z(1))=H^k(\R P^\infty,\Z(1))$ one can combine equations
\eqref{eq:recurs1ev} - \eqref{eq:recurs2odd} with \eqref{eq_proj_hom} and \eqref{eq:cohomPR_Z1} to obtain  a complete description for the equivariant cohomology of $\tilde{\n{S}}^d\equiv(\n{S}^d,\tau)$. More in detail one has
\beql{eq:equi_cohom_spher_d_int}
H^{\rm even}_{\Z_2}(\tilde{\n{S}}^d,\Z)\;\simeq\;
\left\{
\begin{aligned}
&\Z&&\text{if}\ \ k=0\\
&\Z_2\oplus\Z&&\text{if}\ \ k=d\\
&\Z_2&&\text{if}\ \  k< d\\
&\Z_2\oplus \Z_2&&\text{if}\ \ k>d\\
\end{aligned}
\right.\;,
\qquad\quad H^{\rm odd}_{\Z_2}(\tilde{\n{S}}^d,\Z)\;\simeq\;0
\eeq
for the cohomology with integer coefficients and

\beql{eq:equi_cohom_spher_d_loc}
H^{\rm even}_{\Z_2}(\tilde{\n{S}}^d,\Z(1))\;\simeq\;0\qquad\quad H^{\rm odd}_{\Z_2}(\tilde{\n{S}}^d,\Z(1))\;\simeq\;
\left\{
\begin{aligned}
&\Z_2\oplus\Z&&\text{if}\ \ k=d\\
&\Z_2&&\text{if}\ \  k< d\\
&\Z_2\oplus \Z_2&&\text{if}\ \ k>d\
\end{aligned}
\right.\;
 \eeq
when the coefficients are twisted.
The immediate consequence is:

\begin{proposition}[Classification of $\rr{R}$-line bundle over TR-spheres]
\label{prob:R_linS}
Let $(\n{S}^d,\tau)$ be  the involutive space described in 
Definition \ref{def:free_ferm}. Then
$$
{\rm Vec}^1_\rr{R}(\n{S}^d,\tau)\;=\; 0
$$
in any dimension $d\in\N$.
\end{proposition}
\proof
We need only to use the first in \eqref{eq:equi_cohom_spher_d_loc} and 
${\rm Vec}^1_\rr{R}(\n{S}^d,\tau)\simeq H^2_{\Z_2}(\tilde{\n{S}}^d,\Z(1))$ from Theorem \ref{theo:clasR-lin}.
\qed

\subsection{The $4$-dimensional case}
\label{sect:4-d}
The first observation for the classification of \virg{Real} vector bundles over $\tilde{\n{T}}^4\equiv(\n{T}^4,\tau)$ or $\tilde{\n{S}}^4\equiv(\n{S}^4,\tau)$ and with rank bigger than $m\geqslant 2$ (the line bundle case has already been discussed in Section \ref{sect:eq_cohom_tor} and Section \ref{sect:eq_cohom_sphe})
is a consequence of 
Theorem \ref{theo:stab_ran_R} which implies that
\beql{eq:R_vec_m_4_01}
{\rm Vec}^m_{\rr{R}}(\n{T}^4,\tau)\;\simeq\; {\rm Vec}^2_{\rr{R}}(\n{T}^4,\tau)\;,\quad\quad{\rm Vec}^m_{\rr{R}}(\n{S}^4,\tau)\;\simeq\; {\rm Vec}^2_{\rr{R}}(\n{S}^4,\tau)\;,\qquad\forall\ m\geqslant 2\;;
\eeq
namely, we need to consider only the case of rank 2 \virg{Real} vector bundles.

\medskip

To classify  $\rr{R}$-bundles over the sphere  we can invoke Proposition \ref{prop:eq_clut_constr} which provides
\beql{eq:R_vec_m_4_02}
{\rm Vec}^m_{\rr{R}}(\n{S}^4,\tau)\;\simeq\; {\rm Vec}^2_{\rr{R}}(\n{S}^4,\tau)\;\simeq\;[\hat{\n{S}}^3,{\rm S}\n{U}(2)]_{\rm eq}\;/\;\Z_2\;,\qquad\forall\ m\geqslant 2\;.
\eeq
\begin{remark}\label{rk:deg_eq}{\upshape
Just for convenience of the reader, let us recall some useful well established facts: there is an isomorphism of groups
\beql{eq:homo_deg}
\pi_3\big({\rm S}\n{U}(2)\big)\;\simeq\; [{\n{S}}^3,{\rm S}\n{U}(2)]\;\stackrel{{\rm deg}}{\simeq}\;\Z
\eeq
where the group structure on $[{\n{S}}^3,{\rm S}\n{U}(2)]$ is induced from the pointwise multiplication of maps, while the group structure on 
$\pi_3\big({\rm S}\n{U}(2)\big)$ is the standard one on the homotopy groups. From its very definition  $\pi_3\big({\rm S}\n{U}(2)\big)$ is made by equivalence classes of maps preserving the base points $\ast\in{\n{S}}^3$ and $\n{1}\in {\rm S}\n{U}(2)$. Then, there is a homomorphism  
$\pi_3\big({\rm S}\n{U}(2)\big)\to [{\n{S}}^3,{\rm S}\n{U}(2)]$ forgetting the base points which is actually bijective. Indeed, for any map $\varphi:{\n{S}}^3\to{\rm S}\n{U}(2)$ the map $\varphi'$, defined by $\varphi'(k):=\frac{\varphi(k)}{\varphi(\ast)}$, preserves the base points.
An application of the splitting Lemma provides
$$
[{\n{S}}^3,{\rm S}\n{U}(2)]\;\simeq\;\pi_3\big({\rm S}\n{U}(2)\big)\;\rtimes\;\pi_0\big({\rm S}\n{U}(2)\big)\;\simeq\;\pi_3\big({\rm S}\n{U}(2)\big)
$$
where we used the identification $[\{\ast\},{\rm S}\n{U}(2)]=\pi_0\big({\rm S}\n{U}(2)\big)=0$. The second isomorphism in \eqref{eq:homo_deg} is given by the Brouwer's notion of \emph{degree} for maps ${\n{S}}^3\to {\n{S}}^3$ \cite{bott-tu-82,	hatcher-02} and the (diffeomorphic) identification ${\rm S}\n{U}(2)\simeq {\n{S}}^3$. If $\varphi:{\n{S}}^3\to{\rm S}\n{U}(2)$ is a $C^1$-map we can compute ${\rm deg}(\varphi)$ by the help of the differential geometric formula
\beql{eq:homo_deg01}
{\rm deg}(\varphi)\;=\;\frac{-1}{24\pi^2}\;\int_{\n{S}^3}\Omega_\varphi\;,\qquad\quad \Omega_\varphi:={\rm Tr}_{\C^2}\big[(\varphi^{-1}\dd\varphi)^{\wedge 3}\big]\;.
\eeq
The quantity $\Omega_\varphi$ is  known as \emph{Cartan $3$-form} (see \eg \cite[Chapter 22]{frankel-97}).	 
Let $\bb{E}_\varphi\to {\n{S}}^4$ be the rank 2 complex vector bundle associated with $\varphi\in [{\n{S}}^3,{\rm S}\n{U}(2)]$ under the identification ${\rm Vec}^2_\C({\n{S}}^4)\simeq [{\n{S}}^3,{\rm S}\n{U}(2)]$ (see \cite[Lemma 1.4.9]{atiyah-67}). From Section \ref{ssec:chern_classes}
we know that $\bb{E}_\varphi$ is classified (up to isomorphisms) by the second Chern class $c_2(\bb{E}_\varphi)\in H^4({\n{S}}^4,\Z)$
which can be identified also with a differential form. 
An application of the Stokes theorem provides the equality
\beql{eq:chern_num}
C_2(\bb{E}_\varphi)\;:=\;\int_{\n{S}^4}c_2(\bb{E}_\varphi)\;=\;{\rm deg}(\varphi)
\eeq
between the \emph{second Chern number} of $\bb{E}_\varphi$ and the degree of $\varphi$.
}\hfill $\blacktriangleleft$
\end{remark}

\medskip

We are now in position to prove the following fundamental result:

\begin{proposition}
\label{prop:S_4}
For each $m\geqslant 2$ there exists a bijection
$$
{\rm Vec}^m_{\rr{R}}(\n{S}^4,\tau)\;\simeq\;2\Z
$$
given by the map $C_2\circ\jmath$ where $\jmath$  forgets the $\rr{R}$-structure (\cf Proposition \ref{prop:Rv101}) and 
$C_2:{\rm Vec}^m_{\C}(\n{S}^4)\to\Z$ is the second Chern number.
\end{proposition}
\proof
Because of the bijection \eqref{eq:R_vec_m_4_02} we start with showing that $[\hat{\n{S}}^3,{\rm S}\n{U}(2)]_{\rm eq}\simeq2\Z$. 
Let us consider the sequence of maps
$$
[\hat{\n{S}}^3,{\rm S}\n{U}(2)]_{\rm eq}\;\stackrel{\imath}{\to}\;\pi_3\big({\rm S}\n{U}(2)\big)\;\simeq\;[{\n{S}}^3,{\rm S}\n{U}(2)]\;\stackrel{{\rm deg}}{\simeq}\;\Z
$$
where $\imath$ forgets  the $\Z_2$-action  and the last two isomorphisms have been
discussed in Remark \ref{rk:deg_eq}
In order to study the homomorphism $\imath$ we need to associate to each $\Z_2$-equivariant map $\varphi: \hat{\n{S}}^3\to{\rm S}\n{U}(2)$, the restrictions $\varphi_\pm:\ell^3_\pm\to {\rm S}\n{U}(2)$ on the hemispheres
$$
\ell^3_\pm\;:=\;\{(k_1,\ldots,k_4)\in\hat{\n{S}}^3\ |\ \pm k_4\geqslant 0\}\;.
$$
Since $\ell^3_+\cap \ell^3_-=\hat{\n{S}}^2$ is a $\Z_2$-subcomplex of $\hat{\n{S}}^3$ which has only free cells up to dimension 2 and $\pi_j({\rm S}\n{U}(2))\simeq0$ for $j=0,1,2$ 
one can apply Lemma \ref{lemma:Z2-reduct} which assures that 
$\varphi$ is $\Z_2$-homotopy equivalent to a map such that
$\varphi|_{\ell^3_+\cap \ell^3_-}=\n{1}_2$.
Shrinking the equator $\ell^3_+\cap \ell^3_-$ to a point we can regard $\varphi_+$ and $\varphi_-$ as elements of $\pi_3\big({\rm S}\n{U}(2)\big)$. Moreover, since $\varphi_+=\varphi_-\circ\vartheta$ and the antipodal map  $\vartheta$ on the sphere as well as the complex conjugation on ${\rm S}\n{U}(2)$ preserve the orientations one concludes that
both $\varphi_+$ and $\varphi_-$ are representatives for the same element  $[\psi]\in\pi_3\big({\rm S}\n{U}(2)\big)$.
Then, we can define the homomorphism $\imath$ as follows: $\imath([\varphi]):=[\varphi_-]+[\varphi_+]=2[\psi]$. Clearly,
 the image of $f:={\rm deg}\circ\imath$ is contained in the subgroup $2\Z\subset\Z$ of even integers.

The map $f$ is injective. In fact, let us assume that $f([\varphi])=0$. Since $\pi_3\big({\rm S}\n{U}(2)\big)\simeq\Z$ is torsion free, we get $[\psi]=0$ which means that both $\varphi_+$ and $\varphi_-$ can be homotopy deformed to constant maps. This allows us to construct an equivariant homotopy between $\varphi$ and the constant map with value $\n{1}_2$, \ie $[\varphi]=0$. 
This suffices to prove the injectivity, since $f$ is a homomorphism of groups. 
The injectivity of $f$ implies $[\hat{\n{S}}^3,{\rm S}\n{U}(2)]_{\rm eq}\simeq2n\Z$ for some $n\in \N\cup\{0\}$. 
Since  $\imath$ is the  homomorphism $[\hat{\n{S}}^3,{\rm S}\n{U}(2)]_{\rm eq}\to[{\n{S}}^3,{\rm S}\n{U}(2)]$ which forgets the $\Z_2$-action, in view of  \eqref{eq:chern_num} the map $f$ can be identified with the second Chern number of the complex vector bundle $\jmath(\bb{E}_\varphi)=\bb{E}_{\imath(\varphi)}$ associated with the $\rr{R}$-bundle $\bb{E}_\varphi$. 
This implies $n=1$ since there exists a non-trivial element in ${\rm Vec}^2_{\rr{R}}(\n{S}^4,\tau)$ with second Chern number equals to 2 
  (see Remark \ref{rk:deg2}  for the explicit construction).

To complete the proof, we examine the $\Z_2$-action on $[\hat{\n{S}}^3,{\rm S}\n{U}(2)]_{\rm eq}$. The isomorphism $[\hat{\n{S}}^3,{\rm S}\n{U}(2)]_{\rm eq} \to 2\Z$ is realized as $[\varphi] \mapsto \mathrm{deg} (\varphi)$. The fact that $g \mapsto \epsilon g \epsilon$ is an orientation preserving map on ${\rm S}\n{U}(2)$ implies $\mathrm{deg} (\epsilon \varphi \epsilon) = \mathrm{deg}(\varphi)$. Thus the $\Z_2$-action is trivial, and hence ${\rm Vec}^m_{\rr{R}}(\n{S}^4,\tau)\;\simeq\;2\Z$.
\qed

\medskip

In order to classify \virg{Real} vector bundle over $\tilde{\T}^4\equiv(\T^4,\tau)$ let us introduce the $\tau$-invariant sub-tori $\tilde{\T}^3_i\subset \tilde{\T}^4$, $i=1,2,3,4$ defined by
\beql{eq:T3subT4}
\begin{aligned}
\tilde{\T}^3_1\;&:=\;\{\ast\}\;\times\;\tilde{\n{S}}^1\;\times\;\tilde{\n{S}}^1\;\times\;\tilde{\n{S}}^1&\qquad\tilde{\T}^3_2\;&:=\;\tilde{\n{S}}^1\;\times\;\{\ast\}\;\times\;\tilde{\n{S}}^1\;\times\;\tilde{\n{S}}^1\\
\tilde{\T}^3_3\;&:=\;\tilde{\n{S}}^1\;\times\;\tilde{\n{S}}^1\;\times\;\{\ast\}\;\times\;\tilde{\n{S}}^1&\qquad\tilde{\T}^3_4\;&:=\;\tilde{\n{S}}^1\times\;\tilde{\n{S}}^1\;\times\;\tilde{\n{S}}^1\;\times\;\{\ast\}\;
\end{aligned}
\eeq
where $\{\ast\}\in \tilde{\n{S}}^1$ is one of the two fixed points $(\pm1,0)$. Let $\s{T}_4:=\bigcup_{i=1}^4\tilde{\T}^3_i$.
 In view of the (equivariant) reduced suspension construction \eqref{eq:sphere_susp}, we have $\tilde{\n{S}}^4\simeq \tilde{\n{S}}^1\wedge\tilde{\n{S}}^1\wedge\tilde{\n{S}}^1\wedge\tilde{\n{S}}^1$ (here $\wedge$ denotes the \emph{smash product}) and we can construct  $\tilde{\n{S}}^4$ from $\tilde{\n{T}}^4$ as
$$
\tilde{\n{S}}^4\;\simeq\;\tilde{\n{T}}^4/\s{T}_4
$$
where the notation denotes that the $\Z_2$ sub-complex $\s{T}_4$ is collapsed to  a (fixed) point \cite[Chapter 0]{hatcher-02}. The projection $\upsilon:\tilde{\n{T}}^4\to\tilde{\n{T}}^4/\s{T}_4\simeq \tilde{\n{S}}^4$ is an equivariant map which defines via pullback a map
\beql{eq:mapS4toT4}
\upsilon^\ast\;:\;{\rm Vec}^m_{\rr{R}}(\n{S}^4,\tau)\;\to\; {\rm Vec}^m_{\rr{R}}(\n{T}^4,\tau)\;.
\eeq
The description of ${\rm Vec}^m_{\rr{R}}(\n{T}^4,\tau)$ is a consequence of the fact that  $\upsilon^\ast$ is a bijection. In order to prove this claim we need a preliminary result:

\begin{lemma}
\label{lemm:S4T4}
For all $m\in\N$
\beql{eq:S4toT4_01}
{\rm Vec}^m_{\rr{R}}(\s{T}_4,\tau)\;\simeq\; H^2_{\Z_2}(\s{T}_4,\Z(1))\;\simeq\;0\;.
\eeq
\end{lemma}

Since the proof of this Lemma is  quite technical  we postpone it at the end of this section.

\medskip

\begin{proposition}
\label{prop:T_4}
For each $m\geqslant 2$ there exists a bijection
$$
{\rm Vec}^m_{\rr{R}}(\n{T}^4,\tau)\;\simeq\;2\Z
$$
provided by the map $C_2\circ\jmath$ where $\jmath$  forgets the $\rr{R}$-structure (\cf Proposition \ref{prop:Rv101}) and 
$C_2:{\rm Vec}^m_{\C}(\n{T}^4)\to\Z$ is the second Chern number.
\end{proposition}
\proof
Let us start proving that the 
map \eqref{eq:mapS4toT4} induced by the projection  $\upsilon:\tilde{\n{T}}^4\to\tilde{\n{T}}^4/\s{T}_4$ is surjective. Let $\bb{E}\to \tilde{\T}^4$ be an $\rr{R}$-bundle of rank $m$. Since Lemma \ref{lemm:S4T4} we know that the restriction $\bb{E}|_{\s{T}_4}\to \s{T}_4$
is trivial. Hence, we can choose a trivialization $h:\s{T}_4\times\C^m\to \bb{E}|_{\s{T}_4}$ as  $\rr{R}$-bundles. Then the \virg{Real} generalization of \cite[Lemma 1.4.7]{atiyah-67} provides us a rank $m$ $\rr{R}$-bundle $\bb{E}/\upsilon\to \tilde{\n{T}}^4/\s{T}_4\simeq \tilde{\n{S}}^4$ such that $\upsilon^*(\bb{E}/\upsilon)=\bb{E}$.

To prove the injectivity of $\upsilon^\ast$ let us consider the following diagram
\begin{equation}\label{eq:diag1}
\begin{diagram}
{\rm Vec}^m_{\rr{R}}(\n{S}^4,\tau) &&\rTo^{\upsilon^\ast}{} &&{\rm Vec}^m_{\rr{R}}(\n{T}^4,\tau)\\
 \dTo^{c_2\circ{\jmath}}&&&&\dTo_{c_2\circ{\jmath}}\\
H^4(\n{S}^4,\Z)&&\rTo_{\upsilon^\ast}{} &&H^4(\n{T}^4,\Z)
\end{diagram}
\end{equation}
which is commutative because of the \emph{naturality property} of the Chern classes
$c_j\circ \upsilon^\ast=\upsilon^\ast\circ c_j$. The left vertical arrow is injective as a consequence of 
Proposition \ref{prop:S_4}. The low horizontal homomorphism  $\upsilon^\ast:H^4(\n{S}^4,\Z)\to H^4(\n{T}^4,\Z)$ induced by $\upsilon$
is a bijection since $\s{T}_4$ consists of cells of dimension not greater than 3. Then $(c_2\circ{\jmath})\circ\upsilon^\ast:{\rm Vec}^m_{\rr{R}}(\n{S}^4,\tau)\to H^4(\n{T}^4,\Z)$ has to be injective and this is possible only if the up horizontal map $\upsilon^\ast$ is injective.

The isomorphism \eqref{eq:mapS4toT4} and the commutativity of diagram \eqref{eq:diag1} imply that classes in ${\rm Vec}^m_{\rr{R}}(\n{T}^4,\tau)$ are uniquely specified by the second Chern class of the underlying complex vector bundle. The isomorphism 
$\upsilon^\ast:H^4(\n{S}^4,\Z)\to H^4(\n{T}^4,\Z)$ and Proposition \ref{prop:S_4} imply the isomorphism $C_2\circ\jmath:{\rm Vec}^m_{\rr{R}}(\n{R}^4,\tau)\to 2\Z$.
\qed

\medskip

Before the proof of Lemma \ref{lemm:S4T4} let us recall some basic facts and introduce some notations. First of all from an iterated use of equations \eqref{eq:recurs01} one gets
$$
H^0_{\Z_2}\big(\tilde{\n{T}}^d,\Z(1)\big)\;=\;0\;,\qquad H^1_{\Z_2}\big(\tilde{\n{T}}^d,\Z(1)\big)\;\simeq\;\Z^d\oplus\Z_2\;,\qquad H^2_{\Z_2}\big(\tilde{\n{T}}^d,\Z(1)\big)\;=\;0\;\;.
$$
Note that  the summand $\Z_2$ in $H^1_{\Z_2}$ comes from 
$H^1_{\Z_2}\big(\{\ast\},\Z(1)\big)= H^{1}\big(\R P^\infty,\Z(1)\big)$ generated by the class $\alpha_1$ introduced in \eqref{eq:cohomPR_Z1} and \eqref{eq:KR-coefficZ1}. Under the isomorphism $H^1_{\Z_2}\big(\tilde{\n{T}}^d,\Z(1)\big)\simeq[\tilde{\n{T}}^d,\tilde{\n{S}}^1]_{\rm eq}$ showed in \cite[Proposition A.2]{gomi-13} we have that $\alpha_1\simeq[\epsilon]$
where $\epsilon:\tilde{\n{T}}^d\to \tilde{\n{S}}^1$ is the equivariant map that carries all the points of $\tilde{\n{T}}^d$ to $\{\ast\}\equiv(-1,0)\in\tilde{\n{S}}^1$. The second summand $\Z^d$ is generated by the classes $[\pi_j]$, with $j=1,\ldots,d$, where $\pi_j:\tilde{\n{T}}^d\to \tilde{\n{S}}^1$ are the projections onto the $j$-th component of $\tilde{\n{T}}^d=\tilde{\n{S}}^1\times\ldots\times\tilde{\n{S}}^1$.

\medskip

\medskip

Now, for $q=1,2,3$ let $i_1,\ldots ,i_{4-q}$ be a set of  indices  such that $1\leqslant i_1<\ldots <i_{4-q}\leqslant 4$. Generalizing \eqref{eq:T3subT4} we define the $\binom{4}{4-q}$ sub-tori of $\tilde{\T}^4$ of dimension $d$
$$
\tilde{\T}^q_{i_1,\ldots ,i_{4-q}}\;:=\;\big\{(k_1,k_2,k_3,k_4)\in\tilde{\T}^4\;|\; k_{i_1}=\ldots=k_{i_{4-q}}=\{\ast\}\big\}\;.
$$
We notice that each $\tilde{\T}^q_{i_1,\ldots ,i_{4-q}}$ carries the $\Z_2$-action induced by $\tau$.

\medskip

\proof[Proof of Lemma \ref{lemm:S4T4}]
The space $\s{T}_4$ has the structure of a $3$-dimensional $\Z_2$-CW-complex with fixed $Z_2$-cells only in dimension zero and
 free $\Z_2$-cells in positive dimension. Then we arrive at the first isomorphism in \eqref{eq:S4toT4_01} applying 
first Theorem \ref{theo:stab_ran_R} 
 which gives ${\rm Vec}^m_{\rr{R}}(\s{T}_4,\tau)\simeq{\rm Vec}^1_{\rr{R}}(\s{T}_4,\tau)$ and then
Theorem \ref{theo:clasR-lin}.

\medskip

Let us introduce the $\Z_2$-subcomplex $\s{J}_1=\tilde{\T}^1_{i_1,i_2 ,i_{3}}\cup\tilde{\T}^1_{j_1,j_2 ,j_3}$
where the index  sets $(i_1,i_2 ,i_{3})$ and $(j_1,j_2 ,j_{3})$ differ only by one index. This is equivalent to  $\tilde{\T}^1_{i_1,i_2 ,i_{3}}\cap\tilde{\T}^1_{j_1,j_2 ,j_3}\simeq\{\ast\}$. The identifications $\tilde{\T}^1_{i_1,i_2 ,i_{3}}\simeq\tilde{\T}^1_{j_1,j_2 ,j_3}\simeq\tilde{\n{S}}^1$ 
leads to the the Meyer-Vietoris exact sequence 
$$
\ldots\stackrel{\delta_{k-1}}{\to} H^{k-1}\big(\R P^\infty,\Z(1)\big)\to H^k_{\Z_2}\big(\s{J}_1,\Z(1)\big)\to H^k_{\Z_2}\big(\tilde{\n{S}}^1,\Z(1)\big)^{\oplus 2}\;\stackrel{\delta_k}{\to}\;H^k\big(\R P^\infty,\Z(1)\big)\to\ldots
$$
By construction the map $\delta_1$ sends the class $\big(a[\pi]+a'[\epsilon],b[\pi]+b'[\epsilon]\big)$ in $H^1_{\Z_2}\big(\tilde{\n{S}}^1,\Z(1)\big)^{\oplus 2}$ to the class $(a'-b')[\epsilon]$ in $H^1\big(\R P^\infty,\Z(1)\big)$, hence $\delta_1$ is surjective with kernel described by the condition $a'=b'$. This provides
$$
H^0_{\Z_2}\big(\s{J}_1,\Z(1)\big)\;=\;0\;,\qquad H^1_{\Z_2}\big(\s{J}_1,\Z(1)\big)\;\simeq\;\Z^2\oplus\Z_2\;,\qquad H^2_{\Z_2}\big(\s{J}_1,\Z(1)\big)\;=\;0\;\;.
$$
and $H^1_{\Z_2}\big(\s{J}_1,\Z(1)\big)$ is generated by two distinct copies of $[\pi]$ (one for each different 1-torus in the definition of $\s{J}_1$)  for the summand $\Z^2$ and by $[\epsilon]$ for the summand $\Z_2$.

\medskip

This argument can be generalized  for the case of a complex 
$\tilde{\T}^2_{i_1,i_2}\cup\tilde{\T}^2_{j_1,j_2}$
where the index sets $(i_1,i_2 )$ and $(j_1,j_2)$ differ only by one index or for $\s{J}_3=\tilde{\T}^3_{i}\cup\tilde{\T}^3_{j}$ with $i\neq j$. In the first case we have $\tilde{\T}^2_{i_1,i_2}\cap\tilde{\T}^2_{j_1,j_2}\simeq \tilde{\n{S}}^1$ and in the second 
$\tilde{\T}^3_{i}\cap\tilde{\T}^3_{j}\simeq\tilde{\n{T}}^2$.
An iterated use of the Meyer-Vietoris exact sequence 
and a direct inspection  
of the connecting map $\delta_1$ provide
$$
\begin{aligned}
&H^0_{\Z_2}\big(\s{J}_2,\Z(1)\big)\;=\;0\;,\qquad H^1_{\Z_2}\big(\s{J}_2,\Z(1)\big)\;\simeq\;\Z^3\oplus\Z_2\;,\qquad H^2_{\Z_2}\big(\s{J}_2,\Z(1)\big)\;=\;0\;\;,\\
&H^0_{\Z_2}\big(\s{J}_3,\Z(1)\big)\;=\;0\;,\qquad H^1_{\Z_2}\big(\s{J}_3,\Z(1)\big)\;\simeq\;\Z^4\oplus\Z_2\;,\qquad H^2_{\Z_2}\big(\s{J}_3,\Z(1)\big)\;=\;0\;\;.
\end{aligned}
$$
The $\Z_2$ summand of the non trivial groups $H^1_{\Z_2}\big(\s{J}_2,\Z(1)\big)$ and $H^1_{\Z_2}\big(\s{J}_3,\Z(1)\big)$ is generated by the class $[\epsilon]$.
The torsionless part is generated by a number of distinct copies of $[\pi]$ which equals the number of  different 1-tori contained as subcomplex in $\s{J}_2$ or $\s{J}_3$.

\medskip

Now, we want to apply the  Meyer-Vietoris exact sequence to   $\s{T}_4=\s{J}_{3}\cup\s{J}'_{3}$ with 
$\s{J}_{3}=\tilde{\T}^3_{1}\cup\tilde{\T}^3_{2}$ and $\s{J}'_{3}=\tilde{\T}^3_{3}\cup\tilde{\T}^3_{4}$. For that we need first the 
equivariant cohomology of the intersection 
$\s{Y}_2=\s{J}_{3}\cap\s{J}'_{3}=\s{J}_{2}\cup\s{J}'_{2}$ with $\s{J}_{2}=\tilde{\T}^2_{1,3}\cup\tilde{\T}^2_{1,4}$ and $\s{J}'_{2}=\tilde{\T}^2_{2,3}\cup\tilde{\T}^2_{2,4}$. Since $\s{J}_{2}\cap\s{J}'_{2}=\tilde{\T}^1_{1,2,3}\cup\tilde{\T}^2_{1,2,4}\simeq\s{J}_{1}$ we get
$$
\ldots \stackrel{\delta_{k-1}}{\to} H^{k-1}_{\Z_2}\big( \s{J}_1,\Z(1)\big)\to H^k_{\Z_2}\big(\s{Y}_2,\Z(1)\big)\to H^k_{\Z_2}\big(\s{J}_{2},\Z(1)\big)^{\oplus 2}\;\stackrel{\delta_k}{\to}\;H^k_{\Z_2}\big(\s{J}_1,\Z(1)\big)\to\ldots
$$
which provides, after the usual  analysis of the connecting map $\delta_1$,  
$$
H^0_{\Z_2}\big(\s{Y}_2,\Z(1)\big)\;=\;0\;,\qquad H^1_{\Z_2}\big(\s{Y}_2,\Z(1)\big)\;\simeq\;\Z^4\oplus\Z_2\;,\qquad H^2_{\Z_2}\big(\s{Y}_2,\Z(1)\big)\;=\;0\;\;.
$$
As usual the $\Z_2$ summand of $H^1_{\Z_2}\big(\s{Y}_2,\Z(1)\big)$ is generated $[\epsilon]$ while the torsionless part is generated by the four distinct copies of $[\pi]$ associated with the different 1-tori contained in $\s{Y}_2$  as subcomplex.
 Coming back to the Meyer-Vietoris exact sequence for 
$\s{T}_4$ we finally obtain
$$
\ldots \stackrel{\delta_{k-1}}{\to} H^{k-1}_{\Z_2}\big( \s{Y}_2,\Z(1)\big)\to H^k_{\Z_2}\big(\s{T}_4,\Z(1)\big)\to H^k_{\Z_2}\big(\s{J}_{3},\Z(1)\big)^{\oplus 2}\;\stackrel{\delta_k}{\to}\;H^k_{\Z_2}\big(\s{Y}_2,\Z(1)\big)\to\ldots
$$
and the determination of the groups
$$
H^0_{\Z_2}\big(\s{T}_4,\Z(1)\big)\;=\;0\;,\qquad H^1_{\Z_2}\big(\s{T}_4,\Z(1)\big)\;\simeq\;\Z^4\oplus\Z_2\;,\qquad H^2_{\Z_2}\big(\s{T}_4,\Z(1)\big)\;=\;0\;
$$
can be obtained again by inspecting the  map $\delta_{1}$.\qed

\subsection{\virg{Real} Chern classes}
\label{sect:eq_chern_class}
The theory of characteristic classes for $\rr{R}$-vector bundles is formulated in \cite{kahn-59,krasnov-92,pitsch-scherer-13}.
According to Theorem \ref{theo:honotopy_class}, the involutive space $(G_m(\C^\infty),\varrho)$ 
classifies  ${\rm Vec}^m_\rr{R}(X,\tau)$. To each  involutive space $(X,\tau)$ we can associate the graded cohomology ring
\beql{eq:ring_eq_cohom1}
\s{A}(X,\tau)\;:=\;H^\bullet_{\Z_2}(X,\Z(0))\;\oplus\; H^\bullet_{\Z_2}(X,\Z(1))
\eeq
where the ring structure is given by the cup product. Equation \eqref{eq:ring_eq_cohom1}, restricted to a single point $\{\ast\}$ provides a ring
\beql{eq:ring_eq_cohom1'}
\n{A}:=\s{A}(\{\ast\})\;=\;\s{K}_{\Z}\;\oplus\; \s{K}_{\Z(1)}\;\simeq\; \Z[\alpha]/(2\alpha)
\eeq
where  $\s{K}_{\Z}$ and $\s{K}_{\Z(1)}$ have been introduced in Section \ref{ssec:borel_constr}. As explained in \cite[Proposition 2.4]{gomi-13},   $\alpha$ corresponds to the additive generator  $\alpha_1\in H^1(\R P^\infty,\Z(1))=\Z_2$ (\cf equation \eqref{eq:KR-coefficZ1}) 
and it is subject to the only relation $2\alpha$. Moreover, $\alpha^2$ agrees with the generator $\kappa\in H^2(\R P^\infty,\Z)$ and this justifies the more 	
suggestive
 notation $\alpha=\sqrt{\kappa}$ used in \cite{gomi-13}. The ring $\n{A}$ is crucial in the description of 
$\s{A}(G_m(\C^\infty),\varrho)$, in fact one has \cite[Th\'{e}or\`{e}me 3]{kahn-59}
\begin{equation}\label{eq:equi_univ_chern_class}
\s{A}(G_m(\C^\infty),\varrho)\;=\; \n{A}[\tilde{\rr{c}}_1,\ldots,\tilde{\rr{c}}_m]\;=\; \n{Z}[\alpha,\tilde{\rr{c}}_1,\ldots,\tilde{\rr{c}}_m]/(2\alpha)
\end{equation}
where $\tilde{\rr{c}}_j\in H^{2j}_{\Z_2}(G_m(\C^\infty),\Z(j))$ is called $j$-th \emph{universal  \virg{Real} Chern class}.
The cohomology ring $\s{A}(G_m(\C^\infty),\varrho)$ is the polynomial algebra over the ring $\n{A}$ on $m$ generators $\tilde{\rr{c}}_1,\ldots,\tilde{\rr{c}}_m$ and a comparison with equation \eqref{eq:univ_chern_class} shows that 
equation \eqref{eq:equi_univ_chern_class} provides the equivariant generalization of the cohomology ring $H^\bullet\big(G_m(\C^\infty),\Z\big)$ generated by the usual Chern classes. In fact, the process of forgetting the \virg{Real} structure of the pair $(G_m(\C^\infty),\varrho)$ defines a canonical homomorphism $\jmath:H^{k}_{\Z_2}(G_m(\C^\infty),\Z(j))\to H^{k}(G_m(\C^\infty),\Z)$ such that $\jmath(\tilde{\rr{c}}_j)={\rr{c}}_j$ for all $j=1,\ldots,m$.

\medskip

The universal ring $\s{A}(G_m(\C^\infty),\varrho)$ and the homotopy classification of Theorem \ref{theo:honotopy_class} allow us to define equivariant Chern classes for each element in ${\rm Vec}^m_\rr{R}(X,\tau)$. Let $(\bb{E},\Theta)$ be an $\rr{R}$-bundle over $(X,\tau)$
classified by the equivariant map $\varphi\in[X,G_m(\C^\infty)]_{\rm eq}$, then the 
\emph{$j$-th \virg{Real} Chern class} of $\bb{E}$ is by definition
$$
\tilde{c}_j(\bb{E})\;:=\;\varphi^\ast(\tilde{\rr{c}}_j)\;\in\;H^{2j}_{\Z_2}(X,\Z(j))\qquad\quad j=1,2,3,\ldots\;.
$$
where $\varphi^\ast:H^j_{\Z_2}(G_m(\C^\infty),\Z(i))\to H^j_{\Z_2}(X,\Z(i))$ is the homomorphism of cohomology groups induced by $\varphi$.
\virg{Real} Chern classes (or {$\rr{R}$-Chern classes} in short) verify all the {Hirzebruch axioms} of Chern classes (up to a change in the normalization)
and are uniquely specified by these axioms \cite[Th\'{e}or\`{e}me 2]{kahn-59} or \cite[Theorem 4.2]{pitsch-scherer-13}.
An important property  is that $\tilde{c}_j(\bb{E})=0$ for all $j>m$ if $\bb{E}$ has rank $m$. Finally, under the map $\jmath:{\rm Vec}^m_\rr{R}(X,\tau)\to {\rm Vec}^m_\C(X)$ that forgets the \virg{Real} structure one has the identification $\jmath\tilde{c}_j(\bb{E})=c_j(\jmath(\bb{E}))$. 

\medskip

The notion of $\rr{R}$-Chern classes can be quite useful to explore the structure of certain complex vector bundles.
For a given topological space $X$ let us introduce the following notation
$$
{\rm Ch}^m_{\rm ev}(X)\;:=\;\{\bb{E}\in{\rm Vec}^m_{\C}(X)\ |\ c_{2j+1}(\bb{E})=0\;,\ \ j=0,1,\ldots\}\;.
$$
We notice that the definition  of ${\rm Ch}^m_{\rm ev}(X)$ does not depend on the particular choice of the representative $\bb{E}$ in ${\rm Vec}^m_{\C}(X)$. Combining the notions of $\rr{R}$-Chern classes and the forgetting map $\jmath$, we can give a different (and extended) proof of 
\cite[Theorem 5.4]{denittis-lein-11}.
\begin{proposition}
\label{prop:zero_odd_chern}
Let us consider the torus with {TR-involution} $(\n{T}^d,\tau)$ defined in Definition \ref{def:per_ferm} or the 
 sphere with {TR-involution} $(\n{S}^d,\tau)$ defined in Definition \ref{def:free_ferm}. Then,
$$
\jmath\;:\;{\rm Vec}^m_{\rr{R}}(X,\tau)\;\longrightarrow \; {\rm Ch}^m_{\rm ev}(X)\;,\qquad\quad X\;=\;\n{T}^d,\n{S}^d
$$
where $\jmath$ is the map of Proposition \ref{prop:Rv101} which forgets the \virg{Real} structure. 
\end{proposition}
\proof
Since the map $\jmath$ induces the identification
$\jmath\tilde{c}_j(\bb{E})=c_j(\jmath(\bb{E}))$, it is enough to show  the triviality of the groups  $H^{2(2j+1)}_{\Z_2}(X,\Z(1))=0$. For $X=\n{S}^d$ this follows from the first equation in \eqref{eq:equi_cohom_spher_d_loc} and for $X=\T^d$ this is given by \eqref{eq:even_odd_torus}.
\qed

\medskip

Combining together Corollary \ref{corol_(iii)} with Theorem \ref{theo:clasR-lin} we know that the classification of ${\rm Vec}^m_{\rr{R}}(\n{T}^d,\tau)$
and ${\rm Vec}^m_{\rr{R}}(\n{S}^d,\tau)$ with $d\leqslant 3$
is completely equivalent to the fact that $\tilde{c}_1=0$.
On the other hand, when $d=4$ and $m\geqslant 2$
 the only non trivial
$\rr{R}$-{Chern class} is $\tilde{c}_2$ which is an element of
\beql{eq:cohom_4d}
H^{4}_{\Z_2}(\tilde{\n{S}}^4,\Z)\;\simeq\;\Z_2\;\oplus\;\Z\;,\qquad\quad H^{4}_{\Z_2}(\tilde{\n{T}}^4,\Z)\;\simeq\;\Z_2^{15}\;\oplus\;\Z\;.
\eeq
The first computation follows from
\eqref{eq:equi_cohom_spher_d_int} while the second can be derived from 
$$
H^{4}_{\Z_2}(\tilde{\n{T}}^4,\Z)\;\simeq\; H^{4}_{\Z_2}(\tilde{\n{S}}^1,\Z)\;\oplus\;H^{3}_{\Z_2}\big(\tilde{\n{S}}^1,\Z(1)\big)^{\oplus3}\;\oplus\;H^{2}_{\Z_2}\big(\tilde{\n{S}}^1,\Z\big)^{\oplus 3}\;\oplus\;H^{1}_{\Z_2}\big(\tilde{\n{S}}^1,\Z(1)\big)
$$
which is a consequence of \eqref{eq:recurs01}. A comparison between \eqref{eq:cohom_4d} and the classification given in 
Proposition \ref{prop:S_4} and Proposition \ref{prop:T_4}
shows that the cohomology groups $H^{4}_{\Z_2}(\tilde{\n{S}}^4,\Z)$ and $H^{4}_{\Z_2}(\tilde{\n{T}}^4,\Z)$ are \virg{redundant} for the classification of \virg{Real} vector bundles.
This redundancy is a manifestation of the fact that these cohomology groups contain also information about 
other topological objects like
\emph{$\Z_2$-equivariant} vector bundles  (\cf \cite[Section 1.6]{atiyah-67}). Nevertheless, at least in our case of interest, the second $\rr{R}$-{Chern class} provides a complete classification:

\begin{proposition}
\label{prop:c_2-classific}
Let $(\bb{E}_1,\Theta_1)$ and $(\bb{E}_2,\Theta_2)$ be
two \virg{Real} vector bundles over  $(\n{S}^4,\tau)$ or $(\n{T}^4,\tau)$. Then $(\bb{E}_1,\Theta_1)\simeq(\bb{E}_2,\Theta_2)$
if and only if $\tilde{c}_2(\bb{E}_1)=\tilde{c}_2(\bb{E}_2)$.
\end{proposition}
\proof
If $(\bb{E}_1,\Theta_1)$ are $(\bb{E}_2,\Theta_2)$ equivalent then $\tilde{c}_2(\bb{E}_1)=\tilde{c}_2(\bb{E}_2)$ by construction.
Conversely if 
$\tilde{c}_2(\bb{E}_1)=\tilde{c}_2(\bb{E}_2)$ 
then also ${c}_2(\jmath(\bb{E}_1))={c}_2(\jmath(\bb{E}_2))$ 
where $\jmath$ is the map
  which forgets the \virg{Real} structure and $c_2$ is the second Chern class of the associated complex vector bundle. But, according to 
Proposition \ref{prop:S_4} and Proposition \ref{prop:T_4} this is enough to affirm that $(\bb{E}_1,\Theta_1)$ are $(\bb{E}_2,\Theta_2)$ equivalent. 
\qed

\section{Non-trivial \virg{Real} vector bundles in $d=4$}
\label{sect:non-trivial_ex}
Let $\{\Sigma_0,\ldots,\Sigma_4\}\in{\rm Mat}_4(\C)$ be an irreducible representation of the \emph{Clifford algebra} $C\ell_\C(4)$, namely $\Sigma_i\Sigma_j+\Sigma_j\Sigma_i=2\delta_{i,j}\n{1}_4$. An explicit realization is given by
$$
\begin{aligned}
\Sigma_0\;:=\;\sigma_1\otimes\sigma_3&&\Sigma_1\;:=\;\sigma_2\otimes\sigma_3&&\Sigma_2\;:=\;\n{1}_2\otimes\sigma_1\\
&&\Sigma_3\;:=\;\n{1}_2\otimes\sigma_2&&\Sigma_4\;:=\;\sigma_3\otimes\sigma_3\\
\end{aligned}
$$
where
$$
\sigma_1\;:=\;\left(\begin{array}{cc}0 & 1 \\1 & 0\end{array}\right)\;,\qquad\sigma_2\;:=\;\left(\begin{array}{cc}0 & -\ii \\\ii & 0\end{array}\right)\;, \qquad\sigma_3\;:=\;\left(\begin{array}{cc}1 & 0 \\0 & -1\end{array}\right)
$$
are the Pauli matrices. With this special choice  the following relations hold true:  
$$
\Sigma_j^\ast=\Sigma_j\;,\qquad\quad\overline{\Sigma}_j=(-1)^{j}\;{\Sigma}_j\;,\qquad\quad {\Sigma}_0\;{\Sigma}_1\;{\Sigma}_2\;{\Sigma}_3\;{\Sigma}_4\;=\;-\n{1}_4\;.
$$
We need also the following relations for the trace
\beql{eq:trac_form}
{\rm Tr}_{\C^4}(\Sigma_j)\;=\;0\;,\qquad\quad {\rm Tr}_{\C^4}(\Sigma_i\Sigma_j)\;=\;4\delta_{i,j}\;.
\eeq

\medskip

The Clifford matrices  allow us to  construct \virg{Dirac-like} Hamiltonians on  $L^2(\R^4)\otimes\C^4$ of the type
\beql{eq:hamilt_Sigma}
\hat{H}\;:=\;\sum_{j=0}^4F_j(-\ii\partial_{x_1},\ldots,-\ii\partial_{x_4})\otimes\Sigma_j
\eeq
where the  reality  $F_j:\R^4\to\R$ of these functions assures that $\hat{H}$ is
 self-adjoint. In order to establish the desired  symmetry for system in class {\bf AI}, we also set 
\beql{eq:AppA1}
F_j(-\kappa)=(-1)^{j}\; F_j(\kappa)\;,\qquad\quad \kappa:=(\kappa_1,\ldots,\kappa_4)\in\R^4\;.
\eeq
In fact, if
 $\hat{C}(\psi\otimes{\rm v})=\overline{\psi}\otimes \overline{\rm v}$ is the {complex conjugation} on $L^2(\R^4)\otimes\C^4$, 
one easily verifies from the above definitions that $\hat{H}\;=\;\hat{C}\;\hat{H}\;\hat{C}$. Then, using the jargon introduced in 
Section \ref{sect:bloc_bund}, we can say that
 $\hat{H}$ is an operator in class {\bf AI} with respect to the (trivial) unitary  $\hat{J}=\hat{\n{1}}$. As an example of functions which verify \eqref{eq:AppA1} we can set
\beql{eq:AppA5-0}
\begin{aligned}
F_0(\kappa)\;:=\;\frac{\|\kappa\|^2-1}{\|\kappa\|^2+1}&&F_1(\kappa)\;:=\;\frac{2(\kappa_1+\kappa_2)}{\|\kappa\|^2+1}&&F_2(\kappa)\;:=\;\frac{4(\kappa_1\kappa_2-\kappa_3\kappa_4)}{(\|\kappa\|^2+1)^2}&&\\
                       &&F_3(\kappa)\;:=\;\frac{2(\kappa_3+\kappa_4)}{\|\kappa\|^2+1}&&F_4(\kappa)\;:=\;\frac{4(\kappa_1\kappa_4+\kappa_2\kappa_3)}{(\|\kappa\|^2+1)^2}&&\\
\end{aligned}
\eeq
where $\|\kappa\|^2=\sum_{i=1}^4\kappa_i^2$.

\medskip

A simple computation provides
$$
\hat{H}^2\;=\;{Q}\otimes\n{1}_4\;,\qquad\quad {Q}\;:=\;\sum_{j=0}^4F_j(-\ii\partial_{x_1},\ldots,-\ii\partial_{x_4})^2
$$
and assuming that ${Q}>0$ we can define two spectral projections of $\hat{H}$ by the formula
$$
\hat{P}_\pm\;:=\;\frac{1}{2}\left(\hat{\n{1}}\;\pm\;\sum_{j=0}^4\frac{F_j}{\sqrt{Q}}\;\otimes\;\Sigma_j\right)\;.
$$
After simple algebraic computations one can verify the spectral relations
$$
[\hat{P}_\pm\;\hat{H}]\;=\;0,\qquad\quad\hat{P}_\pm\;\hat{H}\;\hat{P}_\pm\;=\;\pm(\sqrt{Q}\oplus\n{1}_4)\; \hat{P}_\pm\;.
$$
together with the projection properties
 $\hat{P}_\pm^2=\hat{P}_\pm$, $\hat{P}_\pm\hat{P}_\mp=0$, $\hat{P}_+\oplus\hat{P}_-=\hat{\n{1}}$.
 
\medskip

The invariance under translations of $\hat{H}$ and $\hat{P}_\pm$ allows us to apply the Fourier transform which
maps
 $\hat{P}_\pm$  into a family of projection-valued  matrices $\R^4\ni\kappa\mapsto {P}_\pm(\kappa)\in{\rm Mat}_4(\C)$  given by
\beql{eq:proj_sphe1}
{P}_\pm(\kappa)\;:=\;\frac{1}{2}\left(\n{1}_4\;\pm\;\sum_{j=0}^4\frac{F_j(\kappa)}{\sqrt{Q(\kappa)}}\;\otimes\;\Sigma_j\right)\;.
\eeq
 This last formula is well defined since the strictly positivity of the operator ${Q}>0$ is equivalent to
 $Q(\kappa):=\sum_{j=0}^4F_j(\kappa)^2 >0$ for all $\kappa\in\R^4$ (for instance, the ansatz \eqref{eq:AppA5-0} verify this positivity  condition). Since $
{\rm Tr}_{\C^4}{P}_\pm(\kappa)=2$
and $\overline{{P}_\pm(\kappa)}={P}_\pm(-\kappa)$
 the collection of ${P}_\pm(\kappa)$ defines a \virg{Real} vector bundle of rank $2$ over $\R^4$. Assuming the existence of the limits
$$
\qquad\lim_{|\kappa|\to\infty}\frac{F_j(\kappa)}{\sqrt{Q(\kappa)}}\;=\;a_j\;,\qquad\quad \sum_{j=0}^4a^2_j=1	
$$
one can define $P_\pm(\infty):=\lim_{|\kappa|\to\infty}{P}_\pm(\kappa)$ and
 the associated vector bundle can be extended over the one-point compactification $\R^4\cup\{\infty\}\simeq\n{S}^4$ (the ansatz \eqref{eq:AppA5-0} verifies also this limit property with $a_0=1$ and $a_1=\ldots=a_4=0$). This extension can be  realized by means of  the  stereographic coordinates $k:=(k_0,k_1,\ldots,k_4)\in \n{S}^4$ given by the (inverse) transformations
\beql{eq:AppA2}
\kappa_j(k)\;:=\;\frac{k_j}{1-k_0}\;,\qquad\quad j=1,\ldots,4\;.
\eeq
This leads to a map 
of projection-valued matrices $\n{S}^4\ni k\mapsto {P}_\pm(k)\in{\rm Mat}_4(\C)$ where ${P}_\pm(k):={P}_\pm(\kappa(k))$ endowed with the symmetry
$\overline{{P}_\pm(k)}={P}_\pm(\tau(k))$  where $\tau:\n{S}^4\to \n{S}^4$ is the {TR-involution} defined by \eqref{eq:inv_sphe}. In accordance with the discussion in Section \ref{sect:bloc_bund} this construction produces a  \virg{Real} Bloch-bundle
$\pi:\bb{E}_\pm\to\n{S}^4$ of  rank 2 and with fibers given by $\pi^{-1}(k)={\rm Ran}\; P_\pm(k)$.

\medskip

We can write a differential representative for the  second Chern class of $\bb{E}_\pm\to\n{S}^4$ by the formula (\cf \cite[Section 8.3]{{gracia-varilly-figueroa-01}})
$$
c_2(\bb{E}_\pm)\;=\;\frac{1}{8\pi^2}\;{\rm Tr}_{\C^4}\big[{P}_\pm \cdot(\dd{P}_\pm\big)^{\wedge4}\big]
$$
where
$$
\dd{P}_\pm \;:=\;\pm\frac{1}{2}\sum_{j=0}^4\dd\left(\frac{F_j }{\sqrt{Q }}\right)\;\Sigma_j
$$
is the exterior derivative of the function $ k\mapsto {P}_\pm(k)$
and $(\dd{P}_\pm\big)^{\wedge4}:=\dd{P}_\pm \wedge\ldots\wedge\dd{P}_\pm $. Since
$$
\dd\left(\frac{F_i}{\sqrt{Q}}\right)\;\Sigma_i\;\wedge\; \dd\left(\frac{F_j}{\sqrt{Q}}\right)\;\Sigma_j\;=\;\dd\left(\frac{F_j}{\sqrt{Q}}\right)\;\Sigma_j\;\wedge\; \dd\left(\frac{F_i}{\sqrt{Q}}\right)\;\Sigma_i\;,\qquad\quad \left(\dd\left(\frac{F_i}{\sqrt{Q}}\right)\;\Sigma_i\right)^{\wedge2}\;=\;0
$$
 a simple algebraic computation gives
 $$
(\dd{P}_\pm\big)^{\wedge4}\;=\;\frac{3}{2}\;\sum_{j=0}^4 (-1)^{j+1}\left[\dd\left(\frac{F_0}{\sqrt{Q}}\right)\;\wedge\;\ldots\;\wedge\;\underline{\dd\left(\frac{F_j}{\sqrt{Q}}\right)}\;\wedge\;\ldots\;\wedge\;\dd\left(\frac{F_4}{\sqrt{Q}}\right)\right]\;\Sigma_j
$$
where the underlined term in the product is omitted. With the help of  formulas \eqref{eq:trac_form} one finally gets
$$
\begin{aligned}
c_2(\bb{E}_\pm)\;&=\; \frac{\pm1}{{\rm Vol}(\n{S}^4)}\;\sum_{j=0}^4 (-1)^{j+1}\left(\frac{F_j}{\sqrt{Q}}\right)\left[\dd\left(\frac{F_0}{\sqrt{Q}}\right)\;\wedge\;\ldots\;\wedge\;\underline{\dd\left(\frac{F_j}{\sqrt{Q}}\right)}\;\wedge\;\ldots\;\wedge\;\dd\left(\frac{F_4}{\sqrt{Q}}\right)\right]\\
&=\frac{\pm1}{{\rm Vol}(\n{S}^4)}\;Q^{-\frac{5}{2}}\;\;\sum_{j=0}^4 (-1)^{j+1}F_j\;\left[\dd F_0\;\wedge\;\ldots\;\wedge\;\underline{\dd F_j}\;\wedge\;\ldots\;\wedge\;\dd\ F_4\right]
\end{aligned}
$$
where ${\rm Vol}(\n{S}^4)=\frac{8}{3}\pi^2$ is the volume of the $4$-sphere.

\begin{remark}[Hopf bundle]\label{rk:hopf_bundle}{\upshape
Let us denote with $\bb{E}_{\rm Hopf}\to \n{S}^4$ the vector bundle associated with the  projections-valued   map $\n{S}^4\ni k\mapsto {P}_{\rm Hopf}(k)\in{\rm Mat}_4(\C)$ given by
\beql{eq:proj_sphe2}
{P}_{\rm Hopf}(k_0,k_1,\ldots,k_4)\;:=\;\frac{1}{2}\left(\n{1}_4\;+\;\sum_{j=0}^4k_j\;\otimes\;\Sigma_j\right)\;.
\eeq
From the above computation one immediately gets
 $c_2(\bb{E}_{\rm Hopf})=\frac{\omega_{\n{S}^4}}{{\rm Vol}(\n{S}^4)}$ where
$$
\omega_{\n{S}^4}\;:=\;\sum_{j=0}^4 (-1)^{j+1}k_j\left[\dd k_0\;\wedge\;\dd k_1\;\wedge\;\ldots\;\wedge\;\underline{\dd k_j}\;\wedge\;\ldots\;\wedge\;\dd k_4\right]
$$
is the \emph{volume form} of $\n{S}^4$. Under the isomorphism $H^4(\n{S}^4,\Z)\to\Z$ provided by the integration one compute the second Chern number
$$
C_2(\bb{E}_{\rm Hopf})\;:=\;\langle c_2(\bb{E}_{\rm Hopf});\n{S}^4\rangle\;=\;\frac{1}{{\rm Vol}(\n{S}^4)}\int_{\n{S}^4}\omega_{\n{S}^4}\;=\;1\;.
$$
Hence $\bb{E}_{\rm Hopf}$ can be chosen as the non-trivial generator of ${\rm Vec}^2_\C(\n{S}^4)$. Finally, let us point out that $\bb{E}_{\rm Hopf}$ has a real structure with respect to the involution $\varpi:\n{S}^4\to \n{S}^4$ defined in \eqref{eq:AppA3}.
 }\hfill $\blacktriangleleft$
\end{remark}

\medskip

A comparison between \eqref{eq:proj_sphe1}
and \eqref{eq:proj_sphe2} immediately provides the following relationship
$$
\bb{E}_\pm\;\simeq\;\varphi_{F,\pm}^*(\bb{E}_{\rm Hopf})
$$
where the map $\varphi_{F,\pm}:\n{S}^4\to \n{S}^4$ is explicitly given by
\beql{eq:AppA3"'}
\varphi_{F,\pm}\;:\:(k_0,k_1,\ldots,k_4)\;\longmapsto\;\frac{\pm1}{\sqrt{Q(k)}}\big(F_0(k),F_1(k),\ldots,F_4(k)\big)\;.
\eeq
This map is well defined since $\sum_{j=1}^5\left(\frac{F_j(k)}{\sqrt{Q(k)}}\right)^2=1$, for all $k\in\n{S}^4$
and from \eqref{eq:AppA1} and  \eqref{eq:AppA2} one deduces that $\varphi_{F,\pm}$ is subjected to the equivariant relation
$$
\varphi_{F,\pm}\big(\tau(k)\big)\;=\;\varpi\big(\varphi_{F,\pm}(k)\big)
$$
where the two involution $\tau$ and $\varpi$ over $\n{S}^4$ are
\beql{eq:AppA3}
\begin{aligned}
\tau:(k_0,k_1,k_2,k_3,k_4)&\mapsto (k_0,-k_1,-k_2,-k_3,-k_4)\\ 
\varpi:(k_0,k_1,k_2,k_3,k_4)&\mapsto (k_0,-k_1,k_2,-k_3,k_4)\;.
\end{aligned}
\eeq
The construction of non-trivial \virg{Real} vector bundle over $\n{S}^4$ is consequence of the following topological result:
\begin{lemma}
\label{lem:AppA}
Let $[\tilde{\n{S}}^4,\check{\n{S}}^4]_{\rm eq}$ be the set of homotopy equivalence classes of equivariant maps
between the involutive spheres $\tilde{\n{S}}^4\equiv(\n{S}^4,\tau)$ and $\check{\n{S}}^4\equiv(\n{S}^4,\varpi)$.
 Then
$$
[\tilde{\n{S}}^4,\check{\n{S}}^4]_{\rm eq}\;\simeq\;2\Z
$$
and the isomorphism is given by the {topological degree}.
\end{lemma}

We postpone the proof of Lemma \ref{lem:AppA} at the end of this section. For the moment let us argue that if we choose an equivariant map 
$\varphi\in [\tilde{\n{S}}^4,\check{\n{S}}^4]_{\rm eq}$ such that $[\varphi]=2n$ with $n\in\Z$ then $\bb{E}_\varphi:=\varphi^*(\bb{E}_{\rm Hopf})$ provides a non-trivial element of ${\rm Vec}^2_{\rr{R}}(\n{S}^4,\tau)$ characterized by a second Chern number $C_2(\bb{E}_\varphi)=2n$. As  a concrete non-trivial example we can consider a map $\varphi_{F,+}$ of type \eqref{eq:AppA3"'} with
\beql{eq:AppA5}
\begin{aligned}
F_0(k)\;:=\;k_0&&F_1(k)\;:=\;k_1+k_2&&F_2(k)\;:=\;k_1k_2-k_3k_4&&\\
                       &&F_3(k)\;:=\;k_3+k_4&&F_4(k)\;:=\;k_1k_4+k_2k_3&\;.&\\
\end{aligned}
\eeq
These are the functions on the sphere $\n{S}^4$ associated with \eqref{eq:AppA5-0} via the stereographic coordinates \eqref{eq:AppA2}.
With the ansatz \eqref{eq:AppA5} one has $[\varphi_{F,+}]=2$.
We will justify this last claim in Remark \ref{rk:deg2}
just after the details of the proof.

\proof[Proof of Lemma \ref{lem:AppA}]
Let us start with 
$$
[\tilde{\n{S}}^4,\check{\n{S}}^4]_{\rm eq}\;\stackrel{\imath_4}{\longrightarrow}\; [{\n{S}}^4,{\n{S}}^4]\;\simeq\;\pi_4({\n{S}}^4)\;\stackrel{\rm deg}{\simeq}\;\Z
$$
where $\imath_4$ is the map forgetting the involution, the second isomorphism is a consequence of $\pi_1({\n{S}}^4)=\pi_1({\n{S}}^3)=0$
and the third isomorphism is given by the topological degree (\cf Remark \ref{rk:deg_eq}).
From Remark \ref{ex:homot_spher} we know that $\tilde{\n{S}}^4\simeq \s{S}(\hat{\n{S}}^{3})$ where $\s{S}$ is the (unreduced) 
suspension along the fixed direction $k_0$ and $\hat{\n{S}}^{3}$ is the sphere endowed with the antipodal action. In the same way we can write
also $\check{\n{S}}^4\simeq \s{S}(\check{\n{S}}^3)$ where $\check{\n{S}}^3$ is the 3-sphere endowed with the (restricted) involution $\varpi:(k_1,k_2,k_3,k_4)\mapsto (-k_1,k_2,-k_3,k_4)$. Associated with the reduced spheres we have  
$$
[\hat{\n{S}}^3,\check{\n{S}}^3]_{\rm eq}\;\stackrel{\imath_3}{\longrightarrow}\; [{\n{S}}^3,{\n{S}}^3]\;\simeq\;\pi_3({\n{S}}^3)\;\stackrel{\rm deg}{\simeq}\;\Z
$$
where $\imath_3$ is again the map forgetting the involution. Let us consider the following diagram
\begin{equation}\label{eq:diagA1}
\begin{diagram}
[\hat{\n{S}}^3,\check{\n{S}}^3]_{\rm eq} &&\rTo^{\s{S}_\ast^{\rm eq}}{} &&[\tilde{\n{S}}^4,\check{\n{S}}^4]_{\rm eq}\\
 \dTo^{{\imath_3}}&&&&\dTo_{{\imath_4}}\\
[{\n{S}}^3,{\n{S}}^3]&&\rTo_{\s{S}_\ast}{} &&[{\n{S}}^4,{\n{S}}^4]\;.
\end{diagram}
\end{equation}
The map $\s{S}_\ast$ is defined as follows: each $\varphi:X\to Y$ has a trivial extension $\varphi\times{\rm Id}:X\times[-1,1]\to Y\times[-1,1]$.
Since $\s{S}(X)$ is constructed from the cylinder $X\times[-1,1]$ by collapsing $X\times\{-1\}$ and $X\times\{+1\}$ to two distinct points,
the extension $\varphi\times{\rm Id}$ induces a map $\s{S}_\ast\varphi:\s{S}(X)\to \s{S}(Y)$. This construction works equivariantly (if the involutions are extended trivially along the suspension) and naturally, so that the diagram \eqref{eq:diagA1} is commutative.
The map $\s{S}_\ast$ coincides with the suspension isomorphism from $\pi_3({\n{S}}^3)\simeq \Z$ to $\pi_4({\n{S}}^4)\simeq\Z$.
Since the involutive space $\check{\n{S}}^3$ can be identified with ${\rm S}\n{U}(2)$ endowed with the complex conjugation, we know from
Proposition
\ref{prop:S_4} that the map ${\imath_3}$ is a bijection onto the subgroup $2\Z$ consisting of even integers in $\pi_3({\n{S}}^3)\simeq \Z$.
If we prove the injectivity of ${\imath_4}$  we obtain that $\s{S}_\ast^{\rm eq}$ is a bijection which preserves the degree of the maps
and this completes the proof.

Let us recall the $\Z_2$-skeleton decomposition of $\tilde{\n{S}}^4$
described in Example \ref{ex:sphere_TR_CW}. In particular we have
$$
\tilde{\n{S}}^4\;\simeq X^4\;=\;\underbrace{({\bf e}^0_+\cup {\bf e}^0_-)\cup\;\tilde{\bf e}^1\;\cup\;\tilde{\bf e}^2\;\cup\;\tilde{\bf e}^3}_{X^3}\;\cup\;\tilde{\bf e}^4
$$
where ${\bf e}^0_\pm$ are the two fixed $0$-cells, $\tilde{\bf e}^j$
are the free cells of dimension $1\leqslant j\leqslant 4$
and $X^3$  denotes the $3$-skeleton. Moreover, for the fixed point set of 
$\check{\n{S}}^4$ one has $(\check{\n{S}}^4)^\varpi\simeq\n{S}^2$. An applications of Lemma \ref{lemma:Z2-reduct} leads to the following fact:
each $\Z_2$-equivariant map 
$\varphi:\tilde{\n{S}}^4\to\check{\n{S}}^4$
is $\Z_2$-equivariantly homotopic to a $\Z_2$-equivariant map $\varphi':\tilde{\n{S}}^4\to\check{\n{S}}^4$ such that $\varphi'(X^3)=\ast$ 
for a given fixed point $\ast\in(\check{\n{S}}^4)^\varpi$. Identifying $X^3$ with the intersection of two hemispheres
$$
\ell^4_\pm\;:=\;\big\{(k_0,k_1,\ldots,k_4)\in {\n{S}}^4\ |\ \pm k_4\geqslant 0\big\}
$$
we get two classes $[\varphi_\pm]\in\pi_4({\n{S}}^4)$ from the restriction of $\varphi\sim\varphi'$ to $\ell^4_\pm$, respectively.
Since the involutions on $\tilde{\n{S}}^4$ and $\check{\n{S}}^4$ preserve the orientations $[\varphi_+]=[\varphi_-]=[\psi]$ in $\pi_4({\n{S}}^4)$ and so $\imath_4[\varphi]=2[\psi]$.
Since $\pi_4({\n{S}}^4)\simeq\Z$ is torsion free we proved that 
$\imath_4$ is injective (in addition to the fact that the image of $\imath_4$ is contained in $2\Z$
).
\qed

\medskip

\begin{remark}[The degree of $\varphi_{F,+}$]\label{rk:deg2}{\upshape
The identifications
$\tilde{\n{S}}^4\simeq \s{S}(\hat{\n{S}}^{3})$ and  $\check{\n{S}}^4\simeq \s{S}(\check{\n{S}}^3)$ 
where the (unreduced) suspensions are defined along the invariant direction $k_0$ together with the fact that $\s{S}_\ast^{\rm eq}$ is a bijection which preserves the degree allow us to compute the degree of $\varphi_{F,+}$ just looking at its image $\phi_{F,+}\in[\hat{\n{S}}^3,\check{\n{S}}^3]_{\rm eq}$. This map is explicitly given by 
$$
\phi_{F,+}\;:\;(k_1,k_2,k_3,k_4)\;\longmapsto \;\frac{1}{\sqrt{\sum_{j=1}^4F_j(k)^2}}\;\big(F_1(k),\ldots,F_4(k)\big)
$$
and the functions $F_j$ are defined  by \eqref{eq:AppA5} (fixing $k_0=0$). The degree of $\phi_{F,+}$ may be computed by using the formula \eqref{eq:homo_deg01} under the standard identification $\check{\n{S}}^3\simeq{\rm S}\n{U}(2)$ (\eg given by the  Pauli matrices), but another method is applied here: First of all, notice that the degree ${\rm deg} \phi_{F, +}$ agrees with the \virg{cohomological degree} of the pullback $\phi_{F, +}^* : H^3(\n{S}^3) \to H^3(\n{S}^3)$, that is, $\phi_{F, +}^*(1) = \mathrm{deg} \phi_{F, +}$ under the isomorphism of cohomology with real coefficients  $H^3(\n{S}^3)\equiv H^3(\n{S}^3,\R) \simeq \R$ induced by the orientation on $\n{S}^3$. Since the inclusion $\n{S}^3 \to \R^4\backslash \{ 0 \}$ and the normalization $\R^4 \backslash \{ 0 \} \to \R^4 \backslash \{ 0 \}$, ($k \mapsto k/\lvert k \rvert$) are homotopy equivalences, the degree of $\phi_{F, +}$ is the same as that of $F : \R^4 \backslash \{ 0 \} \to \R^4 \backslash \{ 0 \}$. If $\R^4 \cup \{ \infty \}$ denotes the one point compactification of $\R^4$, then there are natural isomorphisms
$$
H^3(\R^4 \backslash \{ 0 \}) \;\simeq\;
H^4(\R^4, \R^4 \backslash \{ 0 \}) \;\simeq\;
H^4(\R^4 \cup \{ \infty \}, (\R^4 \cup \{ \infty \}) \backslash \{ 0 \}) \;\simeq\;
H^4(\R^4 \cup \{ \infty \}, \{ \infty \})\;,
$$
where the first  comes from the exact sequence for the pair $(\R^4, \R^4 \backslash \{ 0 \})$, the second  from the excision axiom, and the third from the homotopy equivalence $(\R^4 \cup \{ \infty \}) \backslash \{ 0 \} \sim \{ \infty \}$. 
From its very definition \eqref{eq:AppA5} one checks that the map $F : \R^4 \to \R^4$ extends continuously to $\tilde{F} : \R^4 \cup \{ \infty \} \to \R^4 \cup \{ \infty \}$ satisfying $\tilde{F}(\infty) = \infty$. Hence $\mathrm{deg} \phi_{F, +} = \mathrm{deg} \tilde{F}$. Further, $F : \R^4 \to \R^4$ is \textit{proper} in the sense that the inverse image of every compact set is compact: Actually, if we introduce the complex coordinates $z = k_1 + i k_3$ and $w = k_2 + i k_4$ identifying $\R^4 \cong \C^2$, then $F$ agrees with $f : \C^2 \to \C^2$ defined by $f(z, w) = (z + w, zw)$. The inverse image of $(\alpha, \beta) \in \C^2$ under $f$ consists of the 
pairs $(z_\pm,\alpha-z_\pm)$ where $z_\pm=z_\pm(\alpha,\beta)$ are
solutions of the quadratic equation $z^2 - \alpha z + \beta = 0$. This allows us to see that the inverse image $f^{-1}(B)$ of a bounded set $B \subset \C^2$ is bounded. Since $f$ is continuous, the inverse image of a closed set is closed, so that $f$ is proper.  The proper map $f : \C^2 \to \C^2$ has its own degree as defined in \cite[Chapter I, Section 4]{bott-tu-82} by means of the de Rham cohomology with \textit{compact support} $H^4_c(\R^4) = H^4_c(\C^2)$. Under the identification $H^4(\R^4 \cup \{ \infty \}, \{ \infty \}) \cong H^4_c(\R^4)$, we find $\mathrm{deg}\tilde{F} = \mathrm{deg} f$. The degree of the proper map $f$ can be computed by counting the inverse image of a \textit{regular value} of $f$ with sign. The Jacobian of $f$ at $(z, w) \in \C^2$ is $\lvert z - w \rvert^2$, so that $(\alpha, \beta)$ is a regular value of $f$ if and only if $\alpha^2 - 4\beta \neq 0$. Its inverse image consists of two distinct points, at which $f$ always preserves the orientation on $\C^2$. This concludes $\mathrm{deg}\phi_{F, +} = \mathrm{deg} f = 2$.
 }\hfill $\blacktriangleleft$
\end{remark}

\begin{remark}\label{rk:other_symm}{\upshape
In the realization of models $\hat{H}$ like \eqref{eq:hamilt_Sigma} it is possible to impose the {\bf AI} symmetry in different ways. For instance the choices $\hat{J}=\n{1}\otimes\Sigma_j$ with $j=0,2,4$ is compatible with $\hat{C}\hat{J}\hat{C}=\hat{J}^*$ and the {\bf AI} symmetry 
of $\hat{H}$ can be imposed  with an appropriate choice of the parity of the functions $F_j$. 
\begin{center}
 \begin{table}[h]
 \begin{tabular}{|c||c|c|c|c|c||c|}
\hline
 \rule[-2mm]{0mm}{7mm}
 $\hat{J}$ & $F_0$  & $F_1$ & $F_2$&$F_3$&$F_4$&\\
\hline
 \hline
 \rule[-2mm]{0mm}{7mm}
 $\n{1}_4$& $+$ & $-$ & $+$ & $-$ &  $+$  &\\
\cline{1-6}
 \rule[-2mm]{0mm}{7mm}
$\Sigma_0$ & $+$ & $+$ & $-$ &$+$& $-$ & {\bf AI}\\
\cline{1-6}
 \rule[-2mm]{0mm}{7mm}
$\Sigma_2$ & $-$ & $+$ & $+$ &$+$& $-$ &\\
\cline{1-6}
 \rule[-2mm]{0mm}{7mm}
$\Sigma_4$ & $-$ & $+$ & $-$ &$+$& + &\\
\hline
 \hline
  \rule[-2mm]{0mm}{7mm}
 $\Sigma_1$ & $-$ & $-$ & $-$ & $+$&   $-$    & {\bf AII}\\
\cline{1-6}
 \rule[-2mm]{0mm}{7mm}
$\Sigma_3$ & $-$ & $+$ &$-$&$-$&
 $-$  &\\
\hline
\end{tabular}\vspace{2mm}
 \caption{The signs in the table provide the parity of the functions $F_j(-\kappa)=\pm\; F_j(\kappa)$
 needed to impose the symmetry of the Hamiltonian $\hat{H}$ in \eqref{eq:hamilt_Sigma}. 
  }
 \end{table}
 \end{center}
 As showed in Table 6.1 the number of even and odd functions is independent by the particular choice of $\hat{J}$ and the number of inequivalent systems is always described by Lemma  
\ref{lem:AppA}. This is in agreement with Theorem \ref{theo:honotopy_class} which states that the classification of inequivalent \virg{Real} vector bundles depends strongly on the involutive structure on the base space and not too much on the type of involution on the total space.
On the other side the choices $\hat{J}=\n{1}\otimes\Sigma_j$ with $j=1,3$ lead to a completely different pattern for the parity of $F_j$'s and so to a  different classification. This fact is not surprising since with these  choices the unitary $\hat{J}$ verifies the symmetry $\hat{C}\hat{J}\hat{C}=-\hat{J}^*$ for type {\bf AII} topological insulators.
}\hfill $\blacktriangleleft$
\end{remark}


\appendix

\section{Spatial parity, equivariant and real vector bundles}
\label{sec:inv_symm}
With the notation introduced in Section \ref{sect:bloc_bund}, let us consider a Hamiltonian $\hat{H}$ with {\bf AI} symmetry \eqref{intro:1}.
If the system is also translationally invariant, the application of the Fourier (or Bloch-Floquet) analysis leads to consider \virg{Real} vector bundles over $(\n{S}^d,\tau)$ or $(\n{T}^d,\tau)$.
In this section we want to consider the effect of an extra symmetry on  type {\bf AI} systems. 

\medskip

Let us introduce the  \emph{spatial parity} operator 
$\hat{I}$ defined by $(\hat{I}\psi)(x)=\psi(-x)$ for vectors $\psi$ in
 $L^2(\R^d,\dd x)\otimes\C^L$ (continuous case)  or in  $\ell^2(\Z^d)\otimes\C^L$ (periodic case).
 Clearly
 $$
(\f{F}\hat{I}\psi)(\kappa)\;=\;(\f{F}\psi)(-\kappa)\;=\;I(\f{F}\psi)(\kappa)
$$
where $\f{F}$ denotes both the Fourier transform or the Bloch-Floquet transform. If the type {\bf AI} Hamiltonian $\hat{H}$ is also \emph{parity-invariant} $\hat{I}\hat{H}\hat{I}=\hat{H}$ (and $[\hat{\Theta};\hat{I}]=0$) the associated fibered projection $\kappa \mapsto P_S(\kappa)$ given by \eqref{eq:proj} 
carries an extra symmetry
$
I P_S(-\kappa)= P_S(\kappa) I
$
where  $I$ is a \emph{linear} map which connects the fiber on $\kappa$ with the fiber on $-\kappa$. Introducing the stereographic coordinates
$k=k(\kappa)$ and the involution $\tau$ one can rewrite the symmetry \eqref{intro:2} as follows
 \beql{intro:2sym_bis}
\Theta\;P_S(k)\; \Theta\; =\; P_S\big(\tau(k)\big)\;= I\;P_S(k)\;I\; \qquad\quad \forall\ \ k\in\n{S}^d,\n{T}^d\;.
\eeq
This means that the Bloch-bundle $\bb{E}_S$ associated with $P_S$ can be endowed with: 1) a \virg{Real} structure
given by $\Theta$; 2) a $\Z_2$-equivariant structure \cite{segal-68} induced by $I$; 3) 
an \emph{anti}-linear map $\Theta':=\Theta\circ I= I\circ \Theta$
which preserves the fibers. In particular the pair $(\bb{E}_S,\Theta')$ can be seen as a \virg{Real} vector bundle  with  a \emph{trivial} involution $\tau'={\rm Id}_X$, hence as a real vector bundle in view of
Proposition \ref{prop:Rv1}.

\medskip

Complex vector bundles $\bb{E}\to X$ over an involutive space $(X,\tau)$ endowed with such a compatible mixed 
\virg{Real} structure (given by an anti-linear homeomorphism $\Theta$) and $\Z_2$-equivariant structure (given by a linear homeomorphism $I$) are called $(\rr{R},\Z_2)$-bundles. Morphisms of $(\rr{R},\Z_2)$-bundles are vector bundle maps which are equivariant at the same time with respect to the \virg{Real} structure and the $\Z_2$-equivariant structure. We use the symbol ${\rm Vec}_{(\rr{R},\Z_2)}^m(X,\tau)$ for the set of equivalence classes of $(\rr{R},\Z_2)$-bundles.

\begin{proposition}
Let us denote with ${\text{\bf AI}_0}$ the class of
{\upshape{\bf AI}} topological quantum systems which are also {parity-invariant}. Systems of type ${ \text{\bf AI}_0}$ are classified by $(\rr{R},\Z_2)$-bundles.
\end{proposition}

The classification of $(\rr{R},\Z_2)$-bundles is more complicated than the classification of complex or \virg{Real} vector bundles. Moreover, albeit similar situations (and even more general) have been already considered from a $K$-theoretic point of view in  \cite{freed-moore-13,shiozaki-sato-14},
a structural analysis of the underlying vector bundle theory seems to be not yet complete in the literature. From its very definition the set ${\rm Vec}_{(\rr{R},\Z_2)}^m(X,\tau)$  carries three natural morphisms
\begin{equation}\label{eq:diagRZ_2}
\begin{diagram}
{\rm Vec}_{\rr{R}}^m(X,\tau)&&\lTo^{f_{\rr{R}}}&&{\rm Vec}_{(\rr{R},\Z_2)}^m(X,\tau)&&\rTo^{f_{\Z_2}}&&{\rm Vec}_{\Z_2}^m(X,\tau)\\
&&&&\dTo^{f_\R}&&&&\\
&&&&{\rm Vec}_{\R}^m(X)&&&&
\end{diagram}
\end{equation}
where the maps $f_{\rr{R}}$ and $f_{\Z_2}$ forget the $\Z_2$-equivariant structure and the $\rr{R}$-structure, respectively. The third map $f_\R$ takes care only of the mixed structure $\Theta'$ and is constructed in Proposition \ref{prop:Rv1}. 

\medskip

Let us consider in details the simple (but interesting) one-dimensional case $(\n{S}^1,\tau)$ leaving to future research
the   analysis of  higher dimensional cases.
 A straightforward generalization of Proposition \ref{prop:stab_ran_R} provides the isomorphism 
${\rm Vec}_{(\rr{R},\Z_2)}^m(\n{S}^1,\tau)\simeq {\rm Vec}_{(\rr{R},\Z_2)}^1(\n{S}^1,\tau)$
and this allows us to specialize diagram \eqref{eq:diagRZ_2} as follows:
\begin{equation}\label{eq:diagRZ_2_bis}
\begin{diagram}
0&&\lTo^{f_{\rr{R}}}&&{\rm Vec}_{(\rr{R},\Z_2)}^m(X,\tau)&&\rTo^{f_{\Z_2}}&&\Z_2\oplus\Z_2\\
&&&&\dTo^{f_\R}&&&&\\
&&&&\Z_2&&&&
\end{diagram}
\end{equation}
where we used the isomorphisms ${\rm Vec}_{\rr{R}}^m(\n{S}^1,\tau)=0$ (Proposition \ref{prop_(i)}), 
${\rm Vec}_{\R}^1(\n{S}^1)\simeq\Z_2$ (Remark \ref{rk:Jmap}) and
${\rm Vec}_{\Z_2}^1(\n{S}^1,\tau)\simeq H_{\Z_2}^2(\tilde{\n{S}}^1,\Z)\simeq\Z_2\oplus\Z_2$ (see \eg \cite[Theorem C.47]{ginzburg-guillemin-karshon-2002} and references therein). Let us consider the two $\Z_2$-equivariant line bundles $\bb{C}_i=\tilde{\n{S}}^1\times \C$, $i=0,1$, with equivariant structure given by $I:(k,z)\mapsto (\tau(k),(-1)^iz)$ and
the two $\Z_2$-equivariant line bundles $\bb{L}_\pm=\tilde{\n{S}}^1\times \C$,  with equivariant structure given by $I:(k,z)\mapsto (\tau(k),\pm(k_0+\ii k_1)z)$. These four $\Z_2$-equivariant line bundles are distinguished by the representations of the  $\Z_2$-actions on the fibers over the fixed points $k_\pm:=(\pm 1,0)$, which turn out to be  (complete) invariants  of $\Z_2$-equivariant line bundles. More precisely one has 
 \begin{center}
 \begin{tabular}{|c||c|c|c|c|}
\hline
&$\bb{C}_0$&$\bb{C}_1$&$\bb{L}_+$&$\bb{L}_-$\\
\hline
$\pi^{-1}(k_+)$&$1$&$\sigma$&$1$&$\sigma$\\
\hline
$\pi^{-1}(k_-)$&$1$&$\sigma$&$\sigma$&$1$\\
\hline
 \end{tabular}
 \end{center}
where $1:\Z_2\to 1$ denotes the trivial representation of $\Z_2$ on $\C$ while $\sigma:\Z_2\to \{\pm 1\}$
is the non-trivial \emph{sign} representation, hence
\beql{eq:Z2-equivVBS1}
{\rm Vec}_{\Z_2}^1(\n{S}^1,\tau)\;=\;\Big\{[\bb{C}_0],[\bb{C}_1],[\bb{L}_+],[\bb{L}_-]\Big\}\;.
\eeq
We notice that $\bb{C}_0$ is the trivial element and $(\bb{L}_-)^*\simeq \bb{L}_+$ are dual to each other. Moreover the  relations $\bb{L}_+\otimes\bb{L}_-\simeq\bb{C}_1$ and $\bb{L}_+\oplus\bb{L}_-\simeq\bb{C}_1\oplus \bb{C}_0$
show that  it is enough to use only $\bb{C}_0,\bb{C}_1$ and $\bb{L}_+$ as
additive generators of the equivariant $K$-theory $K_{\Z_2}^0(\tilde{\n{S}}^1)\simeq\Z^3$
(see \cite[Section 4.2]{gomi-13}). These four line bundles can be also endowed with the (trivial) \virg{Real} structure $\Theta_0:(k,z)\mapsto(\tau(k),\bar{z})$
providing four $(\rr{R},\Z_2)$-bundles of rank 1 over $\tilde{\n{S}}^1$. It is elementary to verify that $f_{\rr{R}}(\bb{C}_i)\simeq f_{\rr{R}}(\bb{L}_\pm)$ are the trivial (and only) element of ${\rm Vec}_{\rr{R}}^1(\n{S}^1,\tau)$. 
\begin{proposition}[Homotopic classification of ${\text{\bf AI}_0}$ topological insulators, $d=1$]
The map $f_{\Z_2}$ in diagram \ref{eq:diagRZ_2_bis} is a bijection, hence ${\rm Vec}_{(\rr{R},\Z_2)}^m(\n{S}^1,\tau)\simeq {\rm Vec}_{(\rr{R},\Z_2)}^1(\n{S}^1,\tau)\simeq\Z_2\oplus\Z_2$ (as a group) with distinguished elements $\bb{C}_0,\bb{C}_1,\bb{L}_+,\bb{L}_-$.
Moreover, the map $f_\R$ is surjective 
with $f_{\R}(\bb{C}_{1,2})\simeq \R\times\C$  the trivial element of ${\rm Vec}_{\R}^1(\n{S}^1)$ and $f_{\R}(\bb{L}_\pm)\simeq \bb{M}$  the M\"{o}bius bundle.
\end{proposition}
\proof[{proof} (sketch of)]
Let $(\bb{E},I,\Theta)$ be any line bundle over $\tilde{\n{S}}^1$ of class ${\text{\bf AI}_0}$. Since  
Proposition \ref{prop_(i)}, without loss of
generality,  we can identify the underlying \virg{Real} line bundle $(\bb{E},\Theta)$ with the product bundle $\tilde{\n{S}}^1\times\C$ endowed with the trivial \virg{Real} structure $(k,z)\mapsto (\tau(k),\bar{z})$. This implies that the $\Z_2$-structure must have the form $I:(k,z)\mapsto (\tau(k),\varphi(k){z})$ with $\varphi:\n{S}^1\to \n{U}(1)$ (we are assuming that an equivariant
Hermitian metric
has been fixed) such that $\varphi\circ{\tau}=\bar{\varphi}$. We notice that the latter condition ensures both the involutive property $I^2={\rm Id}_{\bb{E}}$ and the compatibility condition $I\circ\Theta=\Theta\circ I$.
Let $(\bb{E},I,\Theta)$ and $(\bb{E}',I',\Theta)$ be  two line bundles of class ${\text{\bf AI}_0}$ over 
$\tilde{\n{S}}^1$ identified by the maps $\varphi,\varphi':\n{S}^1\to \n{U}(1)$ respectively. An isomorphism between
$(\bb{E},I,\Theta)$ and $(\bb{E}',I',\Theta)$ is specified by a map $\bb{E}\ni (k,z)\mapsto (k,\psi(k)z)\in\bb{E}'$ such that
$\varphi' =(\psi\circ\tau)\cdot\varphi \cdot\psi^{-1}$ ($\Z_2$-equivariance) and $\psi\circ\tau=\bar{\psi}$
($\rr{R}$-equivariance).

For each $i=0,1$ and $j\in\Z$ let $\varphi_{i,j}:\n{S}^1\to\n{U}(1)$ be the map defined by $\varphi_{i,j}(k_0,k_1):=(-1)^i (k_0+\ii k_1)^j$. Clearly, $\bb{C}_0$ and $\bb{C}_1$  correspond to $\varphi_{0,0}$ and $\varphi_{1,0}$, respectively. Similarly, $\bb{L}_+$ and $\bb{L}_-$ are associated with $\varphi_{0,1}$ and $\varphi_{1,1}$, respectively. 
We observe also that the map $\psi_0:\n{S}^1\to \n{U}(1)$ defined by $\psi_0(k_0,k_1):=k_0+\ii k_1$ provides the isomorphism $(\psi_0\circ\tau)\cdot \varphi_{i,j}\cdot\psi_0^{-1}= \varphi_{i,j-2}$ showing that only the values $j=0,1$ are relevant (up to isomorphism). We can show that for each map $\varphi:\n{S}^1\to \n{U}(1)$ such that $\varphi\circ{\tau}=\bar{\varphi}$ there exists an isomorphism $\psi:\n{S}^1\to \n{U}(1)$
such that $\varphi =(\psi\circ\tau)\cdot\varphi_{i,j} \cdot\psi^{-1}$ for some $i=\pm 1$ and $j\in\Z$. This can be done as follows: let $j:={\rm deg}\varphi$ be the map degree and rewrite $\varphi(k_0,k_1)=C(k_0+\ii k_1)^{j} \expo{\ii 2\pi f(k_0,k_1)}$
for some $C\in\n{U}(1)$ and $f:\n{S}^1\to \R$ such that $f(+1,0)=0$. The condition  $\varphi\circ{\tau}=\bar{\varphi}$
implies $C^2=\expo{-\ii 2\pi [f(\tau(k))+f(k)]}$ for all $k\in \n{S}^1$ and this is possible if and only if $C=\pm 1$ and
$f(\tau(k))=-f(k)$ mod. $\Z$. Then, after setting $\psi(k)=\expo{-\ii \frac{\pi}{2}[ f(k)-f(\tau(k))]}=\expo{-\ii \pi f(k)}$,
we obtain that
$\varphi(k)=\psi(\tau(k))(-1)^i(k_0+\ii k_1)^{j} \psi(k)^{-1}$. This concludes the proof of the bijectivity of $f_{\Z_2}$.

Finally let us recall that  the M\"{o}bius bundle $\bb{M}:=(\n{S}\times\R)/\sim$ is defined by the equivalence relation $(k,r)\sim(-k,-r)$ and the fixed point set $\bb{L}_\pm^{\Theta'}$ is given by $\{(k,z)\in \bb{L}_\pm\ |\ \bar{z}=\pm(k_0+\ii k_1)z\}$. An isomorphism $\bb{M}\simeq \bb{L}_\pm^{\Theta'}$ is established by the map $[(k_0,k_1),r]\mapsto \big((k_0^2+k_1^2,2k_0k_1),r(k_0-\ii k_1)\big)$.
\qed
\begin{remark}\label{rk:redKZ2-VBZ2}{\upshape
A closer look to the $\Z_2$-equivariant $K$-theory of the TR involutive space $\tilde{\n{S}}^1$ shows that 
$K_{\Z_2}^0(\tilde{\n{S}}^1)\simeq\Z^3$ (generated for instance by $\bb{C}_0,\bb{C}_1$ and $\bb{L}_+$) and 
$K_{\Z_2}^0(\{\ast\})\simeq\Z^2$ provided that $\ast$ is one of the two fixed point of $\tilde{\n{S}}^1$ (and generators $\bb{C}_0$ and $\bb{C}_1$) \cite[Section 4.2]{gomi-13}. This implies, for the reduced $K$-group that $\tilde{K}_{\Z_2}^0(\tilde{\n{S}}^1)\simeq\Z$ (generated for instance by the difference $\bb{C}_0-\bb{L}_+$ if $\ast=(+1,0)$).
In particular, a comparison with \eqref{eq:Z2-equivVBS1} this shows that the reduced $\tilde{K}_{\Z_2}^0(\tilde{\n{S}}^1)$ does not describe the set of equivalence classes ${\rm Vec}_{\Z_2}^m(\n{S}^1,\tau)$ of $\Z_2$-equivariant vector bundles.
}\hfill $\blacktriangleleft$
\end{remark}

\section{An overview to $KR$-theory}
\label{sect:KR-theory}

According to  \cite{atiyah-66} we denote with  $KR(X,\tau)$  the \emph{Grothendieck group} 
of $\rr{R}$-vector bundles over the involutive space $(X,\tau)$. Restricting to fixed point set $X^\tau$ (hereafter assumed non empty) one has a homomorphisms
$KR(X,\tau)\to KR(X^\tau,{\rm Id_X})\simeq KO(X^\tau)$, where $KO$ denotes the $K$-theory for real vector bundles. The \emph{reduced group} $\widetilde{KR}(X,\tau)$ is the kernel of the homomorphism $KR(X,\tau)\to KR(\{\ast\})$ where $\ast\in X$ is a $\tau$-invariant base point. 
When $X$ is compact one has the usual relation
\beql{eq:KR0}
KR(X,\tau)\;\simeq\;\widetilde{KR}(X,\tau)\;\oplus\; KR(\{\ast\})\;\simeq\;\widetilde{KR}(X,\tau)\;\oplus\; \Z
\eeq
where we used $KR(\{\ast\})\simeq KO(\{\ast\})\simeq\Z$.
The isomorphism
\beql{eq:KR0X}
\tilde{KR}(X,\tau)\;\simeq\;{\rm Vec}_\rr{R}(X,\tau)\;:=\;\bigcup_{m\in\N}{\rm Vec}^m_\rr{R}(X,\tau)
\eeq
 proved in \cite[Lemma 3.4]{nagata-nishida-toda-82} establishes 
 the fact that $\tilde{KR}(X,\tau)$ provides the description for $\rr{R}$-bundles in the stable regime (\ie when the rank of the fiber is assumed to be sufficiently large).

\medskip

The $KR$-theory can be endowed with a grading structure as follows: first of all one introduces the groups
$$
\begin{aligned}
KR^{j}(X,\tau)\;&:=\;KR(X\times\n{D}^{0,j} ; X\times \n{S}^{0,j} ,\tau\times\vartheta)\\
KR^{-j}(X,\tau)\;&:=\;KR(X\times\n{D}^{j,0} ; X\times \n{S}^{j,0} ,\tau)
\end{aligned}\;\qquad\qquad j=0,1,2,3,\ldots
$$
where $\n{D}^{p,q}$ and $\n{S}^{p,q}$ are the unit ball and unit sphere in the involutive space $\R^{p,q}:=\R^p\oplus{\rm i}\R^q$
introduced in Example \ref{ex:homot_spher}. The relative group $KR(X;Y,\tau)$ of an involutive space $(X,\tau)$ with respect to a $\tau$-invariant subset $Y\subset X$ is defined as $\widetilde{KR}(X/Y,\tau)$ and corresponds to the Grothendieck group of $\rr{R}$-bundles over $X$ which vanish on $Y$. The negative groups $KR^{-j}$ agree with the usual suspension groups since the spaces 
$\n{D}^{j,0}$ and $\n{S}^{j,0}$ are invariant. The positive groups $KR^{j}$ are \virg{twisted} suspension groups since the spaces 
$\n{D}^{0,j}$ and $\n{S}^{0,j}$ are endowed with the  $\Z_2$-action induced by the antipodal map $\vartheta$. With respect to this grading the $KR$ groups are 8-periodic,\ie
$$
KR^{j}(X,\tau)\;\simeq\; KR^{j+8}(X,\tau)\;,\qquad\qquad j\in\Z\;.
$$
Moreover, if $X$ has fixed points one can extend the isomorphism \eqref{eq:KR0} for negative groups:
\beql{eq:KR1}
KR^{-j}(X,\tau)\;\simeq\;\widetilde{KR}^{-j}(X,\tau)\;\oplus\; KR^{-j}(\{\ast\})\;,\qquad\qquad  j=0,1,2,3,\ldots
\eeq
where $KR^{-j}(\{\ast\})\simeq KO^{-j}(\{\ast\})$ for all $j$.

\begin{table}[htp]
 \label{tab:KR1}
 \begin{tabular}{|c||c|c|c|c|c|c|c|c|}
\hline
   & $j=0$ & $j=1$&$j=2$&$j=3$&$j=4$&$j=5$&$j=6$&$j=7$\\
\hline
 \hline
 \rule[-2mm]{0mm}{6mm}
$KR^{-j}(\{\ast\})$&   $\Z$ & $\Z_2$ & $\Z_2$ & $0$ &$\Z$&$0$&$0$&$0$\\
\hline
\end{tabular}\vspace{1mm}
 \caption{
 {\footnotesize 
 The table is calculated using  $KR^{-j}(\{\ast\})\simeq KO^{-j}(\{\ast\})\simeq \widetilde{KO}(\n{S}^j)$ \cite[Theorem 5.19, Chapter III]{karoubi-97}.  The Bott periodicity implies $KR^{-j}(\{\ast\})\simeq KR^{-j-8}(\{\ast\})$.
 }
 }
 \end{table}

In order to compute the $KR$ groups for TR-tori and TR-spheres one starts from the following isomorphism (\cf \cite[eq. 7]{doran-mendez-rosenberg-13})
\beql{eq:KR2}
{KR}^{j-9}(\tilde{\n{S}}^{1})\;\simeq\;KR^{j-9}(\{\ast\})\;\oplus\; KR^{j-8}(\{\ast\})\;.
\eeq
where we used the short notation $\tilde{\n{S}}^{1}\equiv ({\n{S}}^{1},\tau)$.
Using the 8-periodicity one gets
\beql{eq:KR3}
{KR}(\tilde{\n{S}}^{1})\;\simeq\;{KR}^{-8}(\tilde{\n{S}}^{1})\;\simeq\;KR(\{\ast\})\;\oplus\; KR^{-7}(\{\ast\})\;\simeq\;\Z\;
\eeq
which implies $\widetilde{KR}(\tilde{\n{S}}^{1})=0$. The $KR$ groups for $\tilde{\n{T}}^{d}\equiv(\n{T}^d,\tau)$ can be computed with the help of the isomorphism
\beql{eq:KR4}
{KR}^{-j}(\tilde{\n{S}}^{1}\times Y,\tau\times\sigma)\;\simeq\;{KR}^{-(j-1)}(Y,\sigma)\;\oplus\; {KR}^{-j}(Y,\sigma)
\eeq
where $(Y,\sigma)$ is any involutive space \cite[eq 5.18]{hori-99}. Since $\tilde{\n{T}}^{d}=\tilde{\n{S}}^{1}\times\ldots\times\tilde{\n{S}}^{1}$  ($d$-times) 
one gets after repeated iterations
\beql{eq:KR4bis}
{KR}^{-j}(\tilde{\n{T}}^{d})\;\simeq\;\bigoplus_{n=0}^d\Big({KR}^{-(j-n)}(\{\ast\})\Big)^{\oplus \;\binom {d} {n}}\;.
\eeq

\begin{table}[htp]
 \label{tab:KR2}
 \begin{tabular}{|c||c|c|c|c|c|c|c|c|}
\hline
   & $d=1$ & $d=2$&$d=3$&$d=4$&$d=5$&$d=6$&$d=7$&$d=8$\\
\hline
\hline
\rule[-2mm]{0mm}{6mm}
$\widetilde{K}(\n{T}^d)$&   $0$ & $\Z$ & $\Z^3$ & $\Z^7$ &$\Z^{15}$&$\Z^{31}$&$\Z^{63}$&$\Z^{127}$\\

 \hline
 \rule[-2mm]{0mm}{6mm}
$\widetilde{KR}(\tilde{\n{T}}^{d})$&   $0$ & $0$ & $0$ & $\Z$ &$\Z_2^5$&$\Z_2^{16}$&$\Z_2^{43}$&$\Z\oplus \Z_2^{106}$\\
\hline
\rule[-2mm]{0mm}{6mm}
$\widetilde{KO}(\n{T}^d)$&   $\Z_2$ & $\Z_2^{3}$ & $\Z_2^6$ & $\Z\oplus\Z_2^{10}$ &$\Z^5\oplus\Z_2^{15}$&$\Z^{15}\oplus\Z_2^{21}$&$\Z^{35}\oplus\Z_2^{28}$&$\Z^{71}\oplus\Z_2^{36}$\\
\hline
\end{tabular}\vspace{1mm}
 \caption{
 {\footnotesize
 The groups $\widetilde{KR}(\n{T}^d,\tau)$ are obtained from equation \eqref{eq:KR4bis}.
 In the real case a recursive formula can be derived from the isomorphism 
 ${KO}^{-j}(\n{S}^{1}\times Y)\simeq{KO}^{-(j+1)}(Y)\;\oplus\; {KO}^{-j}(Y)$ (\cf with \cite[eq 5.17]{hori-99} when $Y$ is a space with trivial involution).
 In the complex case the groups ${K}(\n{T}^d)$ can be computed from \cite[eq 5.19]{hori-99} and the classification agrees with the 
 description in terms of Chern classes according to \cite{peterson-59}.
}
 }
 \end{table}

The {equivariant} reduced suspension \eqref{eq:susp2} can be equivalently written as
$$
\tilde{\Sigma}X\;:=\;(X\times\n{D}^{0,1})/(X\times \partial\n{S}^{0,1}\;\cup\;\{\ast\}\times\n{D}^{0,1})
$$
where $\{\ast\}$ is a fixed point of $(X,\tau)$. The positively graded groups can be written as
\beql{eq:KR_susp}
\widetilde{KR}(\tilde{\Sigma}X,\tau)\;=\;KR^{1}(X,\{\ast\})\;=\; \widetilde{KR}^1(X,\tau)
\eeq
The reduced $KR$ groups for TR-spheres $\tilde{\n{S}}^d$ can be derived from the \emph{suspension formula} \eqref{eq:KR_susp} together with the isomorphism $\tilde{\n{S}}^{d}\simeq \tilde{\Sigma}\tilde{\n{S}}^{d-1}$ which gives
\beql{eq:KR_susp2}
\widetilde{KR}(\tilde{\n{S}}^{d})\;\simeq\; \widetilde{KR}^{d-1}(\tilde{\n{S}}^1)
\eeq
\begin{table}[htp]
 \label{tab:KR3}
 \begin{tabular}{|c||c|c|c|c|c|c|c|c|}
\hline
   & $d=1$ & $d=2$&$d=3$&$d=4$&$d=5$&$d=6$&$d=7$&$d=8$\\
\hline
\hline
\rule[-2mm]{0mm}{6mm}
$\widetilde{K}(\n{S}^d)$&   $0$ & $\Z$ & $0$ & $\Z$ &$0$&$\Z$&$0$&$\Z$\\
 \hline
 \rule[-2mm]{0mm}{6mm}
$\widetilde{KR}(\tilde{\n{S}}^d)$&   $0$ & $0$ & $0$ & $\Z$ &$0$&$\Z_2$&$\Z_2$&$\Z$\\
\hline
\rule[-2mm]{0mm}{6mm}
$\widetilde{KO}(\n{S}^d)$&   $\Z_2$ & ${\Z_2}$ & $0$ & $\Z$ &$0$&$0$&$0$&$\Z$\\
\hline
\end{tabular}\vspace{1mm}
 \caption{
 {\footnotesize
 The reduced $KR$ groups are computed with the help of  \eqref{eq:KR_susp2} and \eqref{eq:KR2}.
 The reduced $K$ groups for real and complex vector bundles are
 computed in \cite[Chapter 9, Corollary 5.2]{husemoller-94}.
 In the complex case the $\Z$ invariant in even dimensions is  the $d/2$-th Chern class according to 
  \cite{peterson-59}.
}
 }
 \end{table}


\medskip
\medskip

\end{document}